\documentclass[12pt]{amsart}
\usepackage[a4paper,margin=2.5cm]{geometry}
\usepackage[T1]{fontenc}
\usepackage[english]{babel}
\usepackage{float}
\usepackage{microtype}
\usepackage{tikz}
\usepackage{tikz-cd}
\usetikzlibrary{calc}
\usepackage{amsmath,amssymb,amsthm,amscd,mathtools}
\usepackage{mathrsfs}
\usepackage{array}
\usepackage{booktabs}
\usepackage{enumitem}
\usepackage{xcolor}
\usepackage{longtable}
\tikzset{
  partial ellipse/.style args={#1:#2:#3}{
    insert path={+ (#1:#3) arc (#1:#2:#3)}
  }
}
\definecolor{citeblue}{RGB}{0,70,140}
\usepackage[
  colorlinks=true,
  linkcolor=citeblue,
  citecolor=citeblue,
  urlcolor=citeblue
]{hyperref}
\usepackage{etoolbox}

\usepackage{tikz}
\usetikzlibrary{arrows.meta}
\definecolor{defbulk}{RGB}{35,78,112}
\definecolor{defother}{RGB}{41,105,93}
\definecolor{defwall}{RGB}{142,90,34}
\tikzset{
  defbase/.style={font=\small, line cap=round, line join=round},
  defline/.style={draw=defbulk,line width=1.05pt,-{Stealth[length=2mm]}},
  defsecond/.style={draw=defother,line width=1.05pt,-{Stealth[length=2mm]}},
  defedge/.style={draw=defwall,line width=1.1pt},
  defend/.style={circle,fill=defwall,inner sep=0pt,minimum size=4pt},
  defnote/.style={font=\footnotesize,text=black!75},
  defpanel/.style={font=\small\bfseries,anchor=west}
}

\usetikzlibrary{decorations.markings,arrows.meta}

\tikzset{->-/.style={decoration={
  markings,
  mark=at position #1 with {\arrow{>}}},postaction={decorate}}}
\tikzset{-<-/.style={decoration={
  markings,
  mark=at position #1 with {\arrow{<}}},postaction={decorate}}}

\makeatletter
\patchcmd{\@setaddresses}
  {\par\addvspace\bigskipamount\indent}
  {\par\addvspace\bigskipamount\noindent}
  {}{}
\makeatother
\makeatletter
\def\l@subsection{\@tocline{2}{0pt}{3pc}{6pc}{}}
\makeatother
\makeatletter
\renewcommand\section{\@startsection{section}{1}%
  \z@{.7\linespacing\@plus\linespacing}{.5\linespacing}%
  {\normalfont\bfseries\centering}}

\renewcommand\subsection{\@startsection{subsection}{2}%
  \z@{.5\linespacing\@plus.7\linespacing}{-.5em}%
  {\normalfont\bfseries}}
\makeatother

\numberwithin{equation}{section}

\newtheorem{theorem}{Theorem}[section]
\newtheorem{prop}[theorem]{Proposition}

\newtheorem{corollary}[theorem]{Corollary}
\newtheorem{definition}[theorem]{Definition}
\newtheorem*{remark}{Remark}

\theoremstyle{definition}
\newtheorem*{theoremA}{Main Theorem A}
\newtheorem*{theoremB}{Main Theorem B}
\newtheorem*{theoremC}{Main Theorem C}
\newtheorem*{theoremD}{Main Theorem D}
\newtheorem*{theoremE}{Main Theorem E}
\newtheorem*{theoremF}{Main Theorem F}
\newtheorem*{propositionA}{Main Proposition}

\DeclareMathOperator{\Aut}{Aut}
\DeclareMathOperator{\Hom}{Hom}

\DeclareMathOperator{\id}{id}

\DeclareMathOperator{\Cob}{Cob}
\DeclareMathOperator{\Sym}{Sym}

\newcommand{\VecCat}{\mathrm{Vec}}

\title[A Functorial Theory of Defects in Abelian Chern--Simons Theory]
{\fontsize{11pt}{13pt}\selectfont A Functorial Theory of Defects in\\Abelian Chern--Simons Theory}

\author{Daniel Galviz}
\address{\noindent Yau Mathematical Sciences Center and Department of Mathematics, Tsinghua University, Beijing, China.}
\date{}
\email{danielgalviz@tsinghua.edu.cn}

\begin{document}

\begin{abstract}
Recent work has constructed Abelian Chern--Simons theories as categorical TQFTs, allowing us to naturally incorporate categorical defects and construct defect extensions of Abelian Chern--Simons TQFTs. We first identify the Turaev--Viro realizations of Abelian Chern--Simons theory in the center and doubled pointed modular cases, clarifying the distinction between single bulk realizations and canonical doubled ones. Then we obtain Turaev--Viro partition functions in terms of Reidemeister-torsion densities. Alternatively, the Alterfold construction supplies the associated topological boundaries, domain walls, and condensation sectors, establishing an explicit Alterfold/Chern--Simons dictionary. We show that the finite quadratic module is the invariant controlling the bulk theory, its topological symmetries, orientation-reversal invariance, and defects. We further show that multicomponent Abelian BF theory arises as the extended TQFT of an off-diagonal Abelian Chern--Simons theory,  placing it naturally within the same  extended framework. Finally, we demonstrate that recently proposed Abelian Chern--Simons dualities do not define a genuine TQFT duality. These results provide a concrete model for defects in Abelian topological orders and suggest a route toward the non-Abelian case. 
\end{abstract}

\maketitle
{\scriptsize
\setlength{\parskip}{0pt}
\hypersetup{linkcolor=black}
\setcounter{tocdepth}{1}
\tableofcontents
}
\newpage

\section{Introduction}
Chern--Simons theory is one of the standard examples of $(2+1)$-dimensional topological quantum field theories \cite{Witten:1988,Freed:1991,FreedCS,Manoliu3}.  Beyond its role in low-dimensional topology, it also provides the effective topological field theory for many topological phases of matter, including fractional quantum Hall states and other topological orders \cite{wen1992,wenzee1992,Wen2016}.  Physically, boundaries, domain walls, condensation or defects in general are not extra structures, but fundamental features of the same theory \cite{LevinWen05,kitaev2012models,BJQ2013,kong2014anyon}.  

Recent work has constructed Abelian $U(1)^n$ Chern--Simons theories both by geometric quantization in real polarization and by a rigorous functional-integral method, producing unitary extended\footnote{Here \textit{extended TQFT} or just \textit{TQFT} refers to the Turaev--Walker bordism formalism for 2--3 TQFT \cite{Walker, Turaev1994}, in which boundary surfaces carry additional Lagrangian data and gluing is corrected by Maslov--Kashiwara indices.} $(2+1)$-dimensional TQFTs with canonical state spaces, bordism vectors, partition functions, and gluing laws \cite{Galviz2,Galviz2.5}.  The beautiful message of this approach is that, after quantization,  Abelian Chern-Simons theories are identified as categorical TQFTs; this allows us to use all the categorical machinery in our model. These theories were shown to be naturally isomorphic to the Reshetikhin--Turaev theories of the pointed modular categories $\mathcal C(G_K,q_K)$.\footnote{$G_K=\Lambda^*/K\Lambda$ is the finite abelian group obtained from the dual lattice $\Lambda^*$ modulo $K\Lambda$. The lattice $\Lambda$ is even and integral, while $K$ is its nondegenerate bilinear form. The function $q_K([x])=\exp(\pi i x^\top K^{-1}x)$ is the induced quadratic form on $G_K$, and $(G_K,q_K)$ determines the pointed modular category $\mathcal C_K$.
} Hence Abelian Chern--Simons theories are classified, up to symmetric monoidal natural isomorphism, by finite quadratic modules \cite{Galviz1, Galviz3,Galviz4}. 

The aim of the present work is to use this functorial construction to organize several aspects of Abelian Chern--Simons theories that have often been treated separately.  Closed $3$-manifold partition functions and surgery Gauss sums \cite{Witten:1988,Murakami,Mattes,Deloup1999,Deloup2001,BelovMoore,Deloup2005,Stirling}, BF theory \cite{Horowitz89,BlauThompson91,BBRT91}, Abelian Reshetikhin--Turaev invariants \cite{Reshetikhin:1991,Turaev1994,Galviz3}, Turaev--Viro realizations \cite{TV92,BW96,BK2010}, topological symmetries, orientation-reversal invariance \cite{Delmas2021}, and defect structures \cite{LevinWen05,Kapustin2009,kapustin2010surface, Kapustin2010,kitaev2012models, BJQ2013} all have natural Abelian formulations, but they belong to different levels of structure.  The TQFT framework provides the mechanism for bringing these formulations together: it compares not only closed-manifold numerical invariants but also state spaces, bordism maps, and defects within a single functorial structure. Hence we show that the finite quadratic module $(G_K,q_K)$ is the invariant classifying the entire theory:
$$
(G_K,q_K)
\;\longmapsto\;
\mathcal C(G_K,q_K)
\;\longmapsto\;
\left\{
\begin{array}{l}
\text{bulk extended TQFT},\\[2pt]
\text{topological order},\\[2pt]
\text{topological symmetries},\\[2pt]
\text{condensation sectors},\\[2pt]
\text{gapped boundaries},\\[2pt]
\text{topological domain walls},\\[2pt]
\text{defect lines and junctions},\\[2pt]
\text{Morita and Alterfold structures}.
\end{array}
\right.
$$
From the viewpoint of geometric gauge theory, however, defects in Abelian/non-Abelian Chern--Simons theory are difficult to construct directly.  A defect is not simply a hypersurface on which the gauge field is restricted: one must specify which boundary variations are allowed, which bulk Wilson lines may terminate on the defect, how boundary line operators fuse, how walls compose, and how these structures behave under cutting and gluing. Thus, a rigorous treatment of defects in Chern--Simons theories has remained an open problem. Existing works identify individual pieces of the defect structure of Abelian topological phases, while general categorical works provide abstract defect-TQFT machinery without a natural connection to physics. What has been missing is an explicit unified functorial description of Abelian Chern--Simons in which the bulk theory, general defects and their compositions are simultaneously expressed in terms of the finite quadratic module $G_K,q_K$, or equivalently by the pointed modular categories $\mathcal C(G_K,q_K)$.

Kapustin and Saulina gave a foundational analysis of this problem using an Action description \cite{kapustin2010surface,Kapustin2010}.  They identified the finite Abelian group of bulk line operators as the natural discrete invariant of the Abelian theory, related topological boundary conditions to Lagrangian subgroups, described boundary line operators and their associators, and studied surface operators or domain walls through the corresponding categorical structures.  However, this Action formulation does not provide a rigorous functorial construction of defects in the TQFT sense used here: it assumes the path-integral description rather than deriving a well-defined TQFT functor, in particular, the full defect $3$-category is not constructed.  In this sense, the present paper gives a natural functorial realization of the Kapustin--Saulina results, placing the Abelian defect structures in a TQFT setting based on \cite{Fuchs:2013,CRS2019,CMS2020}.

The key mechanism is the construction of Abelian Chern--Simons as a categorical TQFT. It allows us to transfer the study of categorical defects to the pointed modular category $\mathcal C(G_K,q_K)$.  As we will show, the finite quadratic module controls the bulk theory, its relabelling symmetries, and the supports of its boundary and wall structures.  Fully gapped boundaries and walls are described by Lagrangian connected \'etale algebras \cite{DMNO2013}. In the pointed braided realization, such an algebra is supported on a Lagrangian subgroup and its multiplication may be written using a compatible cochain. This cochain is a representative: different compatible choices give isomorphic algebra objects. Thus, up to algebra isomorphism, the Lagrangian subgroup itself classifies the corresponding Lagrangian algebra.

Alterfold theory supplies an alternative realization of the bulk and defects.  It realizes, inside a single TQFT, the Turaev--Viro theory of a spherical fusion category and the Reshetikhin--Turaev theory of its Drinfeld center \cite{Alterfold1, Alterfold2}.  Alterfold theory naturally produces the center sector, the doubled sector, modular invariants, condensations, boundaries, and walls.  When the finite quadratic module $(G_K,q_K)$ is hyperbolic, a single Abelian Chern--Simons theory appears as a center sector $Z^{CS}_{\mathbb T,K}$. In general, the $3$--Alterfold recovers the Abelian Chern--Simons  doubled theory $Z^{CS}_{\mathbb T,K}\otimes Z^{CS}_{\mathbb T,-K}.$

We also clarify the relationship between our framework and BF theories \cite{Horowitz89,BlauThompson91,BBRT91} and reciprocity relations \cite{Murakami,Mattes,Deloup1999,Deloup2005}.  In our functorial framework, multicomponent Abelian BF theory is realized as the TQFT associated with an off-diagonal Abelian Chern--Simons theory.  By contrast, reciprocity formulas for finite Gauss sums compare certain closed-surgery partition functions.  They do not, by themselves, produce equivalences of TQFTs.  In particular, we show that the apparent $K\leftrightarrow L$ dualities suggested by closed-manifold partition-function formulas do not generally survive as functorial dualities as claimed recently in \cite{Kim2024,Tagaris2025}.  Genuine equivalence of Abelian Chern--Simons theories is governed instead by the isomorphism class of the discriminant finite quadratic module.

The result is a unified TQFT treatment of Abelian Chern--Simons theories.  BF theory, reciprocity, Turaev--Viro and Alterfold realizations, symmetries, orientation-reversal invariance, and defects are all organized by the same invariant $(G_K,q_K)$.  This provides a precise framework for defects in Abelian Chern--Simons theories and isolates the categorical structures that should persist, in a more subtle form, in the non-Abelian case.

\subsection*{Main Results}

We summarize the main results. Throughout,
$(\Lambda,K)$ denotes an even, integral, nondegenerate lattice,
$$
G_K:=\Lambda^*/K\Lambda,
\qquad
q_K([x])=\exp(\pi i x^\top K^{-1}x),
$$
and $
\mathcal C_K:=\mathcal C(G_K,q_K)
$ denotes the associated pointed modular category.

\begin{theoremA}
\textbf{Functorial defect extension of Abelian Chern--Simons theory.}
\leavevmode\\

For every even, integral, nondegenerate lattice $(\Lambda,K)$,
the Abelian Chern--Simons TQFT admits a symmetric monoidal defect
extension
$$
\widetilde Z^{CS}_{\mathcal C_K}:
\widehat{\mathrm{Bord}}^{\mathrm{def}}_3(\mathcal D_K)
\longrightarrow
\mathrm{Vect}^{\mathrm{fd}}_{\mathbb C},
$$
whose restriction to undecorated bordisms is naturally isomorphic
to $Z^{CS}_{\mathbb T,K}$. Here $\mathcal D_K$ denotes the
categorical defect datum of $\mathcal C_K$, and the source consists
of decorated stratified surfaces and $3$-bordisms with the
Turaev--Walker anomaly data.
\end{theoremA}

\begin{theoremB}\textbf{Classification of defect structures.}
\leavevmode\\

The defect structures of $\widetilde Z^{CS}_{\mathcal C_K},$ are classified by the corresponding categorical data of the pointed modular category $\mathcal C_K$.

\begin{enumerate}

\item Condensation sectors are classified by connected étale algebras in
$\mathcal C_K$. In the Abelian case, their algebra-isomorphism
classes are classified by isotropic subgroups
$$
H\subset G_K,
\qquad
q_K|_H=1 .
$$

\item Fully gapped topological boundary conditions, up to equivalence of
boundary data, are classified by Lagrangian connected étale algebras in
$\mathcal C_K$. Equivalently, in the Abelian case, they are classified
by Lagrangian subgroups
$$
L\subset G_K,
\qquad
q_K|_L=1,
\qquad
L=L^\perp .
$$

\item Fully gapped topological domain walls between the Abelian theories
associated with $(G_K,q_K)$ and $(G_{K'},q_{K'})$, up to equivalence of wall
data, are classified by Lagrangian connected étale algebras in
$$
\mathcal C_K\boxtimes\mathcal C_{K'}^{\mathrm{rev}} .
$$
Equivalently, in the Abelian case, they are classified by Lagrangian
subgroups
$$
M\subset G_K\oplus G_{K'}
$$
with respect to the quadratic form
$
q_K\oplus q_{K'}^{-1}.
$

\item Every
$
u\in \mathcal O(G_K,q_K)  
$ determines an invertible self-wall supported on the graph
$$
\Gamma_u
=
\{(x,u(x)):x\in G_K\}
\subset
G_K\oplus G_K ,
$$
and the fusion of these walls agrees with the group composition law in
$\mathcal O(G_K,q_K)  $.

\end{enumerate}
\end{theoremB}

\begin{theoremC}\textbf{Turaev--Viro partition functions and Reidemeister torsion.}
\leavevmode\\
Let $(\Lambda,K)$ be an even, integral, nondegenerate lattice, and let $\mathcal C_K=\mathcal C(G_K,q_K)$. For every closed connected oriented $3$-manifold $M$, set
$$
\Lambda_{\mathrm d}:=\Lambda\oplus\Lambda,
\qquad
\mathrm t_{\mathrm d}:=\mathrm t\oplus\mathrm t,
\qquad
K_{\mathrm d}:=K\oplus(-K).
$$
Then, 
$$
Z^{TV}_{\mathcal C_K}(M)
=
\frac{|\det K|^{2m_M}}
{\#\operatorname{Tors}H^2(M;\Lambda_{\mathrm d})}
\sum_{p\in\operatorname{Tors}H^2(M;\Lambda_{\mathrm d})}
\int_{\mathcal M_{M,p}(\mathbb T\times\mathbb T)}
\sigma^{K_{\mathrm d}}_{M,p}
\bigl(T_M(\mathrm t_{\mathrm d})\bigr)^{1/2},
$$
where $m_M=\frac12(b_1(M)-1)$ and $\bigl(T_M(\mathrm t_{\mathrm d})\bigr)^{1/2}$ is the Reidemeister-torsion density. If $\mathcal C_K\simeq\mathcal Z(\mathcal A)$ as braided ribbon categories, the corresponding formula for $Z^{TV}_{\mathcal C_K}(M)$ is the single Chern--Simons expression of Definition~\ref{def:closed-partition}.
\end{theoremC}

\begin{theoremD}\textbf{Alterfold realization of Abelian Chern--Simons defect theories.}
\leavevmode\\
Let $\mathcal C_K=\mathcal C(G_K,q_K)$. If $\mathcal C_K$ is braided equivalent to a Drinfeld center, equivalently if
$(G_K,q_K)$ is hyperbolic, then there exists a spherical fusion category
$\mathcal A$ such that
$
\mathcal Z(\mathcal A)\cong\mathcal C_K,
$
and the Alterfold construction associated with $\mathcal A$ realizes the
Abelian Chern--Simons theory together with its defect data. In particular,
the ordinary $B$-sector satisfies
$$
\mathbb V_{\mathcal A}|_{\mathrm{Cob}^{B}}
\cong
Z^{RT}_{\mathcal Z(\mathcal A)}
\cong
Z^{CS}_{\mathbb T,K},
$$
and the corresponding Alterfold defect sectors realize the condensation
algebras, boundary conditions, domain walls, and Morita data associated with
$\mathcal C_K$.

For a general pointed modular category $\mathcal C_K$, Alterfold theory
canonically realizes the doubled defect theory associated with
$
\mathcal Z(\mathcal C_K)
\cong
\mathcal C_K\boxtimes\mathcal C_K^{\mathrm{rev}},
$ corresponding to the orientation-reversal invariant theory
$$
Z^{CS}_{\mathbb T,K}\otimes Z^{CS}_{\mathbb T,-K}.
$$

More generally, isotropic subgroups
$
H\subset G_K
$ determine connected étale condensation algebras $A_H$ up to algebra
isomorphism. After choosing a pointed braided realization, a compatible
multiplication cochain gives an explicit representative of $A_H$. The
associated Alterfold construction produces the corresponding spherical
Morita context, modular invariant, $\alpha$-induction data, and topological
full center. The Lagrangian case gives the fully gapped boundary sector.
\end{theoremD}

\begin{theoremE}\textbf{Symmetries, and orientation-reversal invariance.}
\leavevmode\\
For the Abelian Chern--Simons theory classified by $(G_K,q_K)$, the group of
topological relabelling symmetries is
$$
\operatorname{Sym}(Z^{CS}_{\mathbb T,K})
\cong
\mathcal O(G_K,q_K)  .
$$
The orientation-reversed theory is classified by $(G_K,q_K^{-1}),$ equivalently by the lattice $(\Lambda,-K)$. Hence $Z^{CS}_{\mathbb T,K}$ is orientation-reversal invariant if and only if $ (G_K,q_K)\cong (G_K,q_K^{-1}). $
\end{theoremE}

\begin{theoremF}
\textbf{Reciprocity is not a functorial duality.}
\leavevmode\\
The quadratic reciprocity relations relating Gauss sums do not by themselves define equivalences of extended Abelian Chern--Simons theories. More precisely, the apparent horizontal $K\leftrightarrow L$ duality suggested by closed-scalar reciprocity formulas does not agree with the partition functions of the extended theories and does not determine a symmetric monoidal natural isomorphism of extended TQFTs. Genuine equivalence of Abelian Chern--Simons theories is instead governed by the isomorphism class of the finite quadratic module:
$$
Z^{CS}_{\mathbb T,K}\cong Z^{CS}_{\mathbb T,L}
\quad\Longleftrightarrow\quad
(G_K,q_K)\cong (G_L,q_L).
$$
\end{theoremF}

\begin{propositionA}\textbf{Explicit gauge-theoretic defects.}

A further contribution of this work is the explicit realization of defect data in Abelian Chern--Simons theory. While boundaries and defects admit an abstract categorical description, we derive formulas that express their higher associativity and fusion structures directly in terms of bulk Wilson-line data, screening charges, and condensate endpoint operators.

For a Lagrangian boundary $L\subset G_K$, choose a normalized section
$$
s:G_K/L\longrightarrow G_K,
\qquad
\lambda(a,b)=s(a)+s(b)-s(a+b)\in L .
$$
Writing
$$
h=\lambda(a,b),\qquad
k=\lambda(b,c),\qquad
p=\lambda(a+b,c),\qquad
r=\lambda(a,b+c),
$$
the boundary associator is represented, with the tensor conventions used
below, by
$$
\begin{aligned}
\beta_{L,\mu}(a,b,c)
={}&
\frac{
\omega_K(s(a),s(b),s(c))\,
\omega_K(s(a+b),s(c),h)
}{
\omega_K(s(a),s(b+c),k)\,
\omega_K(s(a+b),h,s(c))
}
\\
&\times
\frac{
\omega_K(s(a+b+c),r,k)
}{
\omega_K(s(a+b+c),p,h)
}
\frac{
c_K(s(c),h)\,\mu(r,k)
}{
\mu(p,h)
}.
\end{aligned}
$$
The first two factors record bulk associator contributions from reassociating the chosen representatives and screening charges. The third combines the braiding phase $c_K(s(c),h)$ with the relative multiplication amplitudes $\mu(r,k)/\mu(p,h)$.  Changing the section, fusion bases, or algebra representative modifies $\beta_{L,\mu}$ only by a coboundary, leaving the class $[\beta_L]$ invariant.

The same viewpoint gives a gauge-theoretic interpretation of condensation multiplication. For a condensed sector $h\in L$, endpoint operators satisfy
$$
\varepsilon_{h_1}\varepsilon_{h_2}
=
\mu(h_1,h_2)\varepsilon_{h_1+h_2},
$$
with consistency conditions
$$
\mu(h_1,h_2)\mu(h_1+h_2,h_3)
=
\omega_K(h_1,h_2,h_3)
\mu(h_2,h_3)\mu(h_1,h_2+h_3),
$$
$$
\mu(h_1,h_2)
=
c_K(h_1,h_2)\mu(h_2,h_1).
$$
Thus the algebra cochain is interpreted as the fusion amplitude of Wilson lines terminating on the defect.

For domain walls, the corresponding charge-support composition is captured by
$$
M_{13}^{\mathrm{rel}}
=
\{(x,z)\in G_1\oplus G_3:
\exists\,y\in G_2,\,
(x,y)\in M_{12},\,
(y,z)\in M_{23}\},
$$
while the full defect composition is given categorically by the relative tensor product. This distinction makes explicit which information is retained by charge-level composition and which requires the full defect data.

\subsection*{Structure of the paper.} The paper is organized as follows. In Section~\ref{sec:Abelian-CS}, we recall the Abelian Chern--Simons TQFT formalism, its equivalence with Reshetikhin--Turaev theory and its classification theorem. In Section~\ref{sec:bf}, we study the BF embedding into our formalism and describe the resulting BF TQFT. In Section~\ref{sec:reciprocity}, we study reciprocity relations in the TQFT framework and explain why the so-called $K\leftrightarrow L$ Chern--Simons duality is not a genuine TQFT duality. In Section~\ref{sec:symmetries}, we study the topological symmetry group, orientation-reversed duality for Abelian Chern--Simons theories. In Section~\ref{sec:tv}, we identify the Turaev–Viro TQFT with Abelian Chern--Simons theory in both the center and the doubled cases. We also express the corresponding Turaev--Viro partition functions in terms of Reidemeister torsion. In Section \ref{sec:defects}, we construct the defect extension of Abelian Chern--Simons theory and classify its defects. In Section \ref{sec:Alterfold}, we use Alterfold theory to obtain a defect extension of Abelian Chern--Simons theories and explain how the same mechanism points toward the non-Abelian case. Finally, in Section \ref{sec:gauge-defect}, we give a gauge-theoretic interpretation of the categorical defect results and apply our unified framework to the Toric code and semion as an examples.
\end{propositionA}

\textbf{Acknowledgements.} I would like to thank Nicolai Reshetikhin for his support and suggestions throughout this work, and Zhengwei Liu and Yilong Wang for helpful conversations and suggestions concerning Alterfold theory, as well as for their comments on the manuscript.

\section{Extended Abelian Chern--Simons Theory}
\label{sec:Abelian-CS}
Abelian Chern--Simons theory is a topological gauge theory whose classical solutions are flat connections \cite{Witten:1988,FreedCS,FreedCS2}. Its quantum state spaces arise by geometric quantization of the global holonomy degrees of freedom \cite{Manoliu2}. We use quantization in real polarization for $U(1)^n$ gauge fields to obtain a topological quantum field theory \cite{Galviz2}.

Let $(\Lambda,K)$ be an even, integral, nondegenerate lattice of rank $n$. The gauge group and its Lie algebra are
$$
\mathbb T:=\mathrm t/\Lambda,
\qquad
\mathrm t:=\Lambda\otimes\mathbb R.
$$
The form $K$ specifies the coupling between the Abelian gauge fields. Using the normalization $c(P)=[F_A/(2\pi)]$, the local classical Lagrangian and action on a trivial bundle are
$$
\mathscr L_K(A)=\frac{1}{4\pi}K(A\wedge dA),
\qquad
S_X(A)=\int_X\mathscr L_K(A).
$$

The quantum phase is $e^{iS_X(A)}$. On a closed manifold the action is defined modulo $2\pi$, including when the bundle is nontrivial. Extending the bundle and connection to an oriented four-manifold $W$ with $\partial W=X$ \cite{FreedCS2}, this phase can be written as

$$
e^{iS_X(A)}
=
\exp\!\left(
\frac{i}{4\pi}\int_W
K(F_{\widetilde A}\wedge F_{\widetilde A})
\right).
$$

Evenness and integrality of $K$ ensure independence of the extension, without requiring a spin structure. Since the action involves no metric, the classical theory is topological. Its first variation is
$$
\delta S_X(A)
=
\frac{1}{2\pi}\int_X K(\delta A\wedge F_A)
-\frac{1}{4\pi}\int_{\partial X}K(A\wedge\delta A).
$$
For variations supported in the interior, nondegeneracy of $K$ therefore gives $F_A=0.$ Thus classical solutions are flat connections. Locally they are pure gauge, but globally they can retain nontrivial holonomy. These global degrees of freedom are what we must quantize.

On the spacetime cylinder $I\times\Sigma$, the temporal component of the connection imposes spatial flatness. Writing $a:=A_\Sigma/(2\pi)$, gauge transformations identify closed forms modulo exact forms and integral periods, giving the reduced phase space
$$
\mathcal M_\Sigma(\mathbb T)
:=
H^1(\Sigma;\mathrm t)/H^1(\Sigma;\Lambda).
$$
 Thus the classical theory associates a finite-dimensional symplectic torus to each closed oriented surface. The boundary term in the variation of the action determines the symplectic pairing. In the integral normalization used throughout our framework, it is

$$
\omega_{\Sigma,K}([\alpha],[\beta])
=
\int_\Sigma K(\alpha\wedge\beta),
$$

where $\alpha,\beta$ represent variations of the normalized field $a$.  For a compact oriented three-manifold $X$ with $\partial X=\Sigma$, the continuous variations of extendable flat boundary fields form the Lagrangian subspace
$$
L_X:=\operatorname{Im}(r_X)\subset H^1(\Sigma;\mathrm t),
$$
where $r_X:H^1(X;\mathrm t)\to H^1(\Sigma;\mathrm t)$ is restriction. Under Poincar\'e duality, $\lambda_X:=\ker\!\bigl(H_1(\Sigma;\mathbb R)\to H_1(X;\mathbb R)\bigr)$ corresponds to

$$
L_X^{\mathbb R}
:=
\operatorname{Im}\!\bigl(
H^1(X;\mathbb R)\to H^1(\Sigma;\mathbb R)
\bigr),
\qquad
L_X=L_X^{\mathbb R}\otimes\mathrm t.
$$
The action also determines the phase information needed for quantization. On a manifold with boundary, its exponential naturally takes values in a complex line over the boundary field. This is why quantum wavefunctions are sections of a line bundle: their phases transform with the boundary gauge data. The canonical Abelian Chern--Simons line construction descends to a Hermitian prequantum line bundle $\mathcal L_{\Sigma,K}\rightarrow  \mathcal M_\Sigma(\mathbb T)$ with unitary connection $\nabla_{\Sigma,K}$ satisfying
$$
F_{\nabla_{\Sigma,K}}=-2\pi i\,\omega_{\Sigma,K}.
$$

This curvature expresses how the phase of a state changes when transported around a small loop in phase space: it measures the enclosed symplectic area.  A polarization selects a representation in one set of conjugate variables, analogous to choosing position or momentum wavefunctions. Here we choose a rational Lagrangian $L\subset H^1(\Sigma;\mathbb R)$, determining the translation-invariant real polarization $P_L:=L\otimes \mathrm t\subset H^1(\Sigma;\mathrm t).$ Rationality ensures that its leaves are compact tori. Physically, the transverse coordinates specify a commuting set of holonomies; motion along a leaf changes their conjugate variables.

The polarization condition requires the prequantum part of a state to be parallel in these leaf directions:
$$
\nabla_v s=0,\qquad v\in P_L.
$$
This condition must be consistent both locally and globally. Locally, it is integrable because the symplectic form vanishes along a Lagrangian leaf, so the restricted connection is flat. Globally, transporting a candidate state around a closed cycle of the leaf must return the same vector in the same fiber. Hence
$$
s(x)=\operatorname{Hol}_{\nabla}(\gamma)\,s(x)
\qquad
\text{for every loop }\gamma\subset\ell
\text{ based at }x.
$$

Nontrivial holonomy around any loop forces a parallel section to vanish on the connected leaf. Conversely, trivial holonomy allows any initial fiber value to extend uniquely to a global parallel section. This is the Bohr--Sommerfeld condition: the global phase-consistency requirement for a polarized quantum state \cite[Sec.~5.1]{Manoliu2}.

We can see explicitly why this selects finitely many states. Choose symplectic lattice coordinates adapted to $L$, writing a phase-space point as $(x_1,\ldots,x_g;y_1,\ldots,y_g)$, with the $x_j$ varying along a leaf and the $y_j$ fixed. With compatible cycle orientations, the holonomy around the leaf cycle
specified by $\lambda=(\lambda_1,\ldots,\lambda_g)\in\Lambda^g$ is
given by \cite{Galviz2}
$$
\operatorname{Hol}_{\nabla}(\gamma_\lambda)
=\exp\!\left(2\pi i\sum_{j=1}^g K(y_j,\lambda_j)\right).
$$
Reversing a cycle inverts this phase and leaves the quantization condition unchanged. Requiring it to equal one for every $\lambda$ gives
$$
K(y_j,\lambda_j)\in\mathbb Z\quad\text{for every }\lambda_j\in\Lambda,
\qquad
Ky_j\in\Lambda^*,
\qquad
y_j\in K^{-1}\Lambda^*, 
$$
with $\Lambda^*:=\operatorname{Hom}(\Lambda,\mathbb Z)$. Large gauge transformations already identify $y_j$ modulo $\Lambda$, so the surviving transverse labels are
$$
K^{-1}\Lambda^*/\Lambda\xrightarrow{\;K\;}\Lambda^*/K\Lambda=G_K,
\qquad
\mathcal{BS}(\Sigma_g,L)\simeq G_K^g.
$$
The final identification uses a choice of origin leaf. Intrinsically the allowed leaves form a torsor for $G_K^g$. Thus compactness of the holonomy variables, their integral symplectic coupling, and the polarization condition together produce a discrete quantum state space. Each handle supplies one label in $G_K$, rather than an independent label for both conjugate holonomies.

On an allowed leaf the initial fiber value determines the parallel section, giving one complex degree of freedom. The half-density factor supplies the measure needed to pair states and normalize their inner product. In the translation-invariant toral setting it has a canonical parallel trivialization, so it does not change the Bohr--Sommerfeld selection rule \cite[Propositions~2.13--2.16]{Galviz2}. One may regard these states as distributions supported on the allowed leaves; ordinary smooth sections on the whole phase space would not implement this real-polarization quantization. The finite-dimensionality therefore concerns the topological quantum state space on a closed surface, not a discretization of the original space of gauge fields.

\begin{prop}\cite[Proposition 2.14, Definition 2.15, Theorem 2.17]{Galviz2}
\label{prop:toral-hilbert-space}
Let $\Sigma_g$ be a connected closed oriented surface of genus $g$, and let
$L\subset H^1(\Sigma_g;\mathbb R)$ be a rational Lagrangian. Then $\mathcal{BS}(\Sigma_g,L)$ is a torsor for $G_K^g$, and
$$
\mathcal H_{\mathbb T,K}(\Sigma_g,L)
=
\bigoplus_{\ell\in \mathcal{BS}(\Sigma_g,L)}
\Gamma_{\mathrm{flat}}\!\bigl(\ell;\,\mathcal L_{\Sigma_g,K}\otimes |\det P_L^*|^{1/2}\bigr),
$$
$$
\dim \mathcal H_{\mathbb T,K}(\Sigma_g,L)=|G_K|^g=|\det K|^g.
$$
In particular, an origin leaf labels the one-dimensional summands by $G_K^g$; choosing a unit vector in each summand gives a basis with these labels. The phases of these basis vectors require a choice beyond the origin leaf.
\end{prop}

\begin{proof}
This is precisely the real-polarization quantization framework established in \cite{Galviz2}.  The Bohr--Sommerfeld leaves are parametrized by $(\Lambda^*/K\Lambda)^g\cong G_K^g$, and each leaf contributes a one-dimensional space of covariantly constant sections twisted by the canonical half-density. The dimension counts the independent quantum states of the global
holonomy degrees of freedom on $\Sigma_g$.
\end{proof}

\begin{remark}
Manoliu's $U(1)$ theory is the rank-one positive even-level case of this framework. Taking
$$
\Lambda=\mathbb Z,\qquad K=[k],\qquad k\in 2\mathbb Z_{>0},
\qquad \mathbb T=\mathbb R/\mathbb Z\cong U(1),
$$
one obtains
$$
G_K\cong\mathbb Z_k,\qquad
y_j\in\tfrac1k\mathbb Z/\mathbb Z,\qquad
\dim\mathcal H_{U(1),k}(\Sigma_g,L)=k^g.
$$
Thus Manoliu's theory is recovered for $K=[k]$, while general $K$ describes coupled $U(1)$ fields \cite[Sec.~5.1]{Manoliu2}\cite[Corollary~4.9]{Galviz2}.
\end{remark}

The choice of polarization specifies how we represent a state, so changing it should amount to changing quantum representation. The analogue of the position--momentum Fourier transform is supplied by the canonical
Blattner--Kostant--Sternberg (BKS) operators
$$
F_{L_2L_1}:\mathcal H_{\mathbb T,K}(\Sigma,L_1)\xrightarrow{\sim}\mathcal H_{\mathbb T,K}(\Sigma,L_2),
$$
whose composition law is
$$
F_{L_3L_2}\circ F_{L_2L_1}
=
e^{\frac{\pi i}{4}\mu_K(L_1,L_2,L_3)}\,F_{L_3L_1},
\qquad
\mu_K(L_1,L_2,L_3)=\sigma(K)\,\mu_\Sigma(L_1,L_2,L_3).
$$
Here $\sigma(K)$ is the signature of the coupling form. The phase in the composition law means that changing representations is initially projective: a sequence of such changes can return the state multiplied by a phase. This is the $K$-twisted Maslov anomaly, absorbed by the additional structures in the extended theory.

A three-manifold with boundary prepares a quantum state on its boundary. Its classical solutions specify which boundary holonomies can occur; their Chern--Simons phases and fluctuation measures specify their amplitudes. Let $X$ be a compact oriented $3$-manifold. Flat torus bundles have torsion characteristic class, since their real curvature class vanishes. Accordingly, the bulk contributions are organized by torsion sectors.
For each torsion class $p\in \operatorname{Tors}H^2(X;\Lambda),$ let $\mathcal M_{X,p}(\mathbb T)$ denote the corresponding moduli space of flat $\mathbb T$-fields.
Its boundary image
$$
\Lambda_{X,p}:=\operatorname{Im}\!\bigl(\mathcal M_{X,p}(\mathbb T)\to \mathcal M_{\partial X}(\mathbb T)\bigr)
$$
is a translated Lagrangian leaf, parallel to $L_X$. The translation records the topological sector, while the tangent directions record the continuous flat fields that extend through the bulk. The phase and measure entering its quantum contribution are, respectively:

\begin{enumerate}[label=\rm(\arabic*)]
\item a covariantly constant Chern--Simons section
$$
\sigma_{X,p}:\mathcal M_{X,p}(\mathbb T)\longrightarrow r_{X,p}^*\,\mathcal L_{\partial X,K},
$$
\item a canonical translation-invariant Reidemeister-torsion half-density
$$
\mu_{X,p}\in
\Gamma\!\bigl(\Lambda_{X,p},|\det(T^*\Lambda_{X,p})|^{1/2}\bigr).
$$
\end{enumerate}

The section carries the classical action phase. The torsion half-density supplies the measure on the remaining flat modes; in the Gaussian functional-integral description it reflects the contribution of fluctuations after gauge fixing. The normalization keeps track of zero modes and residual gauge symmetry through the exponent
$$
m_X:=
\frac14\Bigl(
\dim H^1(X;\mathbb R)+\dim H^1(X,\partial X;\mathbb R)
-\dim H^0(X;\mathbb R)-\dim H^0(X,\partial X;\mathbb R)
\Bigr),
$$
which reduces, for closed connected $M$, to $m_M=\frac12\bigl(b_1(M)-1\bigr)$.

\begin{prop}\cite[Proposition 3.12]{Galviz2}
\label{thm:boundary-vector}
Let $X$ be a compact oriented $3$-manifold with boundary. Then the canonical boundary vector is
$$
Z^{CS}_{\mathbb T,K}(X):=
|\det K|^{m_X}\,
\frac{1}{\#\operatorname{Tors}H^2(X;\Lambda)}
\sum_{p\in \operatorname{Tors}H^2(X;\Lambda)}
\sigma_{X,p}\otimes \mu_{X,p}
\ \in\
\mathcal H_{\mathbb T,K}(\partial X,L_X^{\mathbb R}).
$$
Thus a bordism prepares a superposition of contributions from its allowed bulk topological sectors. Each term carries a classical phase and a half-density on an allowed boundary leaf. Distinct bulk sectors need not give distinct boundary leaves, so the expression is a sum of amplitudes, not merely a count of sectors.
\end{prop}

\begin{proof}
This is the canonical boundary-state formula of \cite{Galviz2}.  Each $\Lambda_{X,p}$ is a Bohr--Sommerfeld leaf, $\sigma_{X,p}$ is a parallel section of the prequantum line over that leaf, and $\mu_{X,p}$ is the corresponding canonical torsion half-density.  The normalization by $|\det K|^{m_X}$ and the torsion average is forced by compatibility with the cylinder and gluing laws.
\end{proof}

When $\partial X=\varnothing$, there are no boundary variables left to label a state. The state space is canonically $\mathbb C$, and the same construction yields a scalar amplitude, the partition function. The remaining integral sums the continuous flat modes within each torsion sector.

\begin{definition}\cite[Definition 4.6]{Galviz2}
\label{def:closed-partition}
For a closed connected oriented $3$-manifold $M$, the Abelian Chern--Simons partition
function is
$$
Z^{CS}_{\mathbb T,K}(M):=
|\det K|^{m_M}\,
\frac{1}{\#\operatorname{Tors}H^2(M;\Lambda)}
\sum_{p\in \operatorname{Tors}H^2(M;\Lambda)}
\int_{\mathcal M_{M,p}(\mathbb T)}
\sigma_{M,p}\,(T_M(\mathrm t))^{1/2},
$$
where $(T_M(\mathrm t))^{1/2}$ is the translation-invariant density induced by the square
root of Reidemeister--Ray--Singer torsion.
\end{definition}

The same amplitude can be evaluated using a surgery presentation. In this description the continuous flat-mode contribution is incorporated into the normalization, while the torsion sectors appear explicitly in a finite quadratic sum. If $M=M_L$ is presented by integral surgery on a framed link $\mathcal L\subset S^3$
with linking matrix $L$, choose $U\in GL_m(\mathbb Z)$ such that
$$
U^\top L U=
\begin{pmatrix}
L_{\mathrm{reg}}&0\\
0&0
\end{pmatrix},
$$
with $L_{\mathrm{reg}}$ nondegenerate of size $\rho\times \rho$, where $\rho=m-b_1(M),$ and $m_M=\frac12\bigl(b_1(M)-1\bigr).$
Then
$$
\operatorname{Tors}H^2(M;\Lambda)\cong
\mathbb Z^{\rho n}/(L_{\mathrm{reg}}\otimes I_n)\mathbb Z^{\rho n},
$$
and the partition function takes the following finite quadratic Gauss-sum form.

\begin{theorem}\cite[Theorem 4.8, Eq.~(4.14)]{Galviz3}
\label{thm:surgery-expression}
For $M=M_L$ as above,
$$
Z^{CS}_{\mathbb T,K}(M_L)
=
|G_K|^{m_M}\,
|\det(L_{\mathrm{reg}})|^{-n/2}
\sum_{[x]\in \mathbb Z^{\rho n}/(L_{\mathrm{reg}}\otimes I_n)\mathbb Z^{\rho n}}
\exp\!\Bigl(\pi i\,x^\top (L_{\mathrm{reg}}^{-1}\otimes K)x\Bigr).
$$
Equivalently, if
$$
q_{L,K}([x])\equiv \frac12\,x^\top(L_{\mathrm{reg}}^{-1}\otimes K)x \pmod 1,
$$
then
$$
Z^{CS}_{\mathbb T,K}(M_L)
=
|G_K|^{m_M}\,
|\det(L_{\mathrm{reg}})|^{-n/2}
\sum_{[x]\in \mathbb Z^{\rho n}/(L_{\mathrm{reg}}\otimes I_n)\mathbb Z^{\rho n}} e^{2\pi i\,q_{L,K}([x])}.
$$
\end{theorem}

\begin{proof}
This is the partition surgery formula proved in \cite{Galviz3}.  It is obtained by rewriting the torsion-sector sum in Definition~\ref{def:closed-partition} as a finite quadratic Gauss sum determined by the quadratic refinement of the torsion linking pairing presented by the surgery matrix.
\end{proof}

The cylinder and gluing laws express the physical consistency of these amplitudes. A cylinder transports a state without changing it, while gluing contracts matching boundary states, equivalently summing over a complete set of intermediate quantum states. In the normalization used here these requirements read as follows.

\begin{prop}\cite[Proposition 4.1 and Theorem 4.5]{Galviz2}
\label{prop:cylinder-gluing}
For every closed oriented surface $\Sigma$,
$$Z^{CS}_{\mathbb T,K}(\Sigma\times I)
=\mathrm{Id}, $$
 and if $X_{\mathrm{cut}}$ is obtained from $X$ by cutting along a closed oriented surface $\Sigma$, then
$$
Z^{CS}_{\mathbb T,K}(X)
=
\operatorname{Tr}_\Sigma\!\bigl(Z^{CS}_{\mathbb T,K}(X_{\mathrm{cut}})\bigr).
$$
\end{prop}

Together, these statements promote the quantization of each surface to a theory of propagation and composition: surfaces carry state spaces and bordisms carry compatible linear amplitudes.

\begin{theorem}\cite[Theorem 4.8]{Galviz2}
\label{thm:toral-cs}
For every even, integral, nondegenerate lattice $(\Lambda,K)$ there is a unitary extended $(2+1)$-dimensional TQFT
$$
Z^{CS}_{\mathbb T,K}:\Cob^{\mathrm{ext}}_{2+1}\longrightarrow \mathrm{Vect}_{\mathbb C}.
$$
If $\Sigma_g$ is a connected closed oriented surface of genus $g$, then
$$
\dim Z^{CS}_{\mathbb T,K}(\Sigma_g)=|G_K|^g=|\det K|^g.
$$
\end{theorem}

\begin{proof}
This is the main theorem of \cite{Galviz2}.  The state spaces are obtained by geometric quantization in real polarization; the bordism vectors are built from the classical Abelian Chern--Simons section and torsion half-densities; the cylinder and gluing axioms are proved directly; and the projective BKS anomaly is absorbed by the $K$-twisted Maslov correction.  The dimension statement is the genus-$g$ formula above.
\end{proof}

\begin{remark}
There is also a complementary route to the same extended Abelian Chern--Simons theory via a rigorous functional-integral construction \cite{Galviz2.5}.  In that approach, one starts from the formal oscillatory quotient integral over connections modulo gauge and evaluates it exactly by zeta-regularized Gaussian methods.  After translating by a flat connection, the Abelian Chern--Simons action becomes quadratic; Hodge decomposition then separates the integral into harmonic, gauge, and nondegenerate coexact sectors.
 
The resulting determinant and $\eta$-invariant factors produce the universal normalization $|\det K|^{m_X}$ and the corresponding signature correction, while the remaining harmonic integral yields the same partition functions as in the geometric-quantization construction.  For manifolds with boundary, the relative functional integral recovers the same canonical boundary states, expressed by the Chern--Simons section together with the torsion half-density on the corresponding Bohr--Sommerfeld leaf.  Thus the functional-integral and real-polarization constructions define the same extended Abelian Chern--Simons TQFT.
\end{remark}

The finite data obtained from quantizing the continuous torus phase space also admit an algebraic description. The discriminant group labels the simple sectors of a pointed modular category, and its quadratic form specifies their twist phases. This provides the link to the Reshetikhin--Turaev construction. We denote by $Z^{RT}_{\mathcal C}$ the Reshetikhin--Turaev TQFT associated with a modular tensor category $\mathcal C$, following \cite{Reshetikhin:1991}. For the pointed modular category $\mathcal  C=\mathcal C(G_K,q_K)$ case, see also \cite{Galviz3}.

\begin{theorem}\cite[Theorem 4.15]{Galviz3} \label{thm:cs-rt-equivalence} For every even, integral, nondegenerate lattice $(\Lambda,K)$ there is a symmetric monoidal natural isomorphism
$$
\Phi_K:\; Z^{RT}_{\mathcal C(G_K,q_K)}\cong Z^{CS}_{\mathbb T,K}.
$$
\end{theorem}

\begin{proof}
This is the main equivalence theorem of \cite{Galviz3}.  On closed manifolds, the comparison reduces to the reciprocity relation between the RT surgery scalar and the Abelian Chern--Simons partition function.  On bordisms with boundary, one compares matrix coefficients against canonical handlebody states, leading to an equivalence of extended functors.
\end{proof}

The final classification makes precise how much of the classical coupling survives in the extended quantum theory under consideration. Different lattice presentations can encode equivalent finite quadratic data and hence equivalent quantum theories.

\begin{theorem}\cite[Theorem 3.2 and Corollary 3.3]{Galviz4}
\label{thm:classification}
Let $(\Lambda,K)$ and $(\Lambda',L)$ be even, integral, nondegenerate lattices. Then the following are equivalent: \begin{enumerate}[label=\rm(\alph*)] \item $(G_K,q_K)\cong (G_L,q_L)$ as finite quadratic modules; \item $Z^{RT}_{\mathcal C(G_K,q_K)}\cong Z^{RT}_{\mathcal C(G_L,q_L)}$ as symmetric monoidal extended TQFTs; \item $Z^{CS}_{\mathbb T,K}\cong Z^{CS}_{\mathbb T,L}$ as symmetric monoidal extended TQFTs.
\end{enumerate}
\end{theorem}

\begin{proof}
This is the classification theorem proved in \cite{Galviz4}.  The equivalence of \rm(a) and \rm(b) is the pointed modular-category statement.  The implication \rm(b)$\Rightarrow$\rm(c) follows by realizing the natural isomorphism through Theorem~\ref{thm:cs-rt-equivalence}, while \rm(c)$\Rightarrow$\rm(b) follows by transport in the opposite direction.
\end{proof}

\section{BF theory and Abelian Chern--Simons theory}
\label{sec:bf}

We now rewrite the BF embedding in toral notation; for background on BF theory, see \cite{Horowitz89,BlauThompson91,BBRT91}. Let $N \in M_n(\mathbb Z)$ be an integral matrix with $\det N \neq 0$. Introduce the upper-triangular matrix
$$
C_{BF}(N):=
\begin{pmatrix}
0 & N\\
0 & 0
\end{pmatrix}
\in M_{2n}(\mathbb Z),
$$
and its symmetrization
$$
K_{BF}(N):=C_{BF}(N)+C_{BF}(N)^{\top}
=
\begin{pmatrix}
0 & N\\
N^{\top} & 0
\end{pmatrix}.
$$
The matrix $C_{BF}(N)$ records the unsymmetrized BF coupling, while the Abelian Chern--Simons theory is defined by the even symmetric lattice with Gram matrix $K_{BF}(N)$.

\begin{prop}
\label{prop:bf-embedding}
The matrix
$$
K_{BF}(N)=
\begin{pmatrix}
0 & N\\
N^{\top} & 0
\end{pmatrix}
$$
is even, integral, symmetric, and nondegenerate. Hence it defines a toral extended TQFT
$$
Z^{BF}_N:=Z^{CS}_{\mathbb T_{BF},K_{BF}}:
\Cob^{\mathrm{ext}}_{2+1}\longrightarrow \mathrm{Vect}_{\mathbb C},
\qquad
\mathbb T_{BF}:=\mathbb R^{2n}/\mathbb Z^{2n}.
$$

Moreover, on a closed oriented $3$-manifold $M$, if we write
$$
\mathcal A=(A,B)^{\top},
\qquad
A=(A_1,\dots,A_n)^{\top},
\qquad
B=(B_1,\dots,B_n)^{\top},
$$
then the Abelian Chern--Simons action with level $K_{BF}(N)$,
$$
S^{CS}_{K_{BF}(N)}[\mathcal A]
:=
\frac{1}{4\pi}\int_M \mathcal A^{\top}K_{BF}(N)\,d\mathcal A,
$$
reduces to
$$
S^{CS}_{K_{BF}(N)}[\mathcal A]
=
\frac{1}{4\pi}\int_M
\bigl(
A^{\top}N\,dB + B^{\top}N^{\top}\,dA
\bigr)
=
\frac{1}{2\pi}\int_M A^{\top}N\,dB,
$$
where the second equality holds on closed $M$ by integration by parts. Thus, up to the standard normalization convention (and, on manifolds with boundary, up to the usual boundary term), the off-diagonal Abelian Chern--Simons theory with level $K_{BF}(N)$ reproduces the multicomponent Abelian BF theory with coupling matrix $N$.
\end{prop}

\begin{proof}
The matrix $K_{BF}(N)$ is visibly integral and symmetric. Its diagonal entries vanish, so it is even. Its determinant is
$$
\det K_{BF}(N)=(-1)^n\det(NN^{\top})\neq 0,
$$
so it is nondegenerate. Theorem~\ref{thm:toral-cs} therefore gives the extended functor $Z^{CS}_{\mathbb T_{BF},K_{BF}}.$ For the action, a direct computation gives
$$
\mathcal A^{\top}K_{BF}(N)\,d\mathcal A
=
A^{\top}N\,dB + B^{\top}N^{\top}\,dA.
$$
Since $M$ is closed,
$$
\int_M B^{\top}N^{\top}\,dA
=
\int_M A^{\top}N\,dB
$$
by Stokes' theorem. Hence
$$
S^{CS}_{K_{BF}(N)}[\mathcal A]
=
\frac{1}{2\pi}\int_M A^{\top}N\,dB.
$$
This is precisely the multicomponent Abelian BF action with coupling matrix $N$, up to
the overall normalization convention. Therefore BF theory is realized as the off-diagonal
Abelian Chern--Simons theory with level $K_{BF}(N)$.
\end{proof}

\begin{remark}
Once BF theory is presented by the even lattice $(\mathbb Z^{2n},K_{BF})$, all of the extended structures: boundary state spaces, BKS comparison operators, bordism vectors, and gluing laws are inherited from the Abelian Chern--Simons construction of \cite{Galviz2}. For background on three-dimensional BF theory and its observables, see also \cite{CCRFM95}.
\end{remark}

\begin{theorem}
\label{thm:bf-tqft-data}
Let $N\in M_n(\mathbb Z)$ with $\det N\neq 0$, and define
$$
K_{BF}(N):=
\begin{pmatrix}
0 & N\\
N^\top & 0
\end{pmatrix},
\qquad
\mathbb T_{BF}:=\mathbb R^{2n}/\mathbb Z^{2n},
\qquad
Z_N^{BF}:=Z^{CS}_{\mathbb T_{BF},K_{BF}}.
$$
Then $Z_N^{BF}$ is a unitary extended $(2+1)$-dimensional TQFT. If $G_N:=\mathbb Z^{2n}/K_{BF}\mathbb Z^{2n},$ then
$$
|G_N|=|\det K_{BF}|=|\det N|^2.
$$

Moreover:

\begin{enumerate}[label=\roman*)]
\item For every connected closed oriented surface $\Sigma_g$ of genus $g$, and every
rational Lagrangian $L\subset H^1(\Sigma_g;\mathbb R)$,
$$
\mathcal H_N^{BF}(\Sigma_g,L)
=
\bigoplus_{\ell\in \mathcal{BS}(\Sigma_g,L)}
\Gamma_{\mathrm{flat}}\!\bigl(\ell;\,\mathcal L_{\Sigma_g,K_{BF}}
\otimes |\det P_L^*|^{1/2}\bigr),
$$
where $\mathcal{BS}(\Sigma_g,L)$ is a torsor for $G_N^g$. In particular,
$$
\dim \mathcal H_N^{BF}(\Sigma_g,L)=|G_N|^g=|\det N|^{2g}.
$$

\item If $X$ is a compact oriented $3$-manifold with boundary, then
$$
Z_N^{BF}(X)
=
|\det N|^{2m_X}\,
\frac{1}{\#\operatorname{Tors}H^2(X;\mathbb Z^{2n})}
\sum_{p\in \operatorname{Tors}H^2(X;\mathbb Z^{2n})}
\sigma_{X,p}\otimes \mu_{X,p}
\;\in\;
\mathcal H_N^{BF}(\partial X,L_X^{\mathbb R}),
$$
with the usual $m_X$.

\item If $M$ is a closed connected oriented $3$-manifold, then
$$
Z_N^{BF}(M)
=
|\det N|^{2m_M}\,
\frac{1}{\#\operatorname{Tors}H^2(M;\mathbb Z^{2n})}
\sum_{p\in \operatorname{Tors}H^2(M;\mathbb Z^{2n})}
\int_{\mathcal M_{M,p}(\mathbb T_{BF})}
\sigma_{M,p}\,(T_M(\mathbb R^{2n}))^{1/2},
$$
where $m_M=\frac12\bigl(b_1(M)-1\bigr).$

\item If $M=M_L$ is presented by integral surgery on a framed link with linking matrix
$L$, and if $U\in GL_m(\mathbb Z)$ is chosen so that $U^\top L U=\left(\begin{smallmatrix} L_{\mathrm{reg}}&0\\ 0&0 \end{smallmatrix}\right)$ with $L_{\mathrm{reg}}$ nondegenerate of size $\rho\times \rho$, where $\rho=m-b_1(M), $ and $m_M=\frac12\bigl(b_1(M)-1\bigr),$ then

$$
\operatorname{Tors}H^2(M;\mathbb Z^{2n})
\cong
\mathbb Z^{2n\rho}/(L_{\mathrm{reg}}\otimes I_{2n})\mathbb Z^{2n\rho},
$$
such that

$$
Z_N^{BF}(M_L)
=
|G_N|^{m_M}\,
|\det(L_{\mathrm{reg}})|^{-n}
\sum_{[x]\in \mathbb Z^{2n\rho}/(L_{\mathrm{reg}}\otimes I_{2n})\mathbb Z^{2n\rho}}
\exp\!\Bigl(\pi i\,x^\top(L_{\mathrm{reg}}^{-1}\otimes K_{BF}(N))x\Bigr).
$$
Equivalently, if
$$
q^{BF}_{L,N}([x])\equiv
\frac12\,x^\top(L_{\mathrm{reg}}^{-1}\otimes K_{BF})x
\pmod 1,
$$
then
$$
Z_N^{BF}(M_L)
=
|G_N|^{m_M}\,
|\det(L_{\mathrm{reg}})|^{-n}
\sum_{[x]\in \mathbb Z^{2n\rho}/(L_{\mathrm{reg}}\otimes I_{2n})\mathbb Z^{2n\rho}}
e^{2\pi i\,q^{BF}_{L,N}([x])}.
$$
\end{enumerate}
\end{theorem}

\begin{proof}
The matrix $K_{BF}(N)$ is integral and symmetric, and its diagonal entries vanish, so it
is even. Moreover,
$$
\det K_{BF}(N)=(-1)^n\det(NN^\top)\neq 0,
$$
since $\det N\neq 0$. Thus $(\mathbb Z^{2n},K_{BF})$ is an even, integral,
nondegenerate lattice, and therefore defines a unitary extended $(2+1)$-dimensional
Abelian Chern--Simons TQFT
$$
Z_N^{BF}=Z^{CS}_{\mathbb T_{BF},K_{BF}}.
$$

Its discriminant group is
$$
G_N=\mathbb Z^{2n}/K_{BF}(N)\mathbb Z^{2n},
$$
hence
$$
|G_N|=|\det K_{BF}(N)|=|\det(NN^\top)|=|\det N|^2.
$$

Statements $(i)$ , $(ii)$, and $(iii)$ are exactly the general surface-space, boundary-state, and closed-partition formulas for Abelian Chern--Simons theory, reduced to the lattice $(\mathbb Z^{2n},K_{BF})$. In particular,
$$
\dim \mathcal H_N^{BF}(\Sigma_g,L)=|G_N|^g=|\det N|^{2g},
$$
and the normalization factors become
$$
|\det K_{BF}(N)|^{m_X}=|\det N|^{2m_X},
\qquad
|\det K_{BF}(N)|^{m_M}=|\det N|^{2m_M}.
$$

For $(iv)$, apply the general surgery formula to the same lattice. Since the
lattice rank is $2n$, one gets
$$
\operatorname{Tors}H^2(M;\mathbb Z^{2n})
\cong
\mathbb Z^{2n\rho}/(L_{\mathrm{reg}}\otimes I_{2n})\mathbb Z^{2n\rho},
$$
with the prefactor
$$
|\det(L_{\mathrm{reg}})|^{-(2n)/2}=|\det(L_{\mathrm{reg}})|^{-n}.
$$
This gives precisely the stated finite quadratic Gauss-sum expression.
\end{proof}

\section{Reciprocity Relations}
\label{sec:reciprocity}

Using the Deligne--Beilinson cohomology framework, \cite{Kim2024,Tagaris2025} constructed a $U(1)^n$ Chern--Simons partition function written as a normalized functional integral, where the normalization factor $\mathcal N_{\mathrm{CS}_C}(M)$ was chosen so that the resulting expression reduces to a finite Gauss sum (see Eqs. (14)-(16) in \cite{Kim2024}), as follows\footnote{See~\cite{Kim2024,Tagaris2025} for the precise definition of these expressions.}:
$$
\begin{aligned}
Z^{\mathrm{CS}}(M)
&=
\frac{1}{\mathcal N_{\mathrm{CS}_C}(M)}
\sum_{\kappa_A\in \left(TH^{2}(M)\right)^{n}}
\exp\!\left(-2\pi i\,{}^{t}\!\kappa_A\,(C\otimes Q)\,\kappa_A\right)
\\
&\qquad\cdot
\int_{\left(\Omega^{1}(M)/\Omega^{1}_{\mathrm{cl}}(M)\right)^{n}}
D\alpha^{\perp}\,
\exp\!\left(2\pi i\int_{M} {}^{t}\!\alpha^{\perp}\star C\alpha^{\perp}\right),
\end{aligned}
$$
where the contribution from $ \left(\Omega^{1}(M)/\Omega^{1}_{\mathrm{cl}}(M)\right)^{n} $ appears as an infinite-dimensional functional integral, which is then removed by choosing
$$
\mathcal N_{\mathrm{CS}_C}(M)=
\int_{\left(\Omega^{1}(M)/\Omega^{1}_{\mathrm{cl}}(M)\right)^{n}}
D\alpha^{\perp}\,
\exp\!\left(2\pi i\int_{M} {}^{t}\!\alpha^{\perp}\star C\alpha^{\perp}\right).
$$
This leads to a heuristic partition function for closed manifolds
\begin{equation}
\label{eq:DB-partition-function}
Z^{\mathrm{CS}}(M)=
\sum_{\kappa_A\in \left(TH^{2}(M)\right)^{n}}
\exp\!\left(-2\pi i\,{}^{t}\!\kappa_A\,(C\otimes Q)\,\kappa_A\right).
\end{equation}

This simplification does not retain some of the topological information tracked by the TQFT formalism, and as a result, it suggests an apparent duality between Abelian Chern--Simons theories. As we show below, however, such duality does not hold in the TQFT formalism. Accordingly, their analysis gives identities between these closed-manifold partition functions and motivates the following reciprocity diagram for a closed oriented $3$-manifold\footnote{In \cite{Tagaris2025}, the analysis is restricted to rational homology spheres, since the surgery linking matrix $L$ is assumed to be non-degenerate, i.e. $(\det L\neq 0)$. Here, we extend the study of these relations to arbitrary closed manifolds $M$.} \cite{Tagaris2025}: 
$$
\begin{tikzcd}[column sep=large,row sep=large]
Z^{\mathrm{CS}}_{K,Q_M}(M)
\arrow[r,leftrightarrow,"\mathrm{K\leftrightarrow L\ duality}"']
\arrow[r,phantom,"\scriptstyle ?",yshift=1.6ex]
\arrow[d,leftrightarrow,"\mathrm{recip.}"']
&
Z^{\mathrm{CS}}_{L,Q_K}(M)
\arrow[d,leftrightarrow,"\mathrm{recip.}"]
\\
Z^{RT}_{Q_K,L}(M)
\arrow[r,leftrightarrow,"\mathrm{K\leftrightarrow L\ duality}"']
\arrow[r,phantom,"\scriptstyle ?",yshift=1.6ex]
&
Z^{RT}_{Q_L,K}(M)
\end{tikzcd}
$$

In particular, our formalism establishes the long-expected equivalence between the Abelian Chern--Simons and Reshetikhin--Turaev TQFTs \cite{Galviz1,Galviz3}. On closed $3$-manifolds, these functorial equivalences induce the vertical identifications represented in the diagram. The genuinely new feature suggested by the diagram above is the presence of the horizontal arrows, which in \cite{Kim2024,Tagaris2025} are interpreted as a duality between Abelian Chern--Simons theories. In the present framework, however, these arrows encode a relation between partition surgery evaluations, not a genuine duality of QFTs. 

Once the Abelian Chern--Simons theory is constructed as a unitary extended $(2+1)$-dimensional TQFT, identified with the Reshetikhin--Turaev theory associated with the discriminant finite quadratic module $(G_K,q_K)$, and classified by finite quadratic modules, one can distinguish clearly between statements that extend functorially and those that do not. Let us consider first the vertical arrows. 

By Theorem~\ref{thm:cs-rt-equivalence}, the vertical arrows in the diagram are symmetric monoidal natural isomorphisms
$$
Z^{RT}_{\mathcal C(G_K,q_K)}\cong Z^{CS}_{\mathbb T_K,K},
\qquad
Z^{RT}_{\mathcal C(G_L,q_L)}\cong Z^{CS}_{\mathbb T_L,L}.
$$
Thus the issue concerns only the proposed horizontal $K\leftrightarrow L$ arrows.

For a general closed $3$-manifold $M$, a surgery linking matrix $L$ need not be nondegenerate; when $b_1(M)>0$, one instead uses its nondegenerate block $L_{\mathrm{reg}}$ in the general surgery formula. This should be distinguished from the Chern--Simons level matrix, which is required to be nondegenerate. The horizontal reciprocity comparison below is therefore restricted to the case in which both $K$ and $L$ are nondegenerate, so that $L^{-1}$ is defined and, in particular, $-L$ may itself be regarded as the level matrix of a second Abelian Chern--Simons theory. Thus this role-swapping argument applies only to the nondegenerate surgery case, while the underlying Chern--Simons TQFT remains defined on arbitrary closed $3$-manifolds.

Let us consider two different Chern--Simons theories $Z^{CS}_{\mathbb T_K,K}(M_L)$ and  $Z^{CS}_{\mathbb T_L,-L}(M_K)$ as constructed in \cite{Galviz2}, with  $K\in M_n(\mathbb Z)$ and $L\in M_m(\mathbb Z)$  for $K,L$  even, symmetric, and nondegenerate, and $M_K$ and $M_L$ denote the closed oriented $3$-manifolds obtained by surgery with linking matrices $K$ and $L$, respectively. Since $b_1(M_K)=b_1(M_L)=0$, Theorem~\ref{thm:surgery-expression} gives
\begin{equation}
\label{eq:toral-surgery-K-on-L}
Z^{CS}_{\mathbb T_K,K}(M_L)
=
|G_K|^{-1/2}\,|\det L|^{-n/2}
\sum_{x\in \mathbb Z^{mn}/(L\otimes I_n)\mathbb Z^{mn}}
\exp\!\bigl(\pi i\,x^\top(L^{-1}\otimes K)x\bigr),
\end{equation}
and, with the sign convention chosen so that the transformed Gauss sum appears as the partition function of the dual closed theory,
\begin{equation}
\label{eq:toral-surgery-minusL-on-K}
Z^{CS}_{\mathbb T_L,-L}(M_K)
=
|G_L|^{-1/2}\,|\det K|^{-m/2}
\sum_{y\in \mathbb Z^{mn}/(K\otimes I_m)\mathbb Z^{mn}}
\exp\!\bigl(-\pi i\,y^\top(K^{-1}\otimes L)y\bigr).
\end{equation}

Let us recall the reciprocity transformation from \cite[Theorem 1]{Deloup2005}:
$$\begin{aligned} \label{eq:CS-reciprocity} &|\det L|^{-n/2} \sum_{x\in \mathbb Z^{mn}/(L\otimes I_n)\mathbb Z^{mn}} \exp\!\bigl(\pi i\,x^\top(L^{-1}\otimes K)x\bigr)\\ &=|\det K|^{-m/2} e^{\frac{\pi i}{4}\sigma(K)\sigma(L)} \sum_{y\in \mathbb Z^{mn}/(K\otimes I_m)\mathbb Z^{mn}} \exp\!\bigl(-\pi i\,y^\top(K^{-1}\otimes L)y\bigr). \end{aligned}$$
Based on this relation, it was argued in \cite[Section 2.5]{Kim2024} that one can define two Abelian Chern--Simons theories,
$Z^{CS}_K(M_L)$ and $Z^{CS}_L(M_K)$, for closed manifolds $M_L$ and $M_K$, respectively, with the following ad hoc normalized partition functions; see \cite[Sec.~(E)]{Kim2024}:
$$
\begin{aligned}
Z^{CS}_K(M_L)
&:=
|\det L|^{-n/2}
\sum_{x\in \mathbb Z^{mn}/(L\otimes I_n)\mathbb Z^{mn}}
\exp\!\bigl(\pi i\,x^\top(L^{-1}\otimes K)x\bigr),
\\
Z^{CS}_{-L}(M_K)
&:=
|\det K|^{-m/2}
e^{\frac{\pi i}{4}\sigma(K)\sigma(L)}
\sum_{y\in \mathbb Z^{mn}/(K\otimes I_m)\mathbb Z^{mn}}
\exp\!\bigl(-\pi i\,y^\top(K^{-1}\otimes L)y\bigr).
\end{aligned}
$$
However, even at the level of closed manifolds, this should be interpreted with care: the two expressions are not obtained from a single uniform normalization prescription, since the signature phase is assigned differently in the two theories; compare this with their initial partition function \eqref{eq:DB-partition-function}. Before applying reciprocity, our framework does not produce the phase factor
$
e^{\frac{\pi i}{4}\sigma(K)\sigma(L)}
$
as can be seen by comparing with \eqref{eq:toral-surgery-K-on-L}--\eqref{eq:toral-surgery-minusL-on-K}. Thus the numerical agreement suggested by reciprocity is better viewed as an identity between differently normalized Gauss sums than as a  duality of partition functions. 
Let
$$
Y_L:\;\bigsqcup_{i=1}^{m}(T^2,\lambda)\to\varnothing,
\qquad
Y_K:\;\bigsqcup_{i=1}^{n}(T^2,\lambda)\to\varnothing
$$
be the standard surgery bordisms associated with framed links whose linking matrices are $L$ and $K$, respectively. These bordisms are geometric surgery data and must be distinguished from the Chern--Simons level used to evaluate them. In particular, the second reciprocity term is computed by applying the theory with level $-L$ to the surgery bordism $Y_K$.
$$
\mathcal S_{K,L}:=
Z^{CS}_{\mathbb T_K,K}(Y_L)
\in
\mathcal H_{\mathbb T_K,K}(T^2,\lambda)^{\otimes m\,\vee},\quad
\mathcal S_{-L,K}:=
Z^{CS}_{\mathbb T_L,-L}(Y_K)
\in
\mathcal H_{\mathbb T_L,-L}(T^2,\lambda)^{\otimes n\,\vee},
$$
where
$$
\mathcal H_{\mathbb T_K,K}(T^2,\lambda)
:=
Z^{CS}_{\mathbb T_K,K}(T^2,\lambda),\qquad
\mathcal H_{\mathbb T_L,-L}(T^2,\lambda):=Z^{CS}_{\mathbb T_L,-L}(T^2,\lambda).
$$

Since $G_{-L}=G_L$ as underlying finite Abelian groups, choose Bohr--Sommerfeld bases identified, under the CS/RT equivalence, with the standard solid-torus core states
$$
\{e_u\}_{u\in G_K}
\subset
\mathcal H_{\mathbb T_K,K}(T^2,\lambda),
\qquad
\{e_v^{-}\}_{v\in G_L}
\subset
\mathcal H_{\mathbb T_L,-L}(T^2,\lambda),
$$
and define the Kirby states
$$
\Omega_K
:=
\sum_{u\in G_K}e_u,
\qquad
\Omega_{-L}
:=
\sum_{v\in G_L}e_v^{-}.
$$
If $M_L$ and $M_K$ denote the closed oriented $3$-manifolds
obtained by surgery with linking matrices $L$ and $K$, then the
gluing axiom gives
$$
\mathcal S_{K,L}(\Omega_K^{\otimes m})
=
Z^{CS}_{\mathbb T_K,K}(M_L), \qquad
\mathcal S_{-L,K}(\Omega_{-L}^{\otimes n})
=
Z^{CS}_{\mathbb T_L,-L}(M_K).
$$

Here the notation $M_K$ refers only to the closed manifold obtained by surgery on the $K$-framed link; the Chern--Simons theory used to evaluate this bordism is the theory with level $-L$. Thus the two sides of the reciprocity relation involve different assignments of level data to the same type of surgery construction.

\begin{prop}
\label{prop:no-intrinsic-horizontal-reciprocity}
Let $K\in M_n(\mathbb Z)$ and $L\in M_m(\mathbb Z)$ be even, symmetric, and nondegenerate. Then the canonical Abelian Chern--Simons partition functions satisfy

\begin{equation}
Z^{CS}_{\mathbb T_K,K}(M_L)
=
e^{\frac{\pi i}{4}\sigma(K)\sigma(L)}
\sqrt{\frac{|G_L|}{|G_K|}}\,
Z^{CS}_{\mathbb T_L,-L}(M_K).
\end{equation}

Hence they are not equal in general. In particular, the horizontal $K\leftrightarrow L$ equality obtained from the normalization of \cite{Kim2024,Tagaris2025} does not, in general, coincide with the canonical Abelian Chern--Simons partition function.
\end{prop}

\begin{proof}
Let $K\in M_n(\mathbb Z)$ and $L\in M_m(\mathbb Z)$ be even, symmetric, and nondegenerate.\footnote{If $K$ or $L$ is degenerate, the corresponding reciprocity relation can still be formulated after passing to the appropriate nondegenerate part of the quadratic data. However, the interpretation of the two sides as partition functions of two distinct Abelian Chern--Simons theories is then lost, since a degenerate matrix does not define a Chern--Simons theory.}
Since $L$ and $K$ are nondegenerate surgery matrices, the corresponding surgery manifolds $M_L$ and $M_K$ are rational homology spheres. Hence
$$
b_1(M_L)=b_1(M_K)=0,
\qquad
m_{M_L}=m_{M_K}=-\frac12.
$$

Therefore, by Theorem~\ref{thm:surgery-expression},

$$
\begin{aligned}
Z^{CS}_{\mathbb T_K,K}(M_L)
&=
|G_K|^{-1/2}|\det L|^{-n/2}
\sum_{x\in (\mathbb Z^m/L\mathbb Z^m)^n}
\exp\!\bigl(
\pi i\,x^\top(L^{-1}\otimes K)x
\bigr),
\\
Z^{CS}_{\mathbb T_L,-L}(M_K)
&=
|G_L|^{-1/2}|\det K|^{-m/2}
\sum_{u\in (\mathbb Z^n/K\mathbb Z^n)^m}
\exp\!\bigl(
-\pi i\,u^\top(K^{-1}\otimes L)u
\bigr).
\end{aligned}
$$

Using $|G_K|=|\det K|,$ and $|G_L|=|\det L|,$ we may rewrite these as

$$
\begin{aligned}
Z^{CS}_{\mathbb T_K,K}(M_L)
&=
|G_K|^{-1/2}|G_L|^{-n/2}
\sum_{x\in (\mathbb Z^m/L\mathbb Z^m)^n}
\exp\!\bigl(
\pi i\,x^\top(L^{-1}\otimes K)x
\bigr),
\\
Z^{CS}_{\mathbb T_L,-L}(M_K)
&=
|G_L|^{-1/2}|G_K|^{-m/2}
\sum_{u\in (\mathbb Z^n/K\mathbb Z^n)^m}
\exp\!\bigl(
-\pi i\,u^\top(K^{-1}\otimes L)u
\bigr).
\end{aligned}
$$

Now apply the reciprocity formula~\eqref{eq:CS-reciprocity}:

$$
\begin{aligned}
&|G_L|^{-n/2}
\sum_{x\in (\mathbb Z^m/L\mathbb Z^m)^n}
\exp\!\bigl(
\pi i\,x^\top(L^{-1}\otimes K)x
\bigr)
\\
&\qquad =
e^{\frac{\pi i}{4}\sigma(K)\sigma(L)}
|G_K|^{-m/2}
\sum_{u\in (\mathbb Z^n/K\mathbb Z^n)^m}
\exp\!\bigl(
-\pi i\,u^\top(K^{-1}\otimes L)u
\bigr).
\end{aligned}
$$

Multiplying both sides by $|G_K|^{-1/2}$ gives

$$
\begin{aligned}
Z^{CS}_{\mathbb T_K,K}(M_L)
&=
e^{\frac{\pi i}{4}\sigma(K)\sigma(L)}
|G_K|^{-1/2}|G_K|^{-m/2}
\\
&\qquad\qquad\times
\sum_{u\in (\mathbb Z^n/K\mathbb Z^n)^m}
\exp\!\bigl(
-\pi i\,u^\top(K^{-1}\otimes L)u
\bigr).
\end{aligned}
$$

On the other hand, from the second partition-function formula,

$$
|G_K|^{-m/2}
\sum_{u\in (\mathbb Z^n/K\mathbb Z^n)^m}
\exp\!\bigl(
-\pi i\,u^\top(K^{-1}\otimes L)u
\bigr)
=
|G_L|^{1/2}
Z^{CS}_{\mathbb T_L,-L}(M_K).
$$

Substituting this identity, we obtain 
$$
Z^{CS}_{\mathbb T_K,K}(M_L)
=
e^{\frac{\pi i}{4}\sigma(K)\sigma(L)}
\sqrt{\frac{|G_L|}{|G_K|}}\,
Z^{CS}_{\mathbb T_L,-L}(M_K).
$$
Thus reciprocity gives a definite scalar relation between the two canonically normalized surgery partition functions. In particular, it does not in general imply

$$
Z^{CS}_{\mathbb T_K,K}(M_L)
=
Z^{CS}_{\mathbb T_L,-L}(M_K),
$$

since the scalar  $e^{\frac{\pi i}{4}\sigma(K)\sigma(L)}
\sqrt{\frac{|G_L|}{|G_K|}}
$ need not be equal to $1$. 

To see this explicitly, consider $K=(2), L=(4).$ Then $
|G_K|=2,$ and $|G_L|=4,$ with $\sigma(K)=\sigma(L)=1.
$ The surgery formula gives

$$
Z^{CS}_{\mathbb T_K,K}(M_L)
=
2^{-1/2}4^{-1/2}
\sum_{x\in\mathbb Z/4\mathbb Z}
\exp\!\left(\frac{\pi i}{2}x^2\right).
$$
The four terms are: $1,\quad i,\quad 1,\quad i.\,$ Then
$$
\sum_{x\in\mathbb Z/4\mathbb Z}
\exp\!\left(\frac{\pi i}{2}x^2\right)
=
2+2i.
$$
Therefore,
$$
Z^{CS}_{\mathbb T_K,K}(M_L)
=
\frac{2+2i}{2\sqrt2}
=
\frac{1+i}{\sqrt2}
=
e^{\pi i/4}.
$$

For the theory with level $-L$ evaluated on $M_K$,

$$
Z^{CS}_{\mathbb T_L,-L}(M_K)
=
4^{-1/2}2^{-1/2}
\sum_{u\in\mathbb Z/2\mathbb Z}
\exp(-2\pi i u^2).
$$

Since $
\exp(-2\pi i u^2)=1,\,\,
\text{for }u=0,1.$ The Gauss sum equals $2$, and hence

$$
Z^{CS}_{\mathbb T_L,-L}(M_K)
=
\frac{2}{2\sqrt2}
=
\frac1{\sqrt2}.
$$

The general reciprocity relation obtained above now gives

$$
\begin{aligned}
e^{\frac{\pi i}{4}\sigma(K)\sigma(L)}
\sqrt{\frac{|G_L|}{|G_K|}}\,
Z^{CS}_{\mathbb T_L,-L}(M_K)
&=
e^{\pi i/4}
\sqrt{\frac42}\,
\frac1{\sqrt2}
\\
&=
e^{\pi i/4}
\\
&=
Z^{CS}_{\mathbb T_K,K}(M_L).
\end{aligned}
$$

Thus reciprocity is satisfied exactly. Nevertheless,

$$
Z^{CS}_{\mathbb T_K,K}(M_L)
=
e^{\pi i/4}
\neq
\frac1{\sqrt2}
=
Z^{CS}_{\mathbb T_L,-L}(M_K).
$$

Hence the reciprocity formula relates the two canonical surgery partition functions by a nontrivial scalar factor, rather than identifying them. This shows explicitly why the horizontal $K\leftrightarrow L$ reciprocity relation should not, in general, be interpreted as an equality of canonical Abelian Chern--Simons partition functions.
\end{proof}

\begin{prop}
\label{prop:not-functorial-equivalence}
The reciprocity formula by itself does not define a symmetric monoidal natural isomorphism $Z^{CS}_{\mathbb T_K,K}\cong Z^{CS}_{\mathbb T_L,-L}$ of extended TQFTs in general.
\end{prop}

\begin{proof}
A symmetric monoidal natural isomorphism between
$Z^{CS}_{\mathbb T_K,K}$ and $Z^{CS}_{\mathbb T_L,-L}$ would in particular induce an
isomorphism on the common object $(T^2,\lambda)$. But
$$
\dim \mathcal H_{\mathbb T_K,K}(T^2,\lambda)=|G_K|,
\qquad
\dim \mathcal H_{\mathbb T_L,-L}(T^2,\lambda)=|G_L|,
$$
and these dimensions need not coincide. Moreover, as seen above, reciprocity compares the surgery covectors $Z^{CS}_{\mathbb T_K,K}(Y_L)$ and $Z^{CS}_{\mathbb T_L,-L}(Y_K)$, which in general belong to different tensor powers of different torus state spaces. Hence, in general, there is no natural isomorphism between the two functors.
\end{proof}

\begin{remark}
The same conclusion holds on the Reshetikhin--Turaev side. Reciprocity relates closed surgery scalars but does not define a symmetric monoidal natural isomorphism of the corresponding TQFTs. A genuine equivalence is instead governed by the classification theorem, 
$$ Z^{CS}_{\mathbb T,K}\cong Z^{CS}_{\mathbb T,L} \quad\Longleftrightarrow\quad (G_K,q_K)\cong(G_L,q_L). $$
Thus reciprocity and equivalence of extended theories are distinct statements.
\end{remark}
Using Deligne--Beilinson cohomology, it was shown that every  $U(1)^n$ Chern--Simons partition function is related, via reciprocity, to an Abelian Reshetikhin--Turaev-type invariant for closed rational homology spheres \cite{Tagaris2025}. More precisely, the invariance of the resulting RT scalar was established through its identification with the reciprocity expression for the partition function, rather than through an independent modular-categorical construction of the full Reshetikhin--Turaev TQFT \cite{Reshetikhin:1991}. See also \cite{Galviz3} for the construction of the Reshetikhin--Turaev invariants corresponding to finite quadratic modules, namely the precise invariants relevant to the equivalence with Abelian Chern--Simons theory.

They also observe that, in their framework, not every Abelian Reshetikhin--Turaev-type invariant arises from a $U(1)^n$ Chern--Simons partition function. At first sight, this suggests that the Chern--Simons side is strictly smaller than the Reshetikhin--Turaev side. However, this apparent non-surjectivity comes from their choice of a non-modular category they called the “twisted” category on the Reshetikhin--Turaev side. In their construction, the reciprocity formulas are interpreted as Abelian Reshetikhin--Turaev-type invariants built from quadratic data associated with the level matrix $K$ and the surgery matrix $L$, but the argument is used only at the level of closed rational homology spheres. In the present setting, the modularity requirement enters naturally, since we use Turaev’s extended Reshetikhin--Turaev machinery to obtain a $(2+1)$-dimensional TQFT \cite{Turaev1994}.

The apparent non-surjectivity is resolved once one restricts to the pointed modular Reshetikhin--Turaev theories relevant to Abelian Chern--Simons theory. Indeed, by Theorems~\ref{thm:cs-rt-equivalence} and~\ref{thm:classification}, every finite quadratic module is realized by an even, integral, nondegenerate lattice, and two such Abelian Chern--Simons theories are equivalent if and only if their discriminant quadratic modules are isomorphic. Hence the assignment
$$
(\Lambda,K)\longmapsto \mathcal C(G_K,q_K)
$$
identifies equivalence classes of Abelian Chern--Simons extended TQFTs with the corresponding pointed modular Reshetikhin--Turaev TQFTs. Thus the apparent non-surjectivity arises from enlarging the Reshetikhin--Turaev side beyond the pointed modular sector naturally associated with Abelian Chern--Simons theory.

\section{Symmetries and orientation-reversal duality of Abelian Chern--Simons theories}
\label{sec:symmetries}

Delmastro and Gomis studied unitary and anti-unitary symmetries of Abelian Chern--Simons theories in terms of automorphisms of the anyon fusion group preserving, or complex-conjugating, the topological spin, and developed an explicit arithmetic analysis in terms of $K$-matrix presentations \cite{Delmas2021}. Here we obtained these results in the functorial framework of Abelian Chern--Simons TQFT.  We identify the unitary topological symmetry group canonically as
$$
\operatorname{Aut}^{\mathrm{br}}\!\bigl(\mathcal C(G_K,q_K)\bigr)
\cong
\mathcal{O}(G_K,q_K),
$$
and identify orientation reversal with the TQFT classified by $(G_K,q_K^{-1}),$ equivalently by the lattice with level $-K$. Thus the usual quantum-symmetry and time-reversal criteria are promoted to statements about equivalences of the TQFT functors.

An advantage of this formulation is that, once the theory is classified by its finite quadratic module $(G,q)$, the symmetry and orientation-reversal analysis does not depend on a particular lattice or $K$-matrix presentation. Different lattices realizing isomorphic finite quadratic modules define equivalent  Abelian Chern--Simons TQFTs \cite{Galviz4}, and their topological symmetry groups are correspondingly identified, up to conjugation by a chosen isometry of finite quadratic modules. In this sense, the $K$-matrix serves as a presentation of the theory, while the finite quadratic module is the presentation-independent invariant governing its topological symmetries.

\subsection*{Topological symmetry groups}

The classification theorem identifies the invariant of the theory with the finite quadratic module $(G_K,q_K)$.  It is therefore natural to extract the  group of topological relabelling symmetries from the pointed modular category $\mathcal C(G_K,q_K)$.
\begin{definition}
Let $(G,q)$ be a finite quadratic module.  We write
$$
\operatorname{Aut}^{\mathrm{br}}(\mathcal C(G,q))
$$
for the group of braided ribbon tensor autoequivalences of $\mathcal C(G,q)$, modulo braided ribbon monoidal natural isomorphism\footnote{The notation above is standard in the theory of braided tensor categories. For general background on braided fusion and modular categories, see \cite{DGNO2010,Joyal-Street}.}. For the Abelian Chern--Simons theory classified by $(G,q)$, we define its \emph{topological symmetry group} by
$$
\Sym\!\bigl(Z_{(G,q)}\bigr):=
\operatorname{Aut}^{\mathrm{br}}(\mathcal C(G,q)).
$$
Equivalently, for the lattice presentation $(\Lambda,K)$,
$$
\Sym\!\bigl(Z^{CS}_{\mathbb T,K}\bigr):=
\operatorname{Aut}^{\mathrm{br}}(\mathcal C(G_K,q_K)).
$$
\end{definition}

\begin{remark}
This is the  categorical notion of topological symmetry attached to the theory. It should not be confused with the group of monoidal natural automorphisms of the functor $Z^{CS}_{\mathbb T,K}$, which is a different $2$-categorical notion.
\end{remark}

\begin{definition}
Let $(G,q)$ be a finite quadratic module.  Its orthogonal group is
$$
\mathcal O(G,q):=\{\phi\in \Aut(G)\mid q(\phi(a))=q(a)\ \text{for all }a\in G\}.
$$
\end{definition}

\begin{theorem}
\label{thm:Eq-O}
Let $(G,q)$ be a finite quadratic module. Then there is a canonical group isomorphism
$$
\operatorname{Aut}^{\mathrm{br}}(\mathcal C(G,q))\cong \mathcal O(G,q).
$$
\end{theorem}

\begin{proof}
Choose a normalized pointed braided realization $
\mathcal C(G,q)\simeq \operatorname{Vec}^{\,\omega,c}_G,$ where $(\omega,c)$ is an Abelian $3$-cocycle representing the
Eilenberg--Mac Lane class determined by $q$. Thus the simple objects
are $\{X_a\}_{a\in G}$, with
$$
X_a\otimes X_b\simeq X_{a+b},
\qquad
\theta_{X_a}=q(a)\,\operatorname{id}_{X_a}.
$$

Let $F:\mathcal C(G,q)\rightarrow\mathcal C(G,q)$ be a braided ribbon tensor autoequivalence. Since the category is pointed, $F$ permutes the isomorphism classes of simple objects. Hence there is a unique bijection $\phi_F:G\to G$ such that
$$
F(X_a)\cong X_{\phi_F(a)}.
$$
Since $F$ is monoidal,
$$
F(X_{a+b})
\cong
F(X_a\otimes X_b)
\cong
F(X_a)\otimes F(X_b)
\cong
X_{\phi_F(a)+\phi_F(b)},
$$
and therefore $\phi_F(a+b)=\phi_F(a)+\phi_F(b).$ Thus $\phi_F\in\operatorname{Aut}(G)$. Since $F$ is ribbon, it preserves twists, and hence
$$
q(\phi_F(a))=q(a)\qquad \text{for all }a\in G.
$$
Therefore $\phi_F\in\mathcal O(G,q)$. Moreover, braided ribbon
monoidally isomorphic autoequivalences induce the same permutation of
simple objects, so this defines a homomorphism
$$
\Psi:
\operatorname{Aut}^{\mathrm{br}}(\mathcal C(G,q))
\longrightarrow
\mathcal O(G,q),
\qquad
[F]\longmapsto\phi_F.
$$
Conversely, let $\phi\in\mathcal O(G,q)$. Pulling back the chosen Abelian $3$-cocycle along $\phi$ gives $\phi^*(\omega,c).$ Its associated quadratic form is
$$
a\longmapsto c(\phi(a),\phi(a))
=q(\phi(a))
=q(a).
$$
Hence $\phi^*(\omega,c)$ and $(\omega,c)$ determine the same Eilenberg--Mac Lane quadratic class. By the classification of pointed braided categories by quadratic forms, they are Abelian-cohomologous. Consequently there exists a normalized $2$-cochain
$$
\eta_\phi:G\times G\longrightarrow\mathbb C^\times
$$
whose Abelian coboundary identifies $\phi^*(\omega,c)$ with $(\omega,c)$. Using $\eta_\phi$ as tensorator makes the relabelling
$$
X_a\longmapsto X_{\phi(a)}
$$
into a braided ribbon tensor autoequivalence, which we denote by $F_\phi$.

Although the cochain $\eta_\phi$ need not be unique, the resulting braided monoidal isomorphism class of $F_\phi$ is unique. Indeed, if $\eta_\phi$ and $\eta'_\phi$ are two choices, then their ratio
$$
\xi(a,b):=\eta'_\phi(a,b)\eta_\phi(a,b)^{-1}
$$
is a normalized $2$-cocycle. Since both tensorators satisfy the braided compatibility condition, one has
$$
\xi(a,b)=\xi(b,a).
$$
Thus $\xi$ is a symmetric normalized $2$-cocycle. For a finite Abelian group $G$ with coefficients in the divisible group $\mathbb C^\times$, every symmetric normalized $2$-cocycle is a coboundary. Hence the two choices of tensorator give braided monoidally isomorphic functors.

It follows that
$$
\Phi:
\mathcal O(G,q)
\longrightarrow
\operatorname{Aut}^{\mathrm{br}}(\mathcal C(G,q)),
\qquad
\phi\longmapsto[F_\phi],
$$
is well defined. It is a homomorphism: the composite $F_\phi\circ F_\psi$ induces the relabelling $\phi\circ\psi$, and by the uniqueness just proved its braided monoidal isomorphism class is $[F_{\phi\circ\psi}]$.

By construction,
$$
\Psi(\Phi(\phi))=\phi.
$$
Conversely, if $F$ is any braided ribbon tensor autoequivalence and $\phi=\phi_F$, then $F$ and $F_\phi$ induce the same relabelling of simple objects. The same uniqueness argument shows that they are braided ribbon monoidally isomorphic. Therefore
$$
\Phi(\Psi([F]))=[F].
$$
Thus $\Phi$ and $\Psi$ are mutually inverse group isomorphisms.
\end{proof}

\begin{corollary}
\label{cor:symmetry-group}
Let $(\Lambda,K)$ be an even, integral, nondegenerate lattice.  Then there is a canonical group isomorphism
$$
\Sym\!\bigl(Z^{CS}_{\mathbb T,K}\bigr)\cong \mathcal O(G_K,q_K)  .
$$
\end{corollary}

\begin{proof}
By definition,
$$
\Sym\!\bigl(Z^{CS}_{\mathbb T,K}\bigr)
:=
\operatorname{Aut}^{\mathrm{br}}(\mathcal C(G_K,q_K)).
$$
Applying Theorem~\ref{thm:Eq-O} with $(G,q)=(G_K,q_K)$ gives the result.
\end{proof}

\begin{remark}
Corollary~\ref{cor:symmetry-group} is a consequence of the classification theorem. In particular, if two lattice presentations $(\Lambda,K)$ and $(\Lambda',L)$ define equivalent extended Abelian Chern--Simons theories, then their finite quadratic modules are isomorphic:
$$
(G_K,q_K)\cong (G_L,q_L).
$$
Consequently their topological symmetry groups are isomorphic. More explicitly, after choosing an isometry
$$
f:(G_K,q_K)\xrightarrow{\sim}(G_L,q_L),
$$
one obtains a group isomorphism
$$
\mathcal O(G_K,q_K)  
\longrightarrow
\mathcal O(G_L,q_L),
\qquad
u\longmapsto fuf^{-1}.
$$
Via Corollary~\ref{cor:symmetry-group}, this induces an isomorphism between the corresponding topological symmetry groups. This identification depends on the chosen isometry $f$; without such a choice there is, in general, no preferred identification of the two symmetry groups.
\end{remark}

\subsection*{Orientation-reversed duality and orientation-reversal invariance}
We now describe the orientation-reversed theory.

\begin{definition}
Let $\mathcal C$ be a ribbon category.  We write $\mathcal C^{\mathrm{rev}}$ for the
reverse ribbon category: it has the same underlying monoidal category as $\mathcal C$,
but its braiding and twist are inverted.
Thus the braiding in $\mathcal C^{\mathrm{rev}}$ is $c^{-1}_{Y,X}$, and the twist is
$\theta^{-1}$, see for example \cite{Joyal-Street}.
\end{definition}

\begin{definition}
If $q:G\to U(1)$ is a quadratic form on a finite Abelian group $G$, we write
$$
q^{-1}(a):=q(a)^{-1}.
$$
\end{definition}

\begin{prop}
\label{prop:reverse-pointed}
Let $(G,q)$ be a finite quadratic module.  Then there is a braided ribbon equivalence
$$
\mathcal C(G,q)^{\mathrm{rev}}\simeq \mathcal C(G,q^{-1}).
$$
\end{prop}

\begin{proof}
Choose an Abelian $3$-cocycle representative $(\omega,c)$
for the pointed braided category $\mathcal C(G,q)$.
The reverse ribbon category has braiding
$$c^{\mathrm{rev}}_{X_a,X_b}=c^{-1}_{X_b,X_a}.$$
Its quadratic form is therefore
$$q^{\mathrm{rev}}(a)=c^{\mathrm{rev}}(a,a)=c(a,a)^{-1}=q(a)^{-1}.
$$
Hence the reverse pointed braided category has Eilenberg--Mac Lane quadratic form $q^{-1}$. By the classification of pointed braided categories by quadratic forms,
$$
\mathcal C(G,q)^{\mathrm{rev}}
\simeq
\mathcal C(G,q^{-1})
$$
as braided ribbon categories.
\end{proof}

\begin{theorem}
\label{thm:orientation-reversed-toral}
Let $(\Lambda,K)$ be an even, integral, nondegenerate lattice.  Then the orientation-reversed theory of the Abelian Chern--Simons theory $Z^{CS}_{\mathbb T,K}$ is the theory classified by $(G_K,q_K^{-1})$.  Equivalently,
$$
\bigl(Z^{CS}_{\mathbb T,K}\bigr)^{\mathrm{rev}}
\;\simeq\;
Z^{RT}_{\mathcal C(G_K,q_K^{-1})}.
$$
Since $(\Lambda,-K)$ is again even, integral, and nondegenerate, and $(G_{-K},q_{-K})\cong (G_K,q_K^{-1}),$ one may also write
$$
\bigl(Z^{CS}_{\mathbb T,K}\bigr)^{\mathrm{rev}}
\;\simeq\;
Z^{CS}_{\mathbb T,-K}.
$$
\end{theorem}

\begin{proof}
By Theorem~\ref{thm:cs-rt-equivalence}, $Z^{CS}_{\mathbb T,K}\simeq Z^{RT}_{\mathcal C(G_K,q_K)}.$ Passing to the orientation-reversed theory on the Reshetikhin--Turaev side corresponds to replacing the
ribbon category by its reverse ribbon category, so
$$
\bigl(Z^{CS}_{\mathbb T,K}\bigr)^{\mathrm{rev}}
\simeq
Z^{RT}_{\mathcal C(G_K,q_K)^{\mathrm{rev}}}.
$$
By Proposition~\ref{prop:reverse-pointed}, $\mathcal C(G_K,q_K)^{\mathrm{rev}}
\simeq
\mathcal C(G_K,q_K^{-1}),$ hence
$$
\bigl(Z^{CS}_{\mathbb T,K}\bigr)^{\mathrm{rev}}
\simeq
Z^{RT}_{\mathcal C(G_K,q_K^{-1})}.
$$

It remains to identify this with an Abelian Chern--Simons theory.
Since $-K\Lambda=K\Lambda$, the discriminant groups $G_{-K}$ and $G_K$ are
canonically identified.
Moreover,
$$
q_{-K}([x])
=
\exp\!\bigl(\pi i\,x^{\top}(-K)^{-1}x\bigr)
=
\exp\!\bigl(-\pi i\,x^{\top}K^{-1}x\bigr)
=
q_K([x])^{-1}.
$$
Thus
$$
(G_{-K},q_{-K})\cong (G_K,q_K^{-1}).
$$
Applying the classification theorem and the extended equivalence theorem once more gives
$$
Z^{RT}_{\mathcal C(G_K,q_K^{-1})}
\simeq
Z^{CS}_{\mathbb T,-K},
$$
which proves the claim.
\end{proof}

\begin{corollary}
\label{cor:orientation-reversal-invariant}
Let $(\Lambda,K)$ be an even, integral, nondegenerate lattice.
The following are equivalent:
\begin{enumerate}[label=\rm(\roman*)]
\item The Abelian Chern--Simons theory $Z^{CS}_{\mathbb T,K}$ is orientation-reversal
      invariant, that is,
      $$
      \bigl(Z^{CS}_{\mathbb T,K}\bigr)^{\mathrm{rev}}
      \simeq
      Z^{CS}_{\mathbb T,K}.
      $$
\item The finite quadratic modules $(G_K,q_K)$ and $(G_K,q_K^{-1})$ are isomorphic.
\item The Abelian Chern--Simons theories $Z^{CS}_{\mathbb T,K}$ and
      $Z^{CS}_{\mathbb T,-K}$ are equivalent as symmetric monoidal extended TQFTs.
\end{enumerate}
\end{corollary}

\begin{proof}
By Theorem~\ref{thm:orientation-reversed-toral}, the orientation-reversed theory of $Z^{CS}_{\mathbb T,K}$ is the theory classified by $(G_K,q_K^{-1})$, equivalently the theory with level $-K$. Hence $\rm(i)$ is equivalent to $\rm(iii)$. By the classification theorem, two Abelian Chern--Simons extended TQFTs are equivalent if and only if their finite quadratic modules are isomorphic. Applying this to $(\Lambda,K)$ and $(\Lambda,-K)$ shows that $\rm(iii)$ is equivalent to $\rm(ii)$. Therefore all three statements are equivalent.
\end{proof}

\section{Turaev–Viro realizations of Abelian Chern--Simons theory}
\label{sec:tv}
Now let us consider the connection between Abelian Chern--Simons theory and extended Turaev--Viro theory \cite{TV92,BW96,Balsam10,TV2017}.  The first is the center case, where the pointed modular category attached to an Abelian theory is itself a Drinfeld center.  The second is the doubled case, where the center is realized as a Deligne tensor product with the orientation-reversed category. We use the following terminology throughout this section. For a finite quadratic module $(G,q)$, we say that $(G,q)$ is hyperbolic if it admits a Lagrangian subgroup $L\subset G$, that is, $q|_L=1$ and $L=L^\perp$. Equivalently, $(G,q)$ is Witt trivial\footnote{Thus ``hyperbolic'' is used here in the metabolic, or Witt-trivial, sense; it does not require a splitting as $A\oplus A^\vee$ with the standard hyperbolic quadratic form.}.

\begin{theorem}
\label{thm:center-criterion}
Let $(\Lambda,K)$ be an even, integral, nondegenerate lattice, and write $\mathcal C_K:=\mathcal C(G_K,q_K).$ Then the following are equivalent.

\begin{enumerate}[label=\rm(\alph*)]
\item There exists a spherical fusion category $\mathcal A$ such that $\mathcal C_K \cong \mathcal Z(\mathcal A)$ as braided ribbon categories.

\item The finite quadratic module $(G_K,q_K)$ is hyperbolic, i.e. there exists a
subgroup $L\subset G_K$ such that $q_K|_L=1,
$ and $ L=L^\perp .$
Equivalently, $L$ is isotropic and $|L|^2=|G_K|$.

\item There exist a finite Abelian group $B$ and a class
$\beta\in H^3(B;\mathbb C^\times)$ such that $\mathcal C_K \cong \mathcal Z(\VecCat_B^{\,\beta})$ as braided ribbon categories.
\end{enumerate}
\end{theorem}

\begin{proof}
The implication $\rm(c)\Rightarrow (a)$ is immediate, since every pointed fusion
category $\VecCat_B^{\,\beta}$ is spherical. For $\rm(a)\Rightarrow (b)$, assume that $\mathcal C_K \cong \mathcal Z(\mathcal A)$ for some fusion category $\mathcal A$. A nondegenerate braided fusion category is a
Drinfeld center if and only if it contains a Lagrangian algebra. Hence $\mathcal C_K$
contains a Lagrangian algebra $R$.

Now $\mathcal C_K$ is pointed. In a pointed modular category $\mathcal C(G_K,q_K)$,
connected Étale algebras are classified by isotropic subgroups $H\subset G_K$, and the
Lagrangian algebras are exactly those for which $H$ is Lagrangian. Therefore $R$
corresponds to a subgroup $L\subset G_K$ satisfying $q_K|_L=1,$ and $
L=L^\perp .$ Thus $(G_K,q_K)$ is hyperbolic.

For $\rm(b)\Rightarrow (c)$, let $L\subset G_K$ be a Lagrangian subgroup. In the pointed modular category $\mathcal C(G_K,q_K)$, the corresponding algebra $R(L):=\bigoplus_{\ell\in L}\ell$ is a Lagrangian algebra. For pointed modular categories, a choice of Lagrangian subgroup determines a pointed fusion category $\VecCat_{G_K/L}^{\,\beta}$ and a braided ribbon equivalence
$$
\mathcal C(G_K,q_K)\cong \mathcal Z(\VecCat_{G_K/L}^{\,\beta})
$$
for some $\beta\in H^3(G_K/L;\mathbb C^\times)$. Hence $\mathcal C_K$ is a Drinfeld
center of a spherical fusion category, proving $\rm(c)$, and therefore also $\rm(a)$.
\end{proof}

\subsection*{The center case}

\begin{theorem}
\label{thm:center-case}
Let $(\Lambda,K)$ be an even, integral, nondegenerate lattice. Suppose $\mathcal C(G_K,q_K)$  admits a Lagrangian subgroup. Equivalently, by Theorem \ref{thm:center-criterion}, there exists a spherical fusion category $\mathcal{A}$ such that $\mathcal C(G_K,q_K)\cong \mathcal Z(\mathcal A)$ as braided ribbon categories. Then there is a symmetric monoidal natural isomorphism of extended TQFTs
$$
Z^{CS}_{\mathbb T,K}\cong Z^{TV}_{\mathcal A}.
$$
\end{theorem}

\begin{proof}
By Theorem~\ref{thm:cs-rt-equivalence}, there is a symmetric monoidal natural
isomorphism
$$
Z^{CS}_{\mathbb T,K}\cong Z^{RT}_{\mathcal C(G_K,q_K)}.
$$
Transporting the theory along the braided ribbon equivalence gives
$$
Z^{RT}_{\mathcal C(G_K,q_K)}\cong Z^{RT}_{\mathcal Z(\mathcal A)}.
$$
Finally, by the equivalence of extended Turaev--Viro theory with extended Reshetikhin--Turaev theory \cite[Theorem 3.1]{Balsam10}, see also \cite{BK2010,TV2017}, there is a symmetric monoidal natural isomorphism
$$
Z^{TV}_{\mathcal A}\cong Z^{RT}_{ \mathcal Z(\mathcal A)}.
$$
Composing the three natural isomorphisms gives the claim.
\end{proof}

Under the symmetric monoidal natural isomorphism of Theorem~\ref{thm:center-case}, the Turaev--Viro theory $Z^{TV}_{\mathcal A}$ inherits its state spaces, bordism vectors, partition functions, and surgery evaluations from the Abelian Chern--Simons theory $Z^{CS}_{\mathbb T,K}$ of Section~\ref{sec:Abelian-CS}. In particular, for every connected closed oriented surface $\Sigma_g$,
$$
\dim Z^{TV}_{\mathcal A}(\Sigma_g)=|G_K|^g=|\det K|^g,
$$
and if $M=M_L$ is presented by integral surgery with regular part $L_{\mathrm{reg}}$, then
$$
Z^{TV}_{\mathcal A}(M_L)
=
|G_K|^{m_M}\,
|\det(L_{\mathrm{reg}})|^{-n/2}
\sum_{[x]\in \mathbb Z^{\rho n}/(L_{\mathrm{reg}}\otimes I_n)\mathbb Z^{\rho n}}
\exp\!\Bigl(\pi i\,x^\top(L_{\mathrm{reg}}^{-1}\otimes K)x\Bigr).
$$
All remaining formulas are obtained by applying the corresponding Abelian Chern--Simons formulas of Section~\ref{sec:Abelian-CS}. The $Z^{TV}_{\mathcal A}(M)$ can be directly expressed in terms of the Reidemeister torsion which leads to the same equation from Definition~\ref{def:closed-partition}.

\subsection*{Multicomponent BF as the hyperbolic center case}

We now apply the center criterion to the multicomponent Abelian BF theories of Section~\ref{sec:bf}. Let $N\in M_n(\mathbb Z)$ with $\det N\neq0$, and recall
$$
K_{BF}(N)=
\begin{pmatrix}
0&N\\
N^\top&0
\end{pmatrix}.
$$
Set $A_N:=\mathbb Z^n/N\mathbb Z^n.$ The discriminant group decomposes as

$$
G_{K_{BF}(N)}
\cong
\frac{\mathbb Z^n}{N\mathbb Z^n}
\oplus
\frac{\mathbb Z^n}{N^\top\mathbb Z^n}.
$$

Since

$$
K_{BF}(N)^{-1}
=
\begin{pmatrix}
0&(N^\top)^{-1}\\
N^{-1}&0
\end{pmatrix},
$$

the associated quadratic form is

$$
q_{K_{BF}(N)}([u],[v])
=
\exp\!\left(2\pi i\,v^\top N^{-1}u\right).
$$

The resulting pairing identifies $\frac{\mathbb Z^n}{N^\top\mathbb Z^n}
\cong A_N^\vee,$ and hence
$$
\bigl(G_{K_{BF}(N)},q_{K_{BF}(N)}\bigr)
\cong
\bigl(A_N\oplus A_N^\vee,q_{\mathrm{hyp}}\bigr),
\qquad
q_{\mathrm{hyp}}(a,\chi)=\chi(a).$$
Thus the BF discriminant form is the standard hyperbolic metric group associated with $A_N$. In particular, $A_N\oplus\{1\}$ is a Lagrangian subgroup, and therefore
$$
\mathcal C\bigl(G_{K_{BF}(N)},q_{K_{BF}(N)}\bigr)
\cong
\mathcal Z(\VecCat_{A_N}).
$$

\begin{corollary}
\label{cor:bf-tv}
Let $N\in M_n(\mathbb Z)$ with $\det N\neq0$, and let $A_N:=\mathbb Z^n/N\mathbb Z^n.$ Then there is a symmetric monoidal natural isomorphism of extended TQFTs

$$
Z_N^{BF}
=
Z^{CS}_{\mathbb T_{BF},K_{BF}(N)}
\cong
Z^{TV}_{\VecCat_{A_N}}.
$$

Thus multicomponent Abelian BF theory with coupling matrix $N$ is naturally identified with the Turaev--Viro theory of $\VecCat_{A_N}$.
\end{corollary}

\begin{proof}
By Proposition~\ref{prop:bf-embedding}, the BF theory with coupling matrix $N$ is the Abelian Chern--Simons theory associated with $K_{BF}(N)$. As shown above, its discriminant finite quadratic module is the standard hyperbolic metric group $\bigl(A_N\oplus A_N^\vee,q_{\mathrm{hyp}}\bigr),$ whose pointed modular category is braided ribbon equivalent to
$\mathcal Z(\VecCat_{A_N})$. The result therefore follows from Theorem~\ref{thm:center-case}.
\end{proof}

For $n=1$ and $N=(k)$, one recovers

$$
K_{\mathrm{hyp}}(k)=
\begin{pmatrix}
0&k\\
k&0
\end{pmatrix},
\qquad
A_N\cong\mathbb Z/k\mathbb Z,
$$
and hence
$
Z^{CS}_{U(1)^2,K_{\mathrm{hyp}}(k)}
\cong
Z^{TV}_{\VecCat_{\mathbb Z/k\mathbb Z}},
$ recovering the rank-one BF theory as the basic hyperbolic example.
\subsection*{The doubled modular case}

Let $(\Lambda,K)$ be an even, integral, nondegenerate lattice and write $\mathcal C_K:=\mathcal C(G_K,q_K).$ Since $\mathcal C_K$ is modular, its Drinfeld center is canonically braided equivalent to its modular double
$$
\mathcal Z(\mathcal C_K)\simeq \mathcal C_K\boxtimes \mathcal C^{\mathrm{rev}}_K,
$$
see for example \cite{DGNO2010}. Here $\mathcal C^{\mathrm{rev}}_K$ denotes the orientation-reversed modular category. On the lattice side, the orthogonal direct sum $K\oplus (-K)$ is again even, integral, and nondegenerate, and its discriminant quadratic module is the orthogonal sum of $(G_K,q_K)$ and $(G_K,q_K^{-1})$. Consequently the pointed modular category attached to $K\oplus(-K)$ is braided equivalent to $\mathcal C_K\boxtimes \mathcal C^{\mathrm{rev}}_K.$

\begin{theorem}
\label{thm:double-case}
Let $(\Lambda,K)$ be an even, integral, nondegenerate lattice, and let $\mathbb T= \mathrm t/\Lambda$. Then there is a symmetric monoidal natural isomorphism of TQFTs
$$
Z^{TV}_{\mathcal C(G_K,q_K)}
\cong
Z^{CS}_{\mathbb T\times \mathbb T,\,K\oplus(-K)}.
$$
\end{theorem}

\begin{proof}
By the extended center theorem for Turaev--Viro theory,
$$
Z^{TV}_{\mathcal C_K}\cong Z^{RT}_{\mathcal Z(\mathcal C_K)}.
$$
Since $\mathcal C_K$ is modular, there is a braided ribbon equivalence $\mathcal Z(\mathcal C_K)\simeq \mathcal C_K\boxtimes \mathcal C^{\mathrm{rev}}_K.$ Thus
$$
Z^{RT}_{\mathcal Z(\mathcal C_K)}\cong
Z^{RT}_{\mathcal C_K\boxtimes \mathcal C^{\mathrm{rev}}_K}.
$$
The discriminant form of the lattice $K\oplus(-K)$ is the orthogonal sum of $(G_K,q_K)$ and $(G_K,q_K^{-1})$, so the pointed modular category attached to $K\oplus(-K)$ is braided equivalent to $\mathcal C_K\boxtimes \mathcal C^{\mathrm{rev}}_K$.  Applying the extended equivalence theorem again yields
$$
Z^{RT}_{\mathcal C_K\boxtimes \mathcal C^{\mathrm{rev}}_K}
\cong Z^{TV}_{\mathcal C_K}\cong
Z^{CS}_{\mathbb T\times \mathbb T,\,K\oplus(-K)}.
$$
Composing these natural isomorphisms proves the theorem.
\end{proof}
This doubled realization is the TQFT counterpart of the string-net picture of Levin and Wen \cite{LevinWen05}: orientation-reversal invariant doubled topological orders arise naturally from a condensation/string-net construction. In the present Abelian setting, the Turaev--Viro center construction realizes precisely the doubled modular data $\mathcal C_K \boxtimes \mathcal C_K^{\mathrm{rev}}$, corresponding to $Z^{CS}_{\mathbb T,K}\otimes Z^{CS}_{\mathbb T,-K}$. Under the symmetric monoidal natural isomorphism of Theorem~\ref{thm:double-case}, the Turaev--Viro theory $Z^{TV}_{\mathcal C_K}$ inherits its TQFT data from the Abelian Chern--Simons theory $Z^{CS}_{\mathbb T\times \mathbb T,\;K\oplus(-K)}.$ In particular, for every connected closed oriented surface $\Sigma_g$,
$$
\dim Z^{TV}_{\mathcal C_K}(\Sigma_g)
=
|G_{K\oplus(-K)}|^g
=
|\det K|^{2g},
$$
and if $M=M_L$ is presented by integral surgery with regular part $L_{\mathrm{reg}}$, then
$$
Z^{TV}_{\mathcal C_K}(M_L)
=
|\det K|^{2m_M}\,
|\det(L_{\mathrm{reg}})|^{-n}
\sum_{[x]\in \mathbb Z^{2\rho n}/(L_{\mathrm{reg}}\otimes I_{2n})\mathbb Z^{2\rho n}}
\exp\!\Bigl(\pi i\,x^\top(L_{\mathrm{reg}}^{-1}\otimes (K\oplus(-K)))x\Bigr).
$$

\begin{remark}

The doubled surgery expression admits a useful factorization. Using

$$
K\oplus(-K)
=
\begin{pmatrix}
K&0\\
0&-K
\end{pmatrix},
$$

the finite summation group decomposes canonically as

$$
\frac{\mathbb Z^{2\rho n}}
     {(L_{\mathrm{reg}}\otimes I_{2n})\mathbb Z^{2\rho n}}
\cong
\frac{\mathbb Z^{\rho n}}
     {(L_{\mathrm{reg}}\otimes I_n)\mathbb Z^{\rho n}}
\oplus
\frac{\mathbb Z^{\rho n}}
     {(L_{\mathrm{reg}}\otimes I_n)\mathbb Z^{\rho n}}.
$$

Thus we may write a summation class as $x=(u,v)$, where $u$ belongs to the $K$-sector and $v$ to the $-K$-sector. Then

$$
x^\top
\bigl(
L_{\mathrm{reg}}^{-1}\otimes(K\oplus(-K))
\bigr)x
=
u^\top(L_{\mathrm{reg}}^{-1}\otimes K)u
-
v^\top(L_{\mathrm{reg}}^{-1}\otimes K)v.
$$

Consequently,

$$
\begin{aligned}
&\exp\!\Bigl(
\pi i\,x^\top
[L_{\mathrm{reg}}^{-1}\otimes(K\oplus(-K))]x
\Bigr)
\\
&\qquad =
\exp\!\Bigl(
\pi i\,u^\top(L_{\mathrm{reg}}^{-1}\otimes K)u
\Bigr)
\exp\!\Bigl(
-\pi i\,v^\top(L_{\mathrm{reg}}^{-1}\otimes K)v
\Bigr).
\end{aligned}
$$

Hence the finite quadratic Gauss sum factors as

$$
\sum_{u,v}
e^{\pi i(Q_K(u)-Q_K(v))}
=
\left(
\sum_u e^{\pi iQ_K(u)}
\right)
\left(
\sum_v e^{-\pi iQ_K(v)}
\right),
$$

where $Q_K(u)=u^\top(L_{\mathrm{reg}}^{-1}\otimes K)u.$ The second factor is the complex conjugate of the first. The normalization factors factorize in the same way:

$$
|\det K|^{2m_M}
|\det(L_{\mathrm{reg}})|^{-n}
=
\left(
|\det K|^{m_M}
|\det(L_{\mathrm{reg}})|^{-n/2}
\right)^2.
$$

Therefore the full doubled partition function satisfies

$$
Z^{CS}_{\mathbb T\times\mathbb T,\,K\oplus(-K)}(M)
=
Z^{CS}_{\mathbb T,K}(M)\,
Z^{CS}_{\mathbb T,-K}(M).
$$

Using the orientation-reversal identification
$$
Z^{CS}_{\mathbb T,-K}(M)
=
Z^{CS}_{\mathbb T,K}(-M),
$$
and unitarity,
$$
Z^{CS}_{\mathbb T,K}(-M)
=
\overline{Z^{CS}_{\mathbb T,K}(M)},
$$
one obtains

$$
Z^{TV}_{\mathcal C_K}(M)=Z^{CS}_{\mathbb T\times\mathbb T,\,K\oplus(-K)}(M)
=
\left|Z^{CS}_{\mathbb T,K}(M)\right|^2.
$$

Thus the factorization of the finite Gauss sum directly reflects the interpretation of the $K\oplus(-K)$ theory as the orientation-reversal invariant double $Z^{CS}_{\mathbb T,K}\otimes Z^{CS}_{\mathbb T,-K}.$
\end{remark}

All other state--space, bordism-vector, and closed-partition formulas are obtained by applying the Abelian Chern--Simons formulas of Section~\ref{sec:Abelian-CS}. The previous surgery expression has a direct formulation in terms of Reidemeister torsion due to our construction\footnote{For the non-abelian hyperbolic case, see, for instance \cite{ming26}, which connects adjoint Reidemeister torsion to the one-loop term in the asymptotic expansion of Turaev–Viro-type invariants.}. Set
$$
A_M:=\operatorname{Tors}H^2(M;\Lambda),
\qquad
a_M:=\#A_M,
$$
and write
$$
\Lambda_{\mathrm d}:=\Lambda\oplus\Lambda,
\qquad
\mathrm t_{\mathrm d}:=\mathrm t\oplus\mathrm t,
\qquad
K_{\mathrm d}:=K\oplus(-K).
$$
Then
$$
\operatorname{Tors}H^2(M;\Lambda_{\mathrm d})
\cong A_M\oplus A_M,
\qquad
|\det K_{\mathrm d}|=|\det K|^2.
$$

Applying Definition~\ref{def:closed-partition} to the doubled lattice, together with Theorem~\ref{thm:double-case}, gives
$$
Z^{TV}_{\mathcal C_K}(M)
=
\frac{|\det K|^{2m_M}}{a_M^2}
\sum_{(p,q)\in A_M\oplus A_M}
\int_{\mathcal M_{M,(p,q)}(\mathbb T\times\mathbb T)}
\sigma^{K_{\mathrm d}}_{M,(p,q)}
\bigl(T_M(\mathrm t_{\mathrm d})\bigr)^{1/2}.
$$
Here $\sigma^{K_{\mathrm d}}_{M,(p,q)}$ is the exponentiated Chern--Simons action in the sector $(p,q)$, and $\bigl(T_M(\mathrm t_{\mathrm d})\bigr)^{1/2}$ denotes the translation-invariant density induced by the square root of Reidemeister--Ray--Singer torsion. This expression can also be written entirely in terms of the undoubled theory. There is a canonical product decomposition
$$
\mathcal M_{M,(p,q)}(\mathbb T\times\mathbb T)
\cong
\mathcal M_{M,p}(\mathbb T)\times
\mathcal M_{M,q}(\mathbb T).
$$
Additivity of the action and reversal of the coupling give
$$
\sigma^{K_{\mathrm d}}_{M,(p,q)}
=
\sigma^K_{M,p}\boxtimes\overline{\sigma^K_{M,q}}.
$$
Meanwhile, multiplicativity of Reidemeister torsion under direct
sums of coefficient systems gives, under the corresponding
determinant-line identification,
$$
T_M(\mathrm t\oplus\mathrm t)
=
T_M(\mathrm t)\otimes T_M(\mathrm t).
$$
Writing
$$
d\nu_{M,p}:=
\left.\bigl(T_M(\mathrm t)\bigr)^{1/2}
\right|_{\mathcal M_{M,p}(\mathbb T)},
$$
the density on the product component is therefore $d\nu_{M,p}\boxtimes d\nu_{M,q}$. Consequently,
$$
Z^{TV}_{\mathcal C_K}(M)
=
\frac{|\det K|^{2m_M}}{a_M^2}
\sum_{p,q\in A_M}
\int_{\mathcal M_{M,p}(\mathbb T)\times
      \mathcal M_{M,q}(\mathbb T)}
\sigma^K_{M,p}\,
\overline{\sigma^K_{M,q}}\,
d\nu_{M,p}\boxtimes d\nu_{M,q}.
$$

This formula distinguishes the classical and measure contributions, i.e the Chern--Simons sections carry the phases, whereas Reidemeister torsion supplies the density on the continuous flat modes. The change $K\mapsto-K$ conjugates the phase without changing this positive density due to the orientation-reversal duality from previous Section \ref{sec:symmetries}.

The determinant-line interpretation is essential when cohomology is nonzero, $T_M(\mathrm t)$ is used here through its density. In particular, the formula remains applicable when $b_1(M)>0$. For a surgery presentation, $a_M=|\det L_{\mathrm{reg}}|^n$, the above finite Gauss sum is the surgery evaluation of this torsion-density expression. Finally, the product formula directly gives
$$
Z^{TV}_{\mathcal C_K}(M)
=
\left|
\frac{|\det K|^{m_M}}{a_M}
\sum_{p\in A_M}
\int_{\mathcal M_{M,p}(\mathbb T)}
\sigma^K_{M,p}\,d\nu_{M,p}
\right|^2
=
\left|Z^{CS}_{\mathbb T,K}(M)\right|^2.
$$

\begin{remark}
The Turaev--Viro identifications above also give a lattice realization of the corresponding Abelian Chern--Simons theories through Levin--Wen string-net models \cite{LevinWen05}. Whenever a theory is identified with $Z^{TV}_{\mathcal A}$, it inherits the associated Levin--Wen Hamiltonian realization. Conversely, Liu and Zhao show that suitable axioms on ground-state wave functions reconstruct the underlying unitary fusion category and recover the corresponding Levin--Wen ground states \cite{LiuZhao2026}. Thus, while the Turaev--Viro construction above and the defect construction below take fusion-categorical data as input, the wave-function perspective shows how such data can instead be reconstructed from microscopic states. In particular, the hyperbolic BF case admits a string-net realization based on $\mathrm{Vec}_{\mathbb Z/k\mathbb Z}$, while the general doubled theory with level $K\oplus(-K)$ is realized by the corresponding Turaev--Viro string-net model.
\end{remark}

\section{Defects in Abelian Chern--Simons theory}
\label{sec:defects}
The equivalence of Theorem~\ref{thm:cs-rt-equivalence} allows us to equip Abelian Chern--Simons theory with categorical defects. We first develop the functorial level of this construction, and then classify the defects of this theory.

We retain the notation of Section~\ref{sec:Abelian-CS}
$$
G_K=\Lambda^*/K\Lambda,\qquad
q_K([x])=\exp\!\bigl(\pi i x^\top K^{-1}x\bigr),\qquad
\mathcal C_K=\mathcal C(G_K,q_K).
$$
All categories are finite semisimple and $\mathbb C$-linear, and we use
the standard unitary ribbon realization of $\mathcal C_K$. Set
$$
b_K(x,y)=\frac{q_K(x+y)}{q_K(x)q_K(y)},\qquad
H^\perp=\{x\in G_K\mid b_K(x,h)=1\text{ for every }h\in H\}.
$$
A subgroup $H$ is \emph{isotropic} if $q_K|_H=1$. It is
\emph{Lagrangian} if, in addition, $H=H^\perp$. Nondegeneracy gives
\begin{equation}
\label{eq:defects-perp-order}
|H|\,|H^\perp|=|G_K|.
\end{equation}
Consequently an isotropic subgroup is Lagrangian exactly when
$|H|^2=|G_K|$. Maximality among isotropic subgroups alone does not
imply this equality.

\subsection{The defect functor}

Let $\mathcal D_K$ denote the defect datum of the categorical RT
construction for $\mathcal C_K$. We write
$$
\widehat{\mathrm{Bord}}^{\mathrm{def}}_3(\mathcal D_K)
$$
for its category of decorated stratified surfaces and $3$-bordisms;
the hat records the anomaly extension. The datum specifies admissible
labels and incidences, rather than allowing arbitrary labels on
arbitrary stratifications. In the single-bulk construction of
\cite{CRS2019}, surfaces are labelled by $\Delta$-separable symmetric
Frobenius algebras in $\mathcal C_K$, and incident lines by the
appropriate multimodules with cyclic compatibility. Point insertions
are included using the compatible junction spaces. Ordinary Wilson
lines retain their ribbon data. The associated  higher
defect structure is discussed in \cite{CMS2020}.

\begin{theorem}
\label{thm:defect-functor}
Fix the symmetric monoidal natural isomorphism
$$
\Phi_K:Z^{RT}_{\mathcal C_K}\xrightarrow{\sim}
Z^{CS}_{\mathbb T,K}
$$
of Theorem~\ref{thm:cs-rt-equivalence}. There is a symmetric monoidal
functor
\begin{equation}
\label{eq:defects-functor}
\widetilde Z^{CS}_{\mathcal C_K}:
\widehat{\mathrm{Bord}}^{\mathrm{def}}_3(\mathcal D_K)
\longrightarrow\mathrm{Vect}^{\mathrm{fd}}_{\mathbb C}
\end{equation}

together with a symmetric monoidal natural isomorphism
$$
\eta_K^{CS}:\iota_K^*\widetilde Z^{CS}_{\mathcal C_K}
\xrightarrow{\sim}Z^{CS}_{\mathbb T,K},
$$
where $\iota_K$ includes the undecorated bulk bordisms. It induces defect evaluations, state spaces, and gluing laws for the single bulk phase $\mathcal C_K$ and its chosen categorical defect datum. Fully gapped boundaries and walls between different bulk phases are described categorically below, using folding where appropriate.
\end{theorem}

\begin{proof}
\cite[Theorem~5.8]{CRS2019}
supplies a defect TQFT functor $\widetilde Z^{RT}_{\mathcal C_K}$.
On undecorated bordisms its construction reduces to RT theory,
giving an identification
$$
\eta_K^{RT}:\iota_K^*\widetilde Z^{RT}_{\mathcal C_K}
\xrightarrow{\sim}Z^{RT}_{\mathcal C_K}.
$$
Define the functor and its bulk identification separately by
\begin{equation}
\label{eq:defects-transport}
\widetilde Z^{CS}_{\mathcal C_K}
:=\widetilde Z^{RT}_{\mathcal C_K},\qquad
\eta_K^{CS}:=\Phi_K\circ\eta_K^{RT}.
\end{equation}
The composite is symmetric monoidal and natural. In particular, for
an undecorated bordism $N:\Sigma_0\to\Sigma_1$,
$$
Z^{CS}_{\mathbb T,K}(N)=
\eta_{K,\Sigma_1}^{CS}\,
\widetilde Z^{RT}_{\mathcal C_K}(N)\,
(\eta_{K,\Sigma_0}^{CS})^{-1}.
$$
Functoriality and the anomaly-corrected gluing law therefore hold
with the prescribed Chern--Simons bulk normalization. The defect
extension is the pair
$(\widetilde Z^{CS}_{\mathcal C_K},\eta_K^{CS})$; the tilde symbol
in \eqref{eq:defects-functor} denotes its functor, not this pair.
\end{proof}

\subsection{Connected \'etale algebras in the Abelian case}
Let us choose  a pointed braided convention, see Figure\ref{fig:bulk-lines} for a representation of bulk Wilson lines, 
$$
\mathcal C_K\simeq\operatorname{Vec}^{\,\omega_K,c_K}_{G_K},
\qquad X_x\otimes X_y\simeq X_{x+y},
$$

where the associator from $(X_x\otimes X_y)\otimes X_z$ to
$X_x\otimes(X_y\otimes X_z)$ has scalar $\omega_K(x,y,z)$.
The braiding has scalar $c_K(x,y)$, with
$c_K(x,x)=q_K(x)$ and $c_K(x,y)c_K(y,x)=b_K(x,y)$.

\begin{figure}[htbp]
\centering
\begin{tikzpicture}[defbase]
\fill[defbulk!5] (0,0) rectangle (11.6,3.1);
\node at (5.8,2.65) {$\mathcal C_K=\mathcal C(G_K,q_K)$};
\foreach \a/\b in {2.6/x,5.8/y,9/z}{
  \draw[defline] (\a,2.05)--(\a,.45);
  \node[right=4pt] at (\a,1.3) {$W_{\b}\;\leftrightarrow\;X_{\b}$};
}
\node[defnote] at (5.8,-.4) {$x,y,z\in G_K=\Lambda^*/K\Lambda$};
\end{tikzpicture}
\caption{Bulk Wilson sectors. A charge modulo monopole screening labels
both a topological Wilson line $W_x$ and a simple object $X_x$ of
$\mathcal C_K$. The arrows specify line orientations; the shaded region
denotes the bulk phase.}
\label{fig:bulk-lines}
\end{figure}
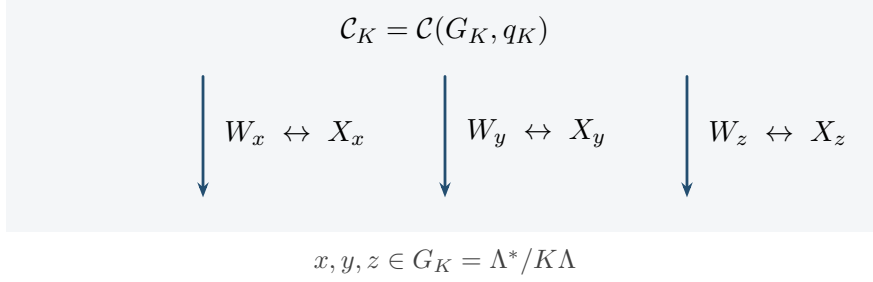

\begin{prop}
\label{prop:defects-etale}
Isotropic subgroups $H\subset G_K$ are in bijection with
algebra-isomorphism classes of connected \'etale algebras\footnote{Here \emph{connected \'etale} means a commutative separable
algebra $A$ with $\dim\Hom(\mathbf1,A)=1$. Such an algebra is
\emph{Lagrangian} if 
$\operatorname{dim}(A)^2=\operatorname{dim}(\mathcal C_K)$. Since all fusion categories considered here are unitary, categorical dimensions agree with Frobenius–Perron (PF) dimensions. We therefore write $\dim$ for this common dimension, see \cite[Sec.8.2]{ENO2005}.} in
$\mathcal C_K$. An explicit representative is
$$
A_{H,\mu}=\bigoplus_{h\in H}X_h,
$$
with unit in degree $0$ and multiplication scalar
$\mu(h_1,h_2)\in U(1)$. Its cochain satisfies
\begin{align}
\mu(h_1,h_2)\mu(h_1+h_2,h_3)
&=\omega_K(h_1,h_2,h_3)
\mu(h_2,h_3)\mu(h_1,h_2+h_3),
\label{eq:defects-associativity}\\
\mu(h_1,h_2)&=c_K(h_1,h_2)\mu(h_2,h_1).
\label{eq:defects-commutativity}
\end{align}
Its isomorphism class is denoted $A_H$. The category of local right
modules is
\begin{equation}
\label{eq:defects-local-modules}
(\mathcal C_K)^0_{A_H}\simeq
\mathcal C(H^\perp/H,\overline q_K),\qquad
\overline q_K(x+H)=q_K(x)\quad(x\in H^\perp).
\end{equation}
In particular,
$$
\operatorname{dim}(A_H)=|H|,\qquad
\operatorname{dim}\bigl((\mathcal C_K)^0_{A_H}\bigr)
=\frac{|G_K|}{|H|^2}.
$$
\end{prop}

\begin{proof}
Let $A$ be connected \'etale. The category $(\mathcal C_K)_A$
of right $A$-modules is a fusion category under $\otimes_A$, and
the free-module functor
$$
F_A:\mathcal C_K\to(\mathcal C_K)_A,\qquad X\mapsto X\otimes A
$$
is tensor, with its standard central structure
\cite[Section~3]{DMNO2013}. Since $X_x$ is invertible, $F_A(X_x)$
is invertible and simple. The unit $A$ is simple by connectedness.
Adjunction gives
$$
\Hom_{(\mathcal C_K)_A}(F_A(X_x),A)
\cong\Hom_{\mathcal C_K}(X_x,A).
$$
Thus every homogeneous multiplicity of $A$ is $0$ or $1$, and
its support is
$$
H=\{x\in G_K\mid F_A(X_x)\simeq A\}.
$$
Tensoriality shows that $H$ is a subgroup. For $h\in H$, a nonzero
map $X_h\to A$ induces an isomorphism $X_h\otimes A\to A$ by
multiplication. Hence all homogeneous multiplication coefficients
on $H$ are nonzero. Associativity and braided commutativity give
\eqref{eq:defects-associativity} and
\eqref{eq:defects-commutativity}. Taking $h_1=h_2=h$ in the latter
shows $q_K(h)=c_K(h,h)=1$.

Conversely, if $q_K|_H=1$, the restriction of the Abelian
$3$-cocycle $(\omega_K,c_K)$ to $H$ has trivial quadratic class.
The classification of pointed braided categories therefore
provides a cochain $\mu$ trivializing this restriction, in exactly
the convention \eqref{eq:defects-associativity}--
\eqref{eq:defects-commutativity}; see \cite{Joyal-Street,DGNO2010}.
In the resulting trivialized braided category, $A_{H,\mu}$ is
the graded group algebra of $H$. It is connected and commutative,
and is separable by the usual averaging splitting in characteristic
zero. Transporting the splitting proves separability in the original
category.

Suppose $\mu'$ is another compatible multiplication on $H$.
The ratio $\xi=\mu'/\mu$ satisfies
$$
\xi(h_1,h_2)\xi(h_1+h_2,h_3)
=\xi(h_2,h_3)\xi(h_1,h_2+h_3),\qquad
\xi(h_1,h_2)=\xi(h_2,h_1).
$$
It is therefore a symmetric normalized $2$-cocycle. The associated
central extension of $H$ by $\mathbb C^\times$ is Abelian and
splits because $\mathbb C^\times$ is divisible. Thus there is a
normalized $1$-cochain $t$ with
$$
\xi(h_1,h_2)=\frac{t(h_1)t(h_2)}{t(h_1+h_2)}.
$$
The diagonal map with component $t(h)^{-1}\id_{X_h}$ is an algebra
isomorphism $A_{H,\mu}\to A_{H,\mu'}$. Different supports cannot
give isomorphic underlying objects. This proves the bijection.

For the residual bulk, write $F_x=X_x\otimes A_H$. Every simple
right module is a summand of a free module, since separability
splits the action map. Each $F_x$ is already simple, and
adjunction gives
$$
\Hom_{A_H}(F_x,F_y)
\cong\Hom_{\mathcal C_K}(X_x,X_y\otimes A_H).
$$
This space is one-dimensional precisely when $x-y\in H$, and zero otherwise. Hence the simple right modules are indexed by $G_K/H$, with fusion induced by addition. Such a module is local exactly when its action is unchanged by the double braiding with $A_H$, which here is the condition $b_K(x,h)=1$ for every $h\in H$. Its twist is then $q_K(x)$. This descends to $H^\perp/H$, since $q_K(x+h)=q_K(x)$ for $x\in H^\perp$ and $h\in H$. he induced pairing is nondegenerate; its radical is $(H^\perp)^\perp/H=H/H$. This proves \eqref{eq:defects-local-modules}, including the dimension formula.
\end{proof}

The algebra $A_H$ specifies the condensation, whereas $(\mathcal C_K)^0_{A_H}$ specifies the surviving deconfined bulk. Different embeddings of condensation  can lead to equivalent residual bulk categories, so the latter alone is not the classification invariant for a condensation in a fixed bulk.

\subsection{Classification of boundaries and domain walls}

We use the bulk--boundary correspondence in its precise categorical form \cite{Fuchs:2013,DMNO2013}. A fully gapped elementary boundary of a modular category $\mathcal C$ is a fusion category $\mathcal W_a$ together with a compatible ribbon equivalence
$$
\widetilde F_{\to a}:\mathcal C\xrightarrow{\sim}
\mathcal Z(\mathcal W_a).
$$
The bulk transport functor is
$F_{\to a}=\varphi_{\mathcal W_a}\circ\widetilde F_{\to a}$,
where $\varphi_{\mathcal W_a}$ forgets the half-braiding.
Equivalence of boundary data must intertwine these bulk functors.
The center equivalence is an essential condition: centrality of
$F_{\to a}$ alone does not imply it.

The correspondence assigns to a Lagrangian algebra $A$ the boundary
category $\mathcal C_A$ with its central free-module functor.
Conversely, the right adjoint of $F_{\to a}$ applied to
$\mathbf1_{\mathcal W_a}$ recovers $A$. For a general connected
\'etale algebra one has
$$
\mathcal Z(\mathcal C_A)\simeq
\mathcal C\boxtimes(\mathcal C_A^0)^{\mathrm{rev}};
$$
the boundary is fully gapped to the vacuum exactly when
$\mathcal C_A^0\simeq\mathrm{Vec}_{\mathbb C}$. Figure~\ref{fig:gapped-boundary} illustrates how condensed bulk lines end on a gapped boundary and how other lines become boundary excitations.

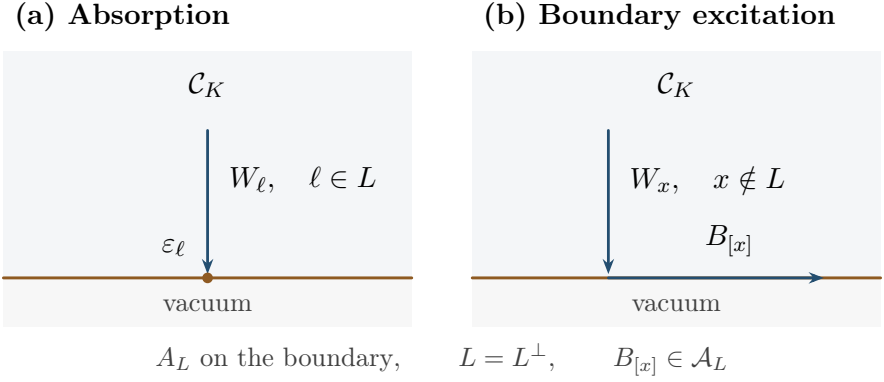
\begin{figure}[H]
\centering
\begin{tikzpicture}[defbase]
\foreach \a in {0,6.2}{
 \fill[defbulk!5] (\a,0) rectangle +(5.4,3);
 \fill[black!3] (\a,-.65) rectangle +(5.4,.65);
 \draw[defedge] (\a,0)--+(5.4,0);
 \node at (\a+2.7,2.55) {$\mathcal C_K$};
 \node[defnote] at (\a+2.7,-.36) {vacuum};
}
\node[defpanel] at (0,3.45) {(a) Absorption};
\node[defpanel] at (6.2,3.45) {(b) Boundary excitation};
\draw[defline] (2.7,1.95)--(2.7,.05);
\node[right=4pt] at (2.7,1.3) {$W_\ell,\quad\ell\in L$};
\node[defend] at (2.7,0) {};
\node[above left=4pt] at (2.7,0) {$\varepsilon_\ell$};
\draw[defline] (8,1.95)--(8,.05);
\node[right=4pt] at (8,1.3) {$W_x,\quad x\notin L$};
\draw[defline] (8,0)--(10.85,0);
\node[above=4pt] at (9.6,0) {$B_{[x]}$};
\node[defnote] at (5.8,-1.08)
 {$A_L\text{ on the boundary},\qquad L=L^\perp,\qquad B_{[x]}\in\mathcal A_L$};
\end{tikzpicture}
\caption{A Lagrangian boundary. (a) A condensed charge $\ell\in L$
can end in the boundary vacuum. (b) A charge $x\notin L$ induces a
nontrivial boundary line $B_{[x]}$, indexed by $[x]\in G_K/L$;
it does not disappear into the boundary vacuum.}
\label{fig:gapped-boundary}
\end{figure}

For a wall from $\mathcal C_K$ to $\mathcal C_{K'}$, folding gives
the bulk

\begin{equation}
\label{eq:defects-folding}
\mathcal C_K\boxtimes\mathcal C_{K'}^{\mathrm{rev}}
\simeq\mathcal C(G_K\oplus G_{K'},q_K\oplus q_{K'}^{-1}).
\end{equation}

We use $K'$ for the second lattice form, reserving $L$ for a Lagrangian subgroup. At the TQFT level the folded bulk is $Z^{CS}_{\mathbb T,K}\otimes Z^{CS}_{\mathbb T',-K'}$. For physical transmission labels we use the convention that the folded object corresponding to an incoming charge $x$ and outgoing charge $y$ is $X_x\boxtimes(X_y^*)^{\mathrm{rev}}$. Its simple objects  are $(x,-y)$. Changing from $(x,-y)$ to $(x,y)$ is an isometry because $q_{K'}(-y)=q_{K'}(y)$. Subgroups describing transmission below use the $(x,y)$ convention. Figure~\ref{fig:domain-wall} shows the folding convention for a transmitted line and its absorption by the folded wall algebra.

\begin{figure}[H]
\centering
\begin{tikzpicture}[defbase]
\node[defpanel] at (0,3.8) {(a) Unfolded interface};
\fill[defbulk!5] (0,0) rectangle (2.25,3.25);
\fill[defother!6] (2.25,0) rectangle (4.5,3.25);
\draw[defedge] (2.25,0)--(2.25,3.25);
\node at (1.05,2.7) {$\mathcal C_1$};
\node at (3.4,2.7) {$\mathcal C_2$};
\draw[defline] (.25,1.5)--(2.2,1.5);
\draw[defsecond] (2.3,1.5)--(4.25,1.5);
\node[above=5pt] at (1.05,1.5) {$X_x$};
\node[above=5pt] at (3.4,1.5) {$X_y$};
\node[defend] at (2.25,1.5) {};
\node[fill=white,inner sep=2pt] at (2.25,.48) {$M$};
\node[defnote] at (2.25,-.45) {$(x,y)\in M$};
\draw[-{Stealth[length=2.5mm]},line width=.9pt] (4.9,1.6)--(6.5,1.6);
\node[defnote,align=center] at (5.7,2.25) {fold the\\second phase};
\node[defpanel] at (6.95,3.8) {(b) Folded boundary};
\fill[defbulk!5] (6.95,0) rectangle (11.6,3.25);
\fill[black!3] (6.95,-.75) rectangle (11.6,0);
\draw[defedge] (6.95,0)--(11.6,0);
\node at (9.275,2.75) {$\mathcal C_1\boxtimes\mathcal C_2^{\mathrm{rev}}$};
\draw[defline] (8.25,2.12)--(8.25,.9)--(9.25,.04);
\draw[defsecond] (10.3,2.12)--(10.3,.9)--(9.3,.04);
\node[left=3pt] at (8.25,1.5) {$X_x$};
\node[right=3pt] at (10.3,1.5) {$(X_y^*)^{\mathrm{rev}}$};
\node[defend] at (9.275,0) {};
\node[above=4pt] at (10.7,0) {$A_M$};
\node[defnote] at (9.275,-.42) {vacuum};
\end{tikzpicture}
\caption{Folding a fully gapped domain wall. Transmission of $x$ to $y$
is expressed by absorption of
$X_x\boxtimes(X_y^*)^{\mathrm{rev}}$ into $A_M$ in the folded bulk.
Here $M$ uses transmission coordinates $(x,y)$; the corresponding
simple-object degree in the folded category is $(x,-y)$.
The superscript $\mathrm{rev}$ reverses the braiding, while $X_y^*$
accounts for the outgoing-line orientation.}
\label{fig:domain-wall}
\end{figure}
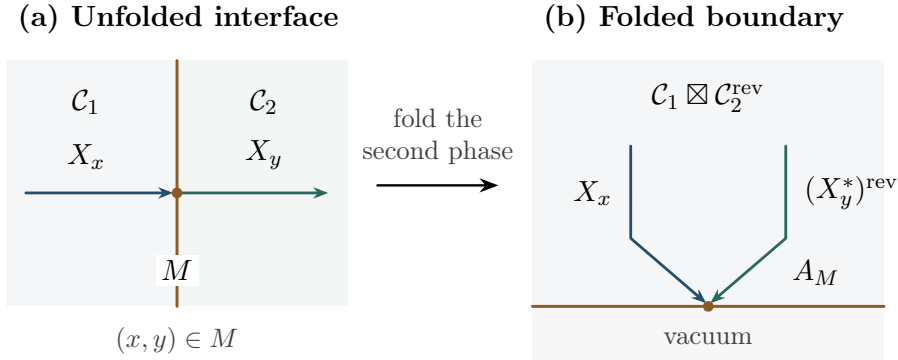

\begin{theorem}
\label{thm:defect-classification}
The following statements hold:
\begin{enumerate}[label=\rm(\arabic*)]
\item Connected condensation data in the $K$-bulk are classified by
isotropic subgroups $H\subset G_K$. Their residual bulk is
$\mathcal C(H^\perp/H,\overline q_K)$.

\item Elementary fully gapped boundaries of the $K$-bulk are
classified by Lagrangian subgroups $L\subset G_K$. The associated
boundary category is
$$
\mathcal A_L:=(\mathcal C_K)_{A_L},\qquad
\mathcal Z(\mathcal A_L)\simeq\mathcal C_K,
$$
with the center equivalence induced by the bulk.
Such boundaries exist if and only if $\mathcal C_K$ is Witt-trivial.

\item Elementary fully gapped walls from the $K$-bulk to the
$K'$-bulk are classified by Lagrangian subgroups
$$
M\subset G_K\oplus G_{K'},\qquad
(q_K\oplus q_{K'}^{-1})|_M=1,\qquad M=M^\perp.
$$
Writing $A_M$ for the corresponding folded algebra in the
transmission convention, their wall-line categories are
$$
\mathcal W_M=
(\mathcal C_K\boxtimes\mathcal C_{K'}^{\mathrm{rev}})_{A_M}.
$$
Such walls exist if and only if
$[\mathcal C_K]=[\mathcal C_{K'}]$ in the Witt group.

\item Invertible walls are the relabelling walls associated with
isometries $u:(G_K,q_K)\xrightarrow{\sim}(G_{K'},q_{K'})$.
Their transmission support is
$$
\Gamma_u=\{(x,u(x))\mid x\in G_K\},
$$
so explicitly their folded algebra has underlying object
$$
A_{\Gamma_u}=
\bigoplus_{x\in G_K}X_x\boxtimes(X_{u(x)}^*)^{\mathrm{rev}}.
$$
For self-walls, fusion gives the group
$\mathcal O(G_K,q_K)$ at the level of equivalence classes:
$$
\mathfrak D_v\circ\mathfrak D_u\simeq\mathfrak D_{v\circ u},
\qquad
\mathfrak D_u^{-1}\simeq\mathfrak D_{u^{-1}}.
$$
Here $\mathfrak D_v\circ\mathfrak D_u$ means first crossing the
$u$-wall and then the $v$-wall.
\end{enumerate}
\end{theorem}

\begin{proof}
Part~\rm(1) is Proposition~\ref{prop:defects-etale}. The
bulk--boundary correspondence identifies the elementary boundary
data in \rm(2) with Lagrangian connected \'etale algebras.
By that proposition and \eqref{eq:defects-perp-order}, $A_H$ is
Lagrangian exactly when $H=H^\perp$. In that case its local-module
category is trivial and the displayed center equivalence follows.
Conversely, a boundary center equivalence recovers a Lagrangian
algebra by the right-adjoint construction. This also gives precisely
the Witt-triviality criterion of Theorem~\ref{thm:center-criterion}.

Folding identifies walls with boundaries for
\eqref{eq:defects-folding}. Applying \rm(2) to this nondegenerate
pointed category proves \rm(3). The folded category is Witt-trivial
exactly when its two factors represent opposite Witt classes.

For \rm(4), an isometry $u$ gives
$$
(q_K\oplus q_{K'}^{-1})(x,u(x))=1,\qquad
|\Gamma_u|^2=|G_K|\,|G_{K'}|.
$$
Thus $\Gamma_u$ is Lagrangian. The pointed braided classification lifts $u$ to a ribbon equivalence $F_u:\mathcal C_K\to\mathcal C_{K'}$ as in Section~\ref{sec:symmetries}. The corresponding wall is the transparent interface with its outgoing bulk identification changed by $F_u$, and therefore has inverse induced by $F_u^{-1}$. Conversely, an invertible wall transports bulk lines invertibly and preserves their fusion, braiding and twist. It thus induces a ribbon equivalence and hence an isometry of finite quadratic modules. Its transmission is the graph of that isometry. Composing the transport functors proves the stated fusion law. In particular, invertibility here follows from the equivalence $F_u$, not merely from the Lagrangian property of the graph.
\end{proof}

For $u=\id$, $A_{\Gamma_u}=\bigoplus_x
X_x\boxtimes(X_x^*)^{\mathrm{rev}}$ is the transparent-wall
algebra. Its wall-line category is the underlying fusion category
$\mathcal C_K$. The two bulk actions lift to its center by the
braiding and inverse braiding, respectively, giving the canonical
equivalence
$\mathcal C_K\boxtimes\mathcal C_K^{\mathrm{rev}}
\simeq\mathcal Z(\mathcal C_K)$.

\subsection{Boundary associators, defect lines and junctions}

The classification above concerns objects up to equivalence.
Their line and junction data retain categorical information.
First, the proof of Proposition~\ref{prop:defects-etale} gives a
concrete description of the boundary fusion category:
\begin{equation}
\label{eq:defects-boundary-associator}
\mathcal A_L\simeq
\operatorname{Vec}^{\,\beta_L}_{G_K/L},\qquad
[\beta_L]\in H^3(G_K/L;U(1)).
\end{equation}
Indeed, its invertible simple objects are the free modules
$X_x\otimes A_L$, indexed by $x+L$, and their product is induced
by addition. Choose representatives and nonzero fusion maps for
these simple objects. The associator then has scalar $\beta_L$;
the pentagon is its $3$-cocycle equation. Changing the fusion maps
changes $\beta_L$ by a coboundary. An algebra isomorphism between
two choices of $A_{L,\mu}$ induces a tensor equivalence of their
module categories preserving the quotient labels. Hence
$[\beta_L]$ is determined by the embedded subgroup $L$ and the
bulk quadratic data; it is not an independent boundary label.
The explicit gauge representative is considered in
Section~\ref{sec:gauge-defect}.

To describe boundary-changing lines, fix a reference Lagrangian
subgroup $L_0$ and put $\mathcal A=\mathcal A_{L_0}$, with its
specified equivalence $\mathcal C_K\simeq\mathcal Z(\mathcal A)$.
The boundary bicategory is modeled by finite semisimple
$\mathcal A$-module categories \cite{Fuchs:2013}. Elementary
boundaries correspond to indecomposable module categories
$\mathcal M_a$. In this model,
\begin{equation}
\label{eq:defects-boundary-lines}
\mathcal W_{a,b}=
\operatorname{Fun}_{\mathcal A}(\mathcal M_a,\mathcal M_b).
\end{equation}
An object of $\mathcal W_{a,b}$ is a boundary-changing line;
a module natural transformation is a point junction between
such lines. Line fusion is composition of module functors. We
choose the endofunctor tensor convention so that evaluation
defines a right action: the tensor product of endofunctors $P,Q$
is $Q\circ P$. With this convention
$\mathcal W_a=\operatorname{End}_{\mathcal A}(\mathcal M_a)$
acts on the left of $\mathcal W_{a,b}$ by precomposition and
$\mathcal W_b$ acts on the right by postcomposition.

In the algebra presentation of the same module-category
bicategory, boundary objects are separable algebras in the
\emph{reference fusion category} $\mathcal A$, lines are bimodule
objects, and junctions are bimodule maps. These presenting
algebras are not the Lagrangian algebras $A_L$ in the bulk center.
This distinction prevents counting extra boundary lines by using
all bulk $A_L$--$A_{L'}$ bimodules in place of
\eqref{eq:defects-boundary-lines}.

For a mixed vertex with bulk line $U$ and boundary lines $W_1,W_2$
on a fixed boundary $a$, the junction space is
$$
\Hom_{\mathcal W_a}
\bigl(F_{\to a}(U)\otimes W_1,W_2\bigr).
$$
Folding gives the corresponding descriptions for wall-changing
lines and wall junctions. In particular, a fixed reference wall
with line category $\mathcal W_d$ identifies the bicategory of
walls between those two bulks with the appropriate
$\mathcal W_d$-module-category bicategory, with the normalized
trace data understood.

There are two composition operations to distinguish. Within a
fixed boundary or wall bicategory, line fusion is the composition
just described, equivalently a relative tensor product of
bimodule objects. For surface-wall fusion, when the bulks have
chosen center presentations
$\mathcal C_i\simeq\mathcal Z(\mathcal A_i)$, a wall is modeled
by an $\mathcal A_i$--$\mathcal A_j$ bimodule category, and
composition is
$$
\mathcal M_{12}\boxtimes_{\mathcal A_2}\mathcal M_{23}.
$$
For fusion categories without braiding the opposite tensor
category is denoted $\mathcal A_2^{\mathrm{op}}$; the notation
$\mathcal C_2^{\mathrm{rev}}$ is reserved for reversal of braiding
and twist in a modular category. For general Witt-equivalent
bulks, wall fusion is supplied by the multiphase categorical
construction, with its central bulk actions. It cannot be
replaced without justification by a tensor product of arbitrary
$\mathcal C_i$--$\mathcal C_j$ bimodule categories. A fusion of
elementary noninvertible walls can decompose, so its support as
a set-theoretic relation alone need not specify the result.

All these compositions and junction evaluations are inherited by \eqref{eq:defects-transport}. Changing a pointed presentation changes their representatives through categorical equivalences. Thus, relative to the fixed categorical defect construction and normalizations, the finite quadratic module determines the available condensation, boundary and wall structures together with their line categories and coherence  information.

\section{Alterfold theory, defects and Chern--Simons theory}
\label{sec:Alterfold}

In this section we explain the Alterfold origin of defects in Abelian Chern--Simons theory constructed in the previous section. Alterfold theory naturally contains the Turaev--Viro theory of a spherical fusion category $A$ and the Reshetikhin--Turaev theory of its Drinfeld center $\mathcal Z(\mathcal A)$ as sub-TQFTs \cite{Alterfold1,Alterfold2}. Its Morita-context enhancement also supplies the geometric and categorical structures that control boundaries, walls, condensations, modular invariants, $\alpha$-induction, and full centers. Since the Abelian Chern--Simons theories considered in this paper are identified with pointed Reshetikhin--Turaev theories, Alterfold theory is the natural topological source of the center and doubled Abelian bulk sectors, and therefore of the corresponding defect Abelian theories.

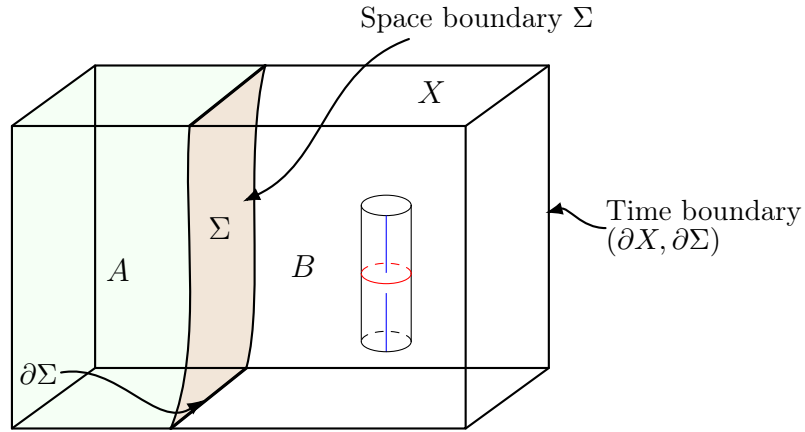
\begin{figure}[ht]
\hspace*{3cm}
\begin{tikzpicture}[scale=1.0, line cap=round, line join=round, >=Latex]
  \begin{scope}[shift={(0.4,0)}]

  %---------------------------------
  % Outer 3-manifold M (rectangular prism)
  %---------------------------------
  \coordinate (A)  at (0,0);
  \coordinate (B)  at (6,0);
  \coordinate (C)  at (6,4);
  \coordinate (D)  at (0,4);

  \coordinate (A') at (1.1,0.8);
  \coordinate (B') at (7.1,0.8);
  \coordinate (C') at (7.1,4.8);
  \coordinate (D') at (1.1,4.8);

  %---------------------------------
  % Separating surface Sigma
  %---------------------------------
  \coordinate (P1) at (2.10,0.00);
  \coordinate (P2) at (2.35,4.00);
  \coordinate (Q1) at (3.10,0.80);
  \coordinate (Q2) at (3.35,4.80);

  %---------------------------------
  % Color the visible A and B regions
  %---------------------------------
  \fill[green!4] (A) -- (D) -- (D') -- (A') -- cycle;
  \fill[white]    (B) -- (C) -- (C') -- (B') -- cycle;

  \fill[green!4]
    (A) -- (P1)
    .. controls (2.55,1.10) and (2.25,2.75) .. (P2)
    -- (D) -- cycle;

  \fill[white]
    (P1) -- (B) -- (C) -- (P2)
    .. controls (2.25,2.75) and (2.55,1.10) .. (P1);

  \fill[green!4] (D) -- (D') -- (Q2) -- (P2) -- cycle;
  \fill[white]    (P2) -- (Q2) -- (C') -- (C) -- cycle;

  %---------------------------------
  % Draw separating surface Sigma
  %---------------------------------
  \fill[brown!20!white]
    (P1)
      .. controls (2.55,1.10) and (2.25,2.75) .. (P2)
      -- (Q2)
      .. controls (3.05,3.60) and (3.35,1.75) .. (Q1)
      -- cycle;

  \draw[thick]
    (P1) .. controls (2.55,1.10) and (2.25,2.75) .. (P2);
  \draw[thick]
    (Q1) .. controls (3.35,1.75) and (3.05,3.60) .. (Q2);
  \draw[thick] (P1)--(Q1);
  \draw[thick] (P2)--(Q2);

  % visible boundary arcs of \partial\Sigma on \partial X
  \draw[very thick] (P1)--(Q1);
  \draw[very thick] (P2)--(Q2);

    %---------------------------------
  % Draw outer frame
  %---------------------------------
  \draw[thick] (A)--(B)--(C)--(D)--cycle;
  \draw[thick] (A')--(B')--(C')--(D')--cycle;
  \draw[thick] (A)--(A');
  \draw[thick] (B)--(B');
  \draw[thick] (C)--(C');
  \draw[thick] (D)--(D');

  \begin{scope}[shift={(4.95,2.05)}, xscale=0.55, yscale=0.45]

    % top ellipse
    \begin{scope}[shift={(0,2)}]
      \draw (0,0) [partial ellipse=0:360:0.6 and 0.3];
    \end{scope}

    % side walls
    \draw (-0.6, 2)--(-0.6, 0);
    \draw ( 0.6, 2)--( 0.6, 0);
    \draw (-0.6,-2)--(-0.6, 0);
    \draw ( 0.6,-2)--( 0.6, 0);

    % blue strand
    \draw[blue] (0,-2.3)--(0,1.7);

    % white break at equator
    \draw[line width=0.18cm,white] (0,-0.18)--(0,-0.38);

    % red equatorial ellipse
    \draw[red]        (0,0) [partial ellipse=180:360:0.6 and 0.3];
    \draw[red,dashed] (0,0) [partial ellipse=0:180:0.6 and 0.3];

    % bottom ellipse
    \begin{scope}[shift={(0,-2)}]
      \draw[dashed] (0,0) [partial ellipse=0:180:0.6 and 0.3];
      \draw        (0,0) [partial ellipse=180:360:0.6 and 0.3];
    \end{scope}

  \end{scope}

  %---------------------------------
  % Region labels
  %---------------------------------
  \node at (1.40,2.10) {$A$};
  \node at (3.85,2.15) {$B$};

  %---------------------------------
  % Object labels
  %---------------------------------
  \node at (5.55,4.45) {$X$};
  \node at (2.75,2.65) {$\Sigma$};

  %---------------------------------
  % Annotation arrows
  %---------------------------------
  \draw[->, thick] (5.25,5.15) to[out=-160,in=25] (3.05,3.00);
  \node[align=center] at (6.15,5.40) {\small Space boundary $\Sigma$};

  \draw[->, thick] (7.85,2.65) to[out=180,in=15] (7.05,2.85);
  \node[align=left] at (9.15,2.65)
    {\small Time boundary\\[-1mm]\small $(\partial X,\partial\Sigma)$};

  \draw[->, thick] (0.65,0.74) to[out=10,in=-148] (2.55,0.38);
  \node[align=center] at (0.35,0.72) {\small $\partial\Sigma$};
  \end{scope}
\end{tikzpicture}
\caption{\small A 3-alterfold with time boundary. Here $X$ is an oriented compact $3$-manifold with boundary, and $\Sigma\subset X$ is an embedded oriented surface meeting $\partial X$ transversely. The surface $\Sigma$ separates $X\setminus \Sigma$ into connected components colored by $A$ and $B$. The separating surface $\Sigma$ is the space boundary, while the pair $(\partial X,\partial \Sigma)$ is the time boundary. The interior of the tube is colored by $A$, and the exterior is colored by $B$; the tube carries additional Alterfold data, that is, the black strings indicate the boundary of the torus, the red circle denotes the Kirby color and the blue strand denotes an object label in the center sector.}\label{fig:alterfold}
\end{figure}

\subsection{A brief review of Alterfold theory}

Let $\mathcal A$ be a spherical fusion category. A $3$-Alterfold is, roughly speaking, a closed oriented $3$-manifold $M$ equipped with an embedded separating surface $\Sigma \subset M$, such that the connected components of $M\setminus \Sigma$ are alternately colored by two labels, usually denoted $A$ and $B$. One may further decorate $\Sigma$ by tensor diagrams. The resulting decorated closed $3$-alterfolds admit a partition function invariant under local topological moves. To pass from this closed picture to the TQFT framework, one considers instead $3$-alterfolds with time boundary, in which $X$ is allowed to have boundary and $\Sigma$ meets $\partial X$ transversely; Figure \ref{fig:alterfold} illustrates this cobordism version.

This distinction is essential for the TQFT picture. The pair $(\partial X,\partial \Sigma)$ gives the boundary data used to compose alterfold cobordisms, while the surface $\Sigma$ inside $X$ carries extra information that does not appear in ordinary Chern--Simons theory. If we restrict to the ordinary $B$-sector, we forget the additional $A$-colored bulk data and are left only with the underlying $3$-manifold $X$, which is precisely the spacetime for the corresponding Abelian Chern--Simons theory. In the figure, the vertical cylinder represents the tube category sector, namely the center data relevant in the pointed case.

When $\mathcal{A}$ is unitary, the resulting TQFT is unitary. The same formalism extends further to Morita contexts and $2$-categorical decorations, where the multi-colored geometry of the separating surface encodes relations between Morita equivalent fusion categories. Then

$$
\begin{tikzcd}[column sep=large, row sep=large]
\mathrm{Cob}
  \arrow[r, hook]
  \arrow[dr, "RT_{\mathcal Z(\mathcal A)}"']
&
\mathrm{\mathcal ACob}_{2+1}^{\mathcal A}
  \arrow[d, "\mathbb V_{\mathcal A}"]
&
\mathrm{Cob}
  \arrow[l, hook']
  \arrow[dl, "TV_{\mathcal A}"]
\\
& \mathrm{Vect}_{\mathbb C}
\end{tikzcd}
$$
The previous diagram expresses the content of \cite[Theorem~\ref{thm:center-criterion}]{Alterfold1}. The category $\mathrm{\mathcal ACob}^{\mathcal A}_{2+1}$ is the $2+1$-dimensional Alterfold cobordism category, and
$$
\mathbb V_{\mathcal A}\colon \mathrm{\mathcal ACob}_{2+1}^{\mathcal A}\longrightarrow \mathrm{Vect}_{\mathbb C}
$$
is the Alterfold TQFT. The two copies of $\mathrm{Cob}$ denote the ordinary cobordism category, embedded into the Alterfold category as the undecorated one-color sectors. The theorem states that, after restricting $\mathbb V$ to these ordinary sectors, one recovers the Reshetikhin--Turaev theory $Z^{RT}_{\mathcal Z(\mathcal A)}$ and the Turaev--Viro theory $Z^{TV}_{\mathcal A}$. Thus Alterfold theory provides a single ambient TQFT in which both $Z^{RT}_{\mathcal Z(\mathcal A)}$ and $Z^{TV}_{\mathcal A}$ appear as naturally identified subtheories.  For the resulting $2+1$-dimensional Alterfold TQFT, two structural facts will be used repeatedly below.

\begin{enumerate}  \renewcommand{\labelenumi}{\roman{enumi}.}
\item
The tube category of Alterfold theory is a topological model for the Drinfeld center
$\mathcal Z(\mathcal A)$. Thus the center appears intrinsically in the Alterfold formalism rather
than being added externally, see Figure \ref{fig:alterfold}.

\item
The restriction of $\mathbb V_{\mathcal A}$ to the ordinary $B$-colored subcategory
recovers both the Turaev--Viro TQFT of $\mathcal A$ and the Reshetikhin--Turaev TQFT of
$\mathcal Z(\mathcal A)$. More precisely,
$$
\mathbb V_{\mathcal A}\big|_{\mathrm{Cob}^B}
\;\cong\;
Z^{TV}_{\mathcal A}
\;\cong\;
Z^{RT}_{\mathcal Z(\mathcal A)}.
$$
This is the precise Alterfold realization of the Turaev--Viro/Reshetikhin--Turaev center
correspondence \cite[Theorem 5.1]{Alterfold1}.
\end{enumerate}

In particular, Alterfold theory is not merely another construction of $Z^{TV}_{\mathcal A}$. Rather, it is a larger topological framework in which $Z^{TV}_{\mathcal A}$, $Z^{RT}_{\mathcal Z(\mathcal A)}$, and more refined structures arising from modular fusion categories and Morita contexts appear simultaneously \cite{Alterfold1,Alterfold2}. Thus the bridge to Alterfold theory is immediate at the RT level: one should compare the pointed modular category $\mathcal C(G_K,q_K)$ with the center sector that Alterfold theory produces, as in Section \ref{sec:tv}.

\subsection{The Alterfold realization of the Abelian bulk theory}

The key observation is that Alterfold theory always reduces canonically to an RT theory of the form $Z^{RT}_{\mathcal Z(\mathcal A)}$, and therefore reduces to Abelian Chern--Simons theory exactly in those cases where this center is pointed and agrees with the modular category $\mathcal C(G_K,q_K)$.

\begin{prop}
\label{prop:Alterfold-cs-reduction}
Let $(\Lambda,K)$ be an even, integral, nondegenerate lattice, and let $(G_K,q_K)$ be its associated finite quadratic module. Assume that $(G_K,q_K)$ is hyperbolic, equivalently that it admits a Lagrangian subgroup. Then, by Theorem \ref{thm:center-criterion}, there exists a spherical fusion category $\mathcal A$ such that $\mathcal Z(\mathcal A)\cong \mathcal C(G_K,q_K)$ as braided ribbon categories. For any such choice of $\mathcal A$, the restriction of Alterfold theory to the ordinary $B$-colored subcategory is naturally isomorphic to Abelian Chern--Simons theory:
$$
\mathbb V_{\mathcal A}\big|_{\mathrm{Cob}^B}
\;\cong\;
Z^{RT}_{\mathcal Z(\mathcal A)}
\;\cong\;
Z^{RT}_{\mathcal C(G_K,q_K)}
\;\cong\;
Z^{CS}_{\mathbb T,K}.
$$
\end{prop}

\begin{proof}
The hyperbolic, equivalently Lagrangian, hypothesis is precisely the hypothesis under which Theorem \ref{thm:center-criterion} applies. Hence there exists a spherical fusion category $\mathcal A$ with $\mathcal Z(\mathcal A)\cong \mathcal C(G_K,q_K)$ as braided ribbon categories. The first isomorphism is the Alterfold realization of the Turaev--Viro/Reshetikhin--Turaev correspondence:

$$
\mathbb V_{\mathcal A}\big|_{\mathrm{Cob}^B}\cong Z^{RT}_{\mathcal Z(\mathcal A)}
$$
by \cite[Theorem 5.1]{Alterfold1}. The second is induced by the chosen braided equivalence of modular categories. The third is exactly the Abelian Chern--Simons/Reshetikhin--Turaev equivalence proved in \cite{Galviz3}.
\end{proof}

The hypothesis in Proposition \ref{prop:Alterfold-cs-reduction} is essential. The proposition gives the precise sense in which a single Abelian Chern--Simons bulk is realized directly as an ordinary Alterfold $B$-sector: this happens only when the pointed modular category $\mathcal C(G_K,q_K)$ is itself a Drinfeld center, equivalently, by Theorem \ref{thm:center-criterion}, only when $(G_K,q_K)$ is hyperbolic.

For a general finite quadratic module $(G_K,q_K)$, the pointed modular
category
$$
\mathcal C_K:=\mathcal C(G_K,q_K)
$$
need not itself be a Drinfeld center. Hence one must distinguish two
Alterfold realizations. In the center case, if there exists a spherical
fusion category $\mathcal A$ such that
$$
\mathcal Z(\mathcal A)\simeq \mathcal C_K,
$$
then the ordinary $B$-sector of Alterfold theory realizes the single-copy
Abelian Chern--Simons theory:
$$
\mathbb V_{\mathcal A}\big|_{\mathrm{Cob}^B}
\simeq
Z^{RT}_{\mathcal Z(\mathcal A)}
\simeq
Z^{RT}_{\mathcal C_K}
\simeq
Z^{CS}_{\mathbb T,K}.
$$

In the general pointed case, the canonical Alterfold realization is instead
the doubled one. Since $\mathcal C_K$ is a spherical fusion category and is
modular, there is a braided ribbon equivalence
$$
\mathcal Z(\mathcal C_K)
\simeq
\mathcal C_K\boxtimes \mathcal C_K^{\mathrm{rev}}.
$$
Therefore, taking the Alterfold input to be $\mathcal C_K$, one obtains
$$
\mathbb V_{\mathcal C_K}\big|_{\mathrm{Cob}^B}
\simeq
Z^{RT}_{\mathcal Z(\mathcal C_K)}
\simeq
Z^{RT}_{\mathcal C_K\boxtimes \mathcal C_K^{\mathrm{rev}}}
\simeq
Z^{CS}_{\mathbb T,K}\otimes Z^{CS}_{\mathbb T,-K}
\simeq
Z^{CS}_{\mathbb T\times\mathbb T,\,K\oplus(-K)}.
$$
Equivalently, the finite quadratic module of the doubled theory is
$$
(G_K,q_K)\oplus (G_K,q_K^{-1}).
$$
Thus a general Abelian Chern--Simons theory is not asserted to be directly realized as a single Alterfold center sector. Outside the hyperbolic case, Alterfold theory canonically realizes its doubled orientation-reversal invariant theory. The single-copy and doubled realizations therefore come from different center inputs: $\mathcal A$ in the center case and $\mathcal C_K$ in the canonical doubled case.

The above reduction admits a concrete expression on state spaces. Let $\Sigma_g$ be a closed surface of genus $g$. In Alterfold theory, the $B$-sector state space admits a pants-decomposition basis whose labels are simple objects of the center $\mathcal Z(\mathcal A)$ \cite[Section~4]{Alterfold1}. In the case $\mathcal Z(\mathcal A)\simeq \mathcal C(G_K,q_K)$, the simple objects are indexed by elements $a\in G_K$, and fusion is group addition. Hence the basis vectors on $\Sigma_g$ are indexed by $g$-tuples $(a_1,\dots,a_g)\in G_K^g.$ Therefore
$$
\dim \mathbb V_{\mathcal A}(\Sigma_g)\big|_{\mathrm{Cob}^B} = |G_K|^g.
$$
Since $|G_K|=|\Lambda^*/K\Lambda|=|\det K|,$ one obtains the standard Abelian Chern--Simons dimension formula
$$
\dim \mathcal H^{CS}(\Sigma_g)=|G_K|^g=|\det K|^g.
$$
This is exactly what one expects on the Abelian Chern--Simons side: the Hilbert space on $\Sigma_g$ is a direct sum of one-dimensional sectors labeled by $G_K^g$. Thus, in the center case, the Alterfold pants basis, the RT basis, and the Abelian Chern--Simons basis coincide under the chain of identifications above.

In the doubled case, taking the Alterfold input to be $\mathcal C_K$, the center is $\mathcal Z(\mathcal C_K)\simeq \mathcal C_K\boxtimes \mathcal C_K^{\mathrm{rev}}.$ The simple objects are indexed by pairs $(a,b)\in G_K\times G_K$. Hence the doubled Abelian Chern--Simons Hilbert space is
$$
H^{CS}_{\mathrm{dbl}}(\Sigma_g)
:=
H^{CS}_{\mathbb T,K}(\Sigma_g)\otimes
H^{CS}_{\mathbb T,-K}(\Sigma_g)
\cong
H^{CS}_{\mathbb T\times\mathbb T,\,K\oplus(-K)}(\Sigma_g),
$$
and $\dim H^{CS}_{\mathrm{dbl}}(\Sigma_g)=|G_K|^{2g}.$ Thus the center-case basis is indexed by $G_K^g$, while the doubled basis is indexed by $(G_K\times G_K)^g$.

\textbf{The BF example.} The BF example discussed earlier in the paper fits especially well into the Alterfold picture. Consider the hyperbolic lattice
$$
K_{\mathrm{hyp}}(k)=
\begin{pmatrix}
0 & k\\
k & 0
\end{pmatrix}.
$$
Its discriminant quadratic module is the standard hyperbolic metric group, and the
associated pointed modular category is the Drinfeld center of
$\mathrm{Vec}_{\mathbb Z/k\mathbb Z}$. Hence
$$
Z^{CS}_{\mathbb T,K_{\mathrm{hyp}}(k)}
\;\cong\;
Z^{RT}_{Z(\mathrm{Vec}_{\mathbb Z/k\mathbb Z})}
\;\cong\;
Z^{TV}_{\mathrm{Vec}_{\mathbb Z/k\mathbb Z}}.
$$
Thus Abelian BF theory is one of the direct examples in which the Abelian Chern--Simons theory is realized directly inside Alterfold theory.

We now understand the Abelian Chern--Simons theories considered here as determined by the associated finite quadratic module $(G_K,q_K)$, or equivalently by the pointed modular category $\mathcal C(G_K,q_K)$.  On the Alterfold side, the data is a spherical fusion category $\mathcal A$, but the part relevant to Chern--Simons theory is its center sector $\mathcal Z(\mathcal A)$. Thus the direct identification requires
$$
\mathcal Z(\mathcal A)\simeq \mathcal C(G_K,q_K).
$$

A $3$-Alterfold is a pair $(X,\Sigma)$ consisting of a $3$-manifold with boundary and a separating surface; see Figure~\ref{fig:alterfold}. Passing to the ordinary $B$-sector forgets the extra $A$-colored bulk data and leaves an ordinary oriented $3$-manifold $X$, which is exactly the spacetime of the Chern--Simons theory.  The following dictionary is written for the doubled category $$ \mathcal Z(\mathcal C_K)\simeq \mathcal C_K\boxtimes \mathcal C_K^{\mathrm{rev}}. $$ The other center-case dictionary is obtained instead when there exists a spherical fusion category $A$ with $\mathcal Z(\mathcal A)\simeq \mathcal C_K.$ In that case one replaces the doubled labels $ X_a\boxtimes X_b^{\mathrm{rev}}$ by single labels $ X_a\in \mathcal C_K.$

\begin{itemize}

\item \textbf{Kirby color $\Omega$.}\\
This is represented by the \emph{red circle}:

\hspace{1.52cm}$\displaystyle
\vcenter{\hbox{\begin{tikzpicture}[xscale=0.8, yscale=0.6, baseline=-0.3cm]
  \path[use as bounding box] (-0.45,-2.3) rectangle (0.45,2.3);

  \begin{scope}[shift={(0,2)}]
    \draw (0,0) [partial ellipse=0:360:0.6 and 0.3];
  \end{scope}
  \draw (-0.6, 2)--(-0.6,-2);
  \draw ( 0.6, 2)--( 0.6,-2);
  \draw[red, thick]        (0,0) [partial ellipse=180:360:0.6 and 0.3];
  \draw[red, dashed, thick](0,0) [partial ellipse=0:180:0.6 and 0.3];
  \node[red, right] at (0.55,0) {\scriptsize $\Omega$};

  \begin{scope}[shift={(0,-2)}]
    \draw[dashed] (0,0) [partial ellipse=0:180:0.6 and 0.3];
    \draw         (0,0) [partial ellipse=180:360:0.6 and 0.3];
  \end{scope}
\end{tikzpicture}}}
\qquad \longleftrightarrow \qquad
\Omega=\Omega_{\mathcal C_K}
=\sum_{x\in G_K}X_x.$

Here the red loop carries the Kirby color of the Alterfold input category $\mathcal C_K$. The Kirby color of its center is instead
$$
\Omega_{\mathcal Z(\mathcal C_K)}
=
\sum_{a,b\in G_K}X_a\boxtimes X_b^{\mathrm{rev}}.
$$

\item \textbf{Tube category simple object.}\\
A simple object $X_a\boxtimes X_b^{\mathrm{rev}} \in \mathcal C_K\boxtimes \mathcal C_K^{\mathrm{rev}}$. A convenient pointed version of the tube picture is:

\hspace{1.3cm}$\displaystyle
\vcenter{\hbox{\begin{tikzpicture}[xscale=0.8, yscale=0.6, baseline=-0.3cm]
  \begin{scope}[shift={(0,2)}]
    \draw (0,0) [partial ellipse=0:360:0.6 and 0.3];
  \end{scope}

  \draw (-0.6, 2)--(-0.6,-2);
  \draw ( 0.6, 2)--( 0.6,-2);

  % first blue strand: behind the red circle
  \draw [blue, ->-=0.8] (-0.18,-2.3)--(-0.18,1.7)
    node [left, pos=0.8] {\tiny $a$};

  % red circle
  \draw[red, thick]         (0,0) [partial ellipse=180:360:0.6 and 0.3];
  \draw[red, dashed, thick] (0,0) [partial ellipse=0:180:0.6 and 0.3];
  \node[red, right] at (0.55,0) {\scriptsize $\Omega$};

  % second blue strand: in front of the red circle
  \draw [blue, -<-=0.8] ( 0.18,-2.3)--( 0.18,1.7)
    node [right, pos=0.8] {\tiny $b$};

  \begin{scope}[shift={(0,-2)}]
    \draw[dashed] (0,0) [partial ellipse=0:180:0.6 and 0.3];
    \draw         (0,0) [partial ellipse=180:360:0.6 and 0.3];
  \end{scope}
\end{tikzpicture}
}}
\qquad \longleftrightarrow \qquad
X_a\boxtimes X_b^{\mathrm{rev}}.$\\

Notice that one blue line passes behind the red circle, while the other passes in front of it. These two strands represent $X_a$ and $X_b^{\mathrm{rev}}$, respectively. Thus, the cylinder represents an object in the tube-category, or doubled-center, presentation, labeled by a pair $(a,b)\in G_K\times G_K$.\\

\textbf{Chiral/single-copy embedding.}\\
A simple object $X_a\boxtimes \mathbb I$, for $X_a\in \mathcal C_K$ and $\mathbb I \in \mathcal C_K^{\mathrm{rev}}$.

\hspace{1.3cm}$\displaystyle
\vcenter{\hbox{\begin{tikzpicture}[xscale=0.8, yscale=0.6, baseline=-0.3cm]
  \begin{scope}[shift={(0,2)}]
    \draw (0,0) [partial ellipse=0:360:0.6 and 0.3];
  \end{scope}

  \draw (-0.6, 2)--(-0.6,-2);
  \draw ( 0.6, 2)--( 0.6,-2);

  % first blue strand: behind the red circle
  \draw [blue, ->-=0.8] (-0.18,-2.3)--(-0.18,1.7)
    node [left, pos=0.8] {\tiny $a$};

  % red circle
  \draw[red, thick]         (0,0) [partial ellipse=180:360:0.6 and 0.3];
  \draw[red, dashed, thick] (0,0) [partial ellipse=0:180:0.6 and 0.3];
  \node[red, right] at (0.55,0) {\scriptsize $\Omega$};

  \begin{scope}[shift={(0,-2)}]
    \draw[dashed] (0,0) [partial ellipse=0:180:0.6 and 0.3];
    \draw         (0,0) [partial ellipse=180:360:0.6 and 0.3];
  \end{scope}
\end{tikzpicture}
}}
\qquad \longleftrightarrow \qquad X_a\boxtimes \mathbb I\in \mathcal C_K\boxtimes\mathcal C_K^{\mathrm{rev}}.$

We can interpret $a$ as a topological charge
sector, or equivalently as a simple label in the RT/CS theory.
\\

\item \textbf{Fusion.}\\
Tensor product in the doubled center corresponds to addition of labels, or, equivalently, is represented by two blue strands carrying labels $a,b$:

\hspace{1.39cm}$\displaystyle
\vcenter{\hbox{\begin{tikzpicture}[xscale=0.8, yscale=0.6, baseline=-0.3cm]
  \begin{scope}[shift={(0,2)}]
    \draw (0,0) [partial ellipse=0:360:0.6 and 0.3];
  \end{scope}

  \draw (-0.6, 2)--(-0.6,-2);
  \draw ( 0.6, 2)--( 0.6,-2);

  % four blue strands
  % outer pair: closer to each other
\draw [blue, ->-=0.8] (-0.24,-2.3)--(-0.24,1.7) node [left,  pos=0.8] {\tiny };
    \draw [blue, ->-=0.8] ( 0.11,-2.3)--( 0.11,1.7) node [right, pos=0.8] {\tiny };

  % red circle
  \draw[red, thick]         (0,0) [partial ellipse=180:360:0.6 and 0.3];
  \draw[red, dashed, thick] (0,0) [partial ellipse=0:180:0.6 and 0.3];
  \node[red, right] at (0.55,0) {\scriptsize $\Omega$};

    \draw [blue, -<-=0.9] (-0.11,-2.3)--(-0.11,1.7) node [left,  pos=0.8] {\tiny };
    \draw [blue, -<-=0.9] ( 0.24,-2.3)--( 0.24,1.7) node [right, pos=0.8] {\tiny };

  \begin{scope}[shift={(0,-2)}]
    \draw[dashed] (0,0) [partial ellipse=0:180:0.6 and 0.3];
    \draw         (0,0) [partial ellipse=180:360:0.6 and 0.3];
  \end{scope}
\end{tikzpicture}
}}
\qquad \longleftrightarrow \qquad
(X_a \boxtimes X_b^{\mathrm{rev}}) \otimes (X_c \boxtimes X_d^{\mathrm{rev}})
\cong X_{a+c} \boxtimes X_{b+d}^{\mathrm{rev}}.$\\

\item \textbf{Twist.}
Represent this by a full twist/winding of the strand labeled $a,b$:\newline

\hspace{1.15cm}$\displaystyle
\vcenter{\hbox{
\begin{tikzpicture}[xscale=0.8, yscale=0.6, baseline=-0.3cm,trim left=-0.6cm]
  % top ellipse
  \begin{scope}[shift={(0,2)}]
    \draw (0,0) [partial ellipse=0:360:0.6 and 0.3];
  \end{scope}

  % cylinder sides
  \draw (-0.6, 2)--(-0.6,-2);
  \draw ( 0.6, 2)--( 0.6,-2);

% T(a): visible part from bottom going to the right
\draw [blue]
  (-0.11,-2.3)
  .. controls (0.20,-2.1) and (0.6,-1.7) .. (0.6,-0.7);
  
  % hidden part on the back of the cylinder
  \draw [blue, dashed]
    (0.55,-0.7)
    .. controls (0.55,0.0) and (-0.55,0.0) .. (-0.55,0.7);

% visible part returning to the front and ending at the top
\draw [blue, -<-=0.55]
  (-0.6,0.7)
  .. controls (-0.6,1.45) and (-0.22,1.68) .. (0,1.7)
  node [left, pos=0.82] {};

  % bottom ellipse
  \begin{scope}[shift={(0,-2)}]
    \draw[dashed] (0,0) [partial ellipse=0:180:0.6 and 0.3];
    \draw         (0,0) [partial ellipse=180:360:0.6 and 0.3];
  \end{scope}
    \begin{scope}[shift={(0,2)}]
    \draw (0,0) [partial ellipse=0:360:0.6 and 0.3];
  \end{scope}
  \draw (-0.6, 2)--(-0.6,-2);
  \draw ( 0.6, 2)--( 0.6,-2);
  \draw[red, thick]        (0,0) [partial ellipse=180:360:0.6 and 0.3];
  \draw[red, dashed, thick](0,0) [partial ellipse=0:180:0.6 and 0.3];
  \node[red, right] at (0.55,0) {\scriptsize $\Omega$};

  \begin{scope}[shift={(0,-2)}]
    \draw[dashed] (0,0) [partial ellipse=0:180:0.6 and 0.3];
    \draw         (0,0) [partial ellipse=180:360:0.6 and 0.3];
  \end{scope}

% shifted copy, still inside the cylinder
\draw [blue]
  (-0.06,-2.28)
  .. controls (0.15,-2.05) and (0.52,-1.65) .. (0.52,-0.72);

% hidden part on the back of the cylinder
\draw [blue, dashed]
  (0.52,-0.72)
  .. controls (0.52,-0.20) and (-0.50,0.20) .. (-0.50,0.72);

% visible part returning to the front and ending at the top
\draw [blue, ->-=0.4]
  (-0.50,0.72)
  .. controls (-0.50,1.42) and (-0.16,1.66) .. (0.03,1.69)
  node [left, pos=0.82] {};
  
\end{tikzpicture}
}}
\quad \longleftrightarrow \quad
\theta_{(a,b)}
=
\exp\!\Big(
\pi i\,
\big(
a^{\mathsf T}K^{-1}a
-
b^{\mathsf T}K^{-1}b
\big)
\Big).$\\

This is the twist of the doubled pointed modular category and the corresponding phase in the doubled Abelian Chern--Simons theory.\\

Fix a pointed braided realization $\mathcal C_K\simeq\operatorname{Vec}^{\,\omega_K,c_K}_{G_K},$  where $c_K(a,b)$ denotes the chosen exchange phase. Its monodromy pairing is
$$
\Omega_K(a,b)
:=
c_K(a,b)c_K(b,a)
=
\frac{q_K(a+b)}{q_K(a)q_K(b)}.
$$

\item \textbf{Braiding.}
The braiding is taken componentwise in the doubled category, so that
$$
c_{\,X_a\boxtimes X_b^{\mathrm{rev}},\,X_c\boxtimes X_d^{\mathrm{rev}}}
=
c_{X_a,X_c}\boxtimes c^{\mathrm{rev}}_{X_b,X_d}.
$$

\hspace{1.25cm}$\displaystyle
\vcenter{\hbox{\begin{tikzpicture}[xscale=0.8, yscale=0.6, baseline=-0.3cm]

  % top ellipse
  \begin{scope}[shift={(0,2)}]
    \draw (0,0) [partial ellipse=0:360:0.6 and 0.3];
  \end{scope}

  % cylinder sides
  \draw (-0.6, 2)--(-0.6,-2);
  \draw ( 0.6, 2)--( 0.6,-2);

  % left blue strand: X_b^{rev} (behind the red circle)
\draw[blue, -<-=0.8]
  (-0.45,-2.21)
  .. controls (-0.35,-0.3) and (0.35,0.8) .. (0.35,1.76)
  node[left, pos=0.30] {\tiny };

\draw[blue, ->-=0.8]
  (0.45,-2.21)
  .. controls (0.35,-0.8) and (-0.35,0.8) .. (-0.35,1.76)
  node[right, pos=0.9] {\tiny };

  % white gap for crossing of blue strands

  % red equatorial circle: X_d \boxtimes X_d^{rev}
  \draw[red, thick]         (0,0) [partial ellipse=180:360:0.6 and 0.3];
  \draw[red, dashed, thick] (0,0) [partial ellipse=0:180:0.6 and 0.3];
  \node[red, right] at (0.55,0) {\scriptsize $\Omega$};

  % right blue strand: X_a (in front of the red circle)

\draw[blue, ->-=0.9]
  (-0.52,-2.17)
  .. controls (-0.42,-0.3) and (0.27,0.8) .. (0.27,1.76)
  node[right, pos=0.30] {\tiny };

\draw[blue,-<-=0.9]
  (0.37,-2.21)
  .. controls (0.27,-0.8) and (-0.43,0.8) .. (-0.43,1.76)
  node[left, pos=0.9] {\tiny };

  % bottom ellipse
  \begin{scope}[shift={(0,-2)}]
    \draw[dashed] (0,0) [partial ellipse=0:180:0.6 and 0.3];
    \draw         (0,0) [partial ellipse=180:360:0.6 and 0.3];
  \end{scope}

\end{tikzpicture}}}
\quad \longleftrightarrow \quad
c_{\,X_a\boxtimes X_b^{\mathrm{rev}},\,X_c\boxtimes X_d^{\mathrm{rev}}}
=
c_K(a,c)\,c_K(d,b)^{-1}.$ \\

Equivalently, since the braiding in the reverse category is
$$
c^{\mathrm{rev}}_{X_b,X_d}=c^{-1}_{X_d,X_b},
$$
the braiding in the doubled category is
$$
c_{\,X_a\boxtimes X_b^{\mathrm{rev}},\,X_c\boxtimes X_d^{\mathrm{rev}}}=c_{X_a,X_c}\boxtimes c^{-1}_{X_d,X_b}.
$$

\item \textbf{Hilbert space.}

In the doubled case, the Abelian Chern--Simons Hilbert space on a genus-$g$ surface is $\mathcal H^{\mathrm{CS}}_{\mathrm{dbl}}(\Sigma_g) := \mathcal H^{\mathrm{CS}}_{\mathbb T,K}(\Sigma_g)\otimes \mathcal H^{\mathrm{CS}}_{\mathbb T,-K}(\Sigma_g) \cong \mathcal H^{\mathrm{CS}}_{\mathbb T\times \mathbb T,\,K\oplus(-K)}(\Sigma_g)$.

$$
\vcenter{\hbox{%
\begin{tikzpicture}[scale=0.25, line cap=round, line join=round, >=to]

%---------------------------------
% Outer surface Sigma_g
%---------------------------------
\filldraw[fill=white, draw=black, thick]
  (-6.2,0.0)
  .. controls (-6.2, 1.8) and (-4.8, 2.2) .. (-3.4, 2.2)
  .. controls (-2.0, 2.2) and (-1.2, 1.8) .. (-0.4, 1.2)
  .. controls ( 0.4, 0.7) and ( 1.0, 0.8) .. ( 1.7, 1.3)
  .. controls ( 2.5, 1.9) and ( 3.1, 2.2) .. ( 4.3, 2.2)
  .. controls ( 5.7, 2.2) and ( 6.6, 1.8) .. ( 7.3, 1.2)
  .. controls ( 8.0, 0.7) and ( 8.7, 0.8) .. ( 9.4, 1.2)
  .. controls (10.2, 1.7) and (10.9, 1.5) .. (11.5, 1.2)
  .. controls (12.3, 0.8) and (13.1, 1.7) .. (14.2, 2.2)
  .. controls (15.4, 2.2) and (16.6, 1.8) .. (16.6, 0.0)
  .. controls (16.6,-1.8) and (15.4,-2.2) .. (14.2,-2.2)
  .. controls (13.1,-2.2) and (12.3,-1.7) .. (11.5,-1.2)
  .. controls (10.9,-0.8) and (10.2,-0.7) .. ( 9.4,-1.2)
  .. controls ( 8.7,-1.5) and ( 8.0,-1.7) .. ( 7.3,-1.2)
  .. controls ( 6.6,-0.8) and ( 5.7,-2.2) .. ( 4.3,-2.2)
  .. controls ( 3.1,-2.2) and ( 2.5,-1.9) .. ( 1.7,-1.3)
  .. controls ( 1.0,-0.8) and ( 0.4,-0.7) .. (-0.4,-1.2)
  .. controls (-1.2,-1.8) and (-2.0,-2.2) .. (-3.4,-2.2)
  .. controls (-4.8,-2.2) and (-6.2,-1.8) .. (-6.2,0.0)
  -- cycle;

%---------------------------------
% A small macro for one displayed handle
% #1 = x-shift, #2 = a-label, #3 = b-label
%---------------------------------
\newcommand{\HandleCycles}[3]{%
  \begin{scope}[shift={(#1,0)}]
    % inner hole
    \draw[black, thick]
      (-0.95,0) .. controls (-0.55,0.42) and (0.55,0.42) .. (0.95,0);
    \draw[black, thick]
      (-0.95,0) .. controls (-0.55,-0.42) and (0.55,-0.42) .. (0.95,0);

% upper blue cycle
\draw[blue, ->-=0.80]
  (-1.70,0.25)
  .. controls (-1.70,0.95) and (1.70,0.95) .. (1.70,0.25)
  .. controls (1.70,-0.45) and (-1.70,-0.45) .. (-1.70,0.25);

% lower blue cycle
\draw[blue, -<-=0.70]
  (-1.70,-0.25)
  .. controls (-1.70,0.45) and (1.70,0.45) .. (1.70,-0.25)
  .. controls (1.70,-0.95) and (-1.70,-0.95) .. (-1.70,-0.25);

    % labels
    \node[blue] at (0,1.55) {\tiny $#2$};
    \node[blue] at (0,-1.55) {\tiny $#3$};
  \end{scope}
}

%---------------------------------
% Three displayed handles
%---------------------------------
\HandleCycles{-3.4}{a_1}{b_1}
\HandleCycles{ 4.3}{a_2}{b_2}
\HandleCycles{14.2}{a_g}{b_g}

% Ellipsis
\node at (9.3,0.0) {$\cdots$};

% Surface label
\node at (5.5,-3.45) {\tiny $\Sigma_g$};

\end{tikzpicture}%
}}
\quad \longleftrightarrow \quad \displaystyle \mathcal H^{CS}_{\mathrm{dbl}}(\Sigma_g) \cong \bigoplus_{\mathbf a,\mathbf b\in G_K^{g}} \mathbb C|\mathbf a,\mathbf b\rangle,
$$

where $\mathbf a=(a_1,\dots,a_g),$ and $\mathbf b=(b_1,\dots,b_g)$. Equivalently,
$$
\mathbb V_{\mathcal C_K}(\Sigma_g)\big|_{\mathrm{Cob}^B} \;\cong\; \mathcal H^{CS}_{\mathrm{dbl}}(\Sigma_g),
$$
and the basis is indexed by $G_K^{2g}$. In particular,
$$
\dim \mathcal H^{CS}_{\mathrm{dbl}}(\Sigma_g)=|G_K|^{2g}=|\det K|^{2g}.
$$

\item \textbf{$S$-matrix.}\\
In the doubled theory $\mathcal C_K\boxtimes \mathcal C_K^{\mathrm{rev}},$ the simple objects are labeled by pairs
$$
X_a\boxtimes X_b^{\mathrm{rev}}, \qquad (a,b)\in G_K\times G_K.
$$
Thus the genus-one $S$-matrix is indexed by two such pairs, say $(a,b)$ and $(c,d)$. The following tube picture represents the Hopf-link pairing schematically; the normalization is specified by the matrix entry on the right:
$$
\vcenter{\hbox{\scalebox{0.7}{%
\begin{tikzpicture}[xscale=0.8, yscale=0.6]
\draw[line width=0.8cm] (0,0) [partial ellipse=0:180:2 and 1.5];
\draw[white,line width=0.77cm] (0,0) [partial ellipse=-0.1:180.1:2 and 1.5];
\draw[line width=0.8cm] (0,0) [partial ellipse=180:360:2 and 1.5];
\draw[white,line width=0.77cm] (0,0) [partial ellipse=178:362:2 and 1.5];

\begin{scope}[shift={(-2,0)}]
  \draw[red] (0,0) [partial ellipse=-60:280:0.5 and 0.3];
  \draw[blue,  ->] (0,-0.3) [partial ellipse=0:170:0.2 and 0.45];
  \draw[blue, thick] (0,-0.3) [partial ellipse=190:360:0.2 and 0.45]
    node[pos=0.5, below right] {\tiny $a,b$};
\end{scope}

\begin{scope}[shift={(2,0)}]
  \draw[blue,  ->] (0,0) [partial ellipse=-80:280:0.5 and 0.3]
    node[below] {\tiny $c,d$};
\end{scope}
\end{tikzpicture}%
}}}
\qquad \longleftrightarrow \qquad
\frac{1}{|G_K|}\,\Omega_K(a,c)\,\Omega_K(b,d)^{-1}.
$$
The Hopf-link evaluates the double braiding, hence the monodromy pairing
$$\Omega_K(x,y)=c_K(x,y)c_K(y,x)=\frac{q_K(x+y)}{q_K(x)q_K(y)}.
$$
Therefore the genus-one $S$-matrix entry is
$$
S_{(a,b),(c,d)}
=
|G_K|^{-1}\,\Omega_K(a,c)\,\Omega_K(b,d)^{-1}
=
|G_K|^{-1}\,
\frac{q_K(a+c)}{q_K(a)\,q_K(c)}
\frac{q_K(b)\,q_K(d)}{q_K(b+d)}.
$$
Equivalently,
$$
S_{(a,b),(c,d)}
=
|G_K|^{-1}
\exp\!\Bigl(
2\pi i\,
\bigl(
a^{\mathsf T}K^{-1}c
-
b^{\mathsf T}K^{-1}d
\bigr)
\Bigr).
$$
On the torus Hilbert space with basis
$$
\bigl\{e_{(a,b)}\bigr\}_{(a,b)\in G_K\times G_K}
\subset \mathcal H^{CS}_{\mathrm{dbl}}(T^2),
$$
one has
$$
S\bigl(e_{(a,b)}\bigr)
=
|G_K|^{-1}
\sum_{(c,d)\in G_K\times G_K}
\Omega_K(a,c)\,\Omega_K(b,d)^{-1}\,e_{(c,d)}.
$$
Thus the Alterfold Hopf-link picture gives the doubled Hopf-link evaluation, while the corresponding normalized operator is the product of the Fourier kernel for $\mathcal C_K$ and the inverse Fourier kernel for $\mathcal C_K^{\mathrm{rev}}$.  For the all-genus $S$-matrix calculation, see \cite{Galviz2}.

\item \textbf{Partition function.}\\
Represented by a closed oriented $3$-manifold $M$, with no time boundary and hence no
associated state space:
$$
\vcenter{\hbox{%
\begin{tikzpicture}
\draw[dashed] (0, 0) rectangle (2, 2);
\draw[dashed] (.8, .8) rectangle (2.8, 2.8);
\draw[dashed] (0, 2)--+(0.8, 0.8);
\draw[dashed] (2, 2)--+(0.8, 0.8);
\draw[dashed] (2, 0)--+(0.8, 0.8);
\draw[dashed] (0, 0)--+(0.8, 0.8);
\draw (2, 2.3) node{\tiny{}};
\draw (1.3, 1.3) node[right]{\tiny{$M$}} node{};
\draw (0.3, -0.4) node[right]{\tiny{$(M, \Sigma, \Gamma)$}} node{};
\end{tikzpicture}}}\\
\qquad \longleftrightarrow \qquad
Z^{CS}_{\mathbb T\times\mathbb T,\,K\oplus(-K)}(M_L) \in \mathbb C.
$$
On the Alterfold side, this corresponds to the scalar partition function of a closed decorated $3$-alterfold $(M,\Sigma,\Gamma)$, that is $Z(M,\Sigma,\Gamma)\in \mathbb C.$ After restricting to the ordinary $B$-sector, the general pointed case realizes the canonical doubled Abelian Chern--Simons theory associated with $\mathcal C_K\boxtimes \mathcal C_K^{\mathrm{rev}},$ or equivalently with the lattice $K\oplus(-K)$. Thus, if $M=M_L$ is the closed
connected $3$-manifold determined by surgery with linking matrix $L$, then
$$
Z^{CS}_{\mathrm{dbl}}(M_L) := Z^{CS}_{\mathbb T\times\mathbb T,\,K\oplus(-K)}(M_L) = Z^{CS}_{\mathbb T,K}(M_L)\, Z^{CS}_{\mathbb T,-K}(M_L).
$$
Equivalently, using the surgery formula,
$$
Z^{CS}_{\mathbb T\times\mathbb T,\,K\oplus(-K)}(M_L) = |G_K|^{2m_M}\, |\det(L_{\mathrm{reg}})|^{-n} \sum_{[x]\in \mathbb Z^{2\rho n}/(L_{\mathrm{reg}}\otimes I_{2n})\mathbb Z^{2\rho n}} e^{2\pi i\,q_{L,K\oplus(-K)}([x])}.
$$
Here
$$
q_{L,K\oplus(-K)}([x]) \equiv \frac12 x^{\top}\bigl(L_{\mathrm{reg}}^{-1}\otimes (K\oplus(-K))\bigr)x \pmod 1.
$$

In the center case, suppose there exists a spherical fusion category $A$ such that $\mathcal Z(\mathcal A)\simeq \mathcal C_K.$ Then the ordinary $B$-sector of $\mathbb V_{\mathcal A}$ realizes the single-copy Abelian
Chern--Simons theory:
$$
\mathbb V_{\mathcal A}(M)\big|_{\mathrm{Cob}^B} \simeq Z^{RT}_{\mathcal Z(\mathcal A)}(M) \simeq Z^{RT}_{\mathcal C_K}(M) \simeq Z^{CS}_{\mathbb T,K}(M).
$$
Thus the single-copy theory appears from the center input $A$, whereas the
canonical doubled theory appears from the input $\mathcal C_K$.
For $M=M_L$, the single-copy surgery formula is
$$
Z^{CS}_{\mathbb T,K}(M_L)
=
|G_K|^{m_M}|\det(L_{\mathrm{reg}})|^{-n/2}
\sum_{[x]\in \mathbb Z^{\rho n}/(L_{\mathrm{reg}}\otimes I_n)\mathbb Z^{\rho n}}
e^{2\pi i q_{L,K}([x])}.
$$

\item
\textbf{What is forgotten in the reduction.}
The full Alterfold theory contains information carried by the separating surface, the multi-colored geometry, and the $A$-colored region; see Figure~\ref{fig:alterfold}. A single-copy Abelian Chern--Simons theory is obtained from the ordinary $B$-sector only in the center case $\mathcal Z(\mathcal A)\simeq \mathcal C(G_K,q_K)$. For a general pointed modular category $\mathcal C(G_K,q_K)$, the Alterfold center sector realizes the canonical doubled theory associated with $\mathcal C_K\boxtimes \mathcal C_K^{\mathrm{rev}},$ equivalently the Abelian Chern--Simons theory with lattice $K\oplus(-K)$. Thus the reduction forgets the extra Alterfold data and retains only the ordinary $3$-manifold input together with the center or doubled center sector. For the expressions $\mathcal H_{\mathbb T,K}(\Sigma_g)$, $Z^{CS}_{\mathbb T,K}(M)$, and their doubled analogues, see again Section~\ref{sec:Abelian-CS}, or more explicitly in~\cite{Galviz2}.
\end{itemize}

\subsection{Alterfold realization of condensations, gapped boundaries, and domain walls}
Section~\ref{sec:defects} supplies the condensation algebras and their
boundary and wall categories. We now apply the Morita-context construction
of \cite{Alterfold1,Alterfold2} to these same data. Its additional role is
to represent bulk attachment by colored tubes and to identify the resulting
intertwiner multiplicities with full-center coefficients and torus states.
The diagrams below therefore give a means of evaluating the defects already
constructed, using the bulk identification of
Proposition~\ref{prop:Alterfold-cs-reduction}.

Let $\mathcal C_K=\mathcal C(G_K,q_K)$ and let $A_H$ be the connected
\'etale algebra associated with an isotropic subgroup $H\le G_K$ in
Proposition~\ref{prop:defects-etale}. With the underlying fusion category
$\mathcal C_K$ as input, the Alterfold bulk is
$\mathcal Z(\mathcal C_K)\simeq
\mathcal C_K\boxtimes\mathcal C_K^{\mathrm{rev}}$; by
Theorem~\ref{thm:double-case}, it corresponds to
$$
Z^{CS}_{\mathbb T,K}\otimes Z^{CS}_{\mathbb T,-K}
\simeq Z^{CS}_{\mathbb T\times\mathbb T,K\oplus(-K)}.
$$
Thus $A_H$ is condensation data in the single bulk, whereas its full center is a Lagrangian algebra in this doubled bulk. The right-module category of $A_H$ determines the Morita-dual fusion category \cite{mueger2003-1,mueger2003-2}
$$
\mathcal D_H\simeq{}_{A_H}\mathcal C_K{}_{A_H},
$$
consisting of $A_H$--$A_H$ bimodules, and hence the spherical Morita
context $\{\mathcal C_K,\mathcal D_H\}$. This category is distinct from
the residual deconfined bulk $(\mathcal C_K)^0_{A_H}$ computed in
\eqref{eq:defects-local-modules}.

Fix once and for all:
\begin{itemize}
\item A set $\operatorname{Irr}(\mathcal C_K)=\{X_1,\dots,X_r\}$ of simple objects, with $X_1=\mathbf 1$.
\item $d_j:=\dim(X_j)$ and $\mu:=\sum_{j=1}^r d_j^2$. Here $\mu$ denotes the graphical global-dimension normalization, not a multiplication cochain.
\item A simple generator $J_H$ of the Morita context $\{\mathcal C_K,\mathcal D_H\}$,
  with two-sided dual $\overline{J_H}$, and quantum dimension $d_{J_H}$.
\end{itemize}

We use the standard Alterfold color convention: white surfaces are $\mathcal C_K$-colored, brown surfaces are $\mathcal D_H$-colored, the inside of a tube or handlebody is $A$-colored, and the outside is $B$-colored. Red circles indicate Kirby colors, blue strands indicate objects of $\mathcal C_K$, and black circles indicate the boundary of an $A$-colored tube.  When a square is marked ``Time Boundary'', the corresponding picture is understood in the double-square model for a cobordism whose $B$-colored region is homeomorphic to a torus times an interval. The following statement reduces these constructions to the algebras classified in Section~\ref{sec:defects}.

\begin{theorem}[cf.~\cite{Alterfold1,Alterfold2}]
For every isotropic subgroup $H \subset G_K$, the connected \'etale algebra $A_H$
determines a spherical Morita context $\{\mathcal C_K,\mathcal D_H\}$, and hence
an Alterfold TQFT
$$
V_H := V_{\{\mathcal C_K,\mathcal D_H\}}.
$$
The associated graphical realization has the following properties:

\begin{enumerate}
\item[(i)] The genus-one Alterfold partition function determines a modular invariant
matrix
$$
M_H=(z^H_{jk})_{j,k},
$$
indexed by simple objects $X_j,X_k\in\operatorname{Irr}(\mathcal C_K)$.  In the double-square model,
the torus partition function is expanded as
\begin{align*}
\vcenter{\hbox{\scalebox{0.5}{
\begin{tikzpicture}
\begin{scope}[shift={(1, 1)}, scale=1.8]
\path [fill=brown!20!white] (-1,0.5)--(2, 0.5)--(1, -0.5)--(-2, -0.5)--cycle;
\draw (-1,0.5)--(2, 0.5) (-2,-0.5)--(1, -0.5);
\draw (-2,-0.5)--(-1, 0.5) (1,-0.5)--(2, 0.5)
node [below] {\tiny $B$} node [above] {\tiny $A$};
\node at (0,0) {$\mathcal D_H$};
\end{scope}
\begin{scope}[shift={(1, -1)}, scale=1.8]
\draw (-1,0.5)--(2, 0.5) (-2,-0.5)--(1, -0.5);
\draw (-2,-0.5)--(-1, 0.5) (1,-0.5)--(2, 0.5) node [above] {\tiny $B$};
\node at (2, 0.5) [below right] {\tiny Time Boundary};
\end{scope}
\end{tikzpicture}}}}
=
\sum_{j,k} z^H_{jk}\,\frac{1}{\mu}\,
\vcenter{\hbox{\scalebox{0.5}{
\begin{tikzpicture}
\begin{scope}[shift={(1, 1)}, scale=1.8]
\draw (-1,0.5)--(2, 0.5) (-2,-0.5)--(1, -0.5);
\draw (-2,-0.5)--(-1, 0.5) (1,-0.5)--(2, 0.5)
node [above] {\tiny $B$} node [below] {\tiny $A$};
\draw [red] (-0.5,-0.5)--(0.5, 0.5);
\path [fill=white] (0.2, 0.2) circle (0.2cm);
\draw [blue] (-1.3, 0.2)--(1.7, 0.2);
\draw [blue] (-1.7, -0.2)--(-0.4, -0.2) (0, -0.2)--(1.3, -0.2);
\end{scope}
\begin{scope}[shift={(1, -1)}, scale=1.8]
\draw (-1,0.5)--(2, 0.5) (-2,-0.5)--(1, -0.5);
\draw (-2,-0.5)--(-1, 0.5) (1,-0.5)--(2, 0.5) node [above] {\tiny $B$};
\node at (2, 0.5) [below right] {\tiny Time Boundary};
\end{scope}
\end{tikzpicture}}}}.
\end{align*}

\item[(ii)] The Morita context $\{\mathcal C_K,\mathcal D_H\}$ determines tensor
functors
$$
\alpha_H^\pm:\mathcal C_K\longrightarrow \mathcal D_H,
$$
the Alterfold $\alpha$-induction functors. These realize the two braided
attachments to the algebra: the right actions on $A_H\otimes X$ use
$c_{X,A_H}$ and $c_{A_H,X}^{-1}$, respectively. Graphically, they attach
the $\mathcal C_K$-tube to the $\mathcal D_H$-colored plane:
$$
\resizebox{\textwidth}{!}{$
\alpha_H^+\!\left(
\vcenter{\hbox{\begin{tikzpicture}
\draw [dashed] (-1,-1) rectangle (1,1);
\draw[-<-=0.85,-<-=0.25,blue] (0,-1) node[black, below]{\tiny $X$} -- (0,1) node[black, above]{\tiny $X$};
\draw [fill=white] (-0.3,-0.3) rectangle (0.3,0.3);
\node at (0,0) {\tiny $f$};
\end{tikzpicture}}}
\right)
=
\frac{1}{d_{J_H}}
\vcenter{\hbox{\begin{tikzpicture}
\draw [dashed, fill=brown!20!white] (-1.5,-1.5) rectangle (1.5,1.5);
\path [fill=white] (0,1.5) [partial ellipse=0:-180:1 and 1];
\draw [blue, ->-=0.15] (0,1.5) [partial ellipse=0:-180:1 and 1];
\path [fill=white] (0,-1.5) [partial ellipse=0:180:1 and 1];
\draw [blue, -<-=0.15] (0,-1.5) [partial ellipse=0:180:1 and 1];
\path [fill=white] (0.4,1) .. controls +(0,-0.3) and +(0,0.3) .. (0.8,0)
.. controls +(0,-0.3) and +(0,0.3) .. (0.4,-1) --
(-0.4,-1) .. controls +(0,0.3) and +(0,-0.3) .. (0,0)
.. controls +(0,0.3) and +(0,-0.3) .. (-0.4,1);
\draw[dashed] (0,1) [partial ellipse=0:360:0.4 and 0.2];
\draw[dashed] (0,-1) [partial ellipse=0:360:0.4 and 0.2];
\draw (0.4,1) .. controls +(0,-0.3) and +(0,0.3) .. (0.8,0)
.. controls +(0,-0.3) and +(0,0.3) .. (0.4,-1);
\draw (-0.4,1) .. controls +(0,-0.3) and +(0,0.3) .. (0,0)
.. controls +(0,-0.3) and +(0,0.3) .. (-0.4,-1);
\draw [dashed, red] (0.4,0) [partial ellipse=0:180:0.4 and 0.2];
\draw [red] (0.4,0) [partial ellipse=180:250:0.4 and 0.2];
\draw [red] (0.4,0) [partial ellipse=280:360:0.4 and 0.2];
\draw [blue, ->-=0.7] (0,1.5)--(0,1) .. controls +(0,-0.3) and +(0,0.3) .. (0.4,0)
.. controls +(0,-0.3) and +(0,0.3) .. (0,-1)--(0,-1.5);
\begin{scope}[shift={(0.2,0.5)}]
\draw [fill=white] (-0.2,-0.2) rectangle (0.2,0.2);
\node at (0,0) {\tiny $f$};
\end{scope}
\node[black, above] at (0,1.5) {\tiny $X$};
\node[black, above] at (-1,1.5) {\tiny $\overline{J_H}$};
\node[black, above] at (1,1.5) {\tiny $J_H$};
\end{tikzpicture}}}
\quad,\qquad
\alpha_H^-\!\left(
\vcenter{\hbox{\begin{tikzpicture}
\draw [dashed] (-1,-1) rectangle (1,1);
\draw[-<-=0.85,-<-=0.25,blue] (0,-1) node[black, below]{\tiny $X$} -- (0,1) node[black, above]{\tiny $X$};
\draw [fill=white] (-0.3,-0.3) rectangle (0.3,0.3);
\node at (0,0) {\tiny $f$};
\end{tikzpicture}}}
\right)
=
\frac{1}{d_{J_H}}
\vcenter{\hbox{\begin{tikzpicture}
\draw [dashed, fill=brown!20!white] (-1.5,-1.5) rectangle (1.5,1.5);
\path [fill=white] (0,1.5) [partial ellipse=0:-180:1 and 1];
\draw [blue, ->-=0.15] (0,1.5) [partial ellipse=0:-180:1 and 1];
\path [fill=white] (0,-1.5) [partial ellipse=0:180:1 and 1];
\draw [blue, -<-=0.15] (0,-1.5) [partial ellipse=0:180:1 and 1];
\path [fill=white] (0.4,1) .. controls +(0,-0.3) and +(0,0.3) .. (0.8,0)
.. controls +(0,-0.3) and +(0,0.3) .. (0.4,-1) --
(-0.4,-1) .. controls +(0,0.3) and +(0,-0.3) .. (0,0)
.. controls +(0,0.3) and +(0,-0.3) .. (-0.4,1);
\draw[dashed] (0,1) [partial ellipse=0:360:0.4 and 0.2];
\draw[dashed] (0,-1) [partial ellipse=0:360:0.4 and 0.2];
\draw (0.4,1) .. controls +(0,-0.3) and +(0,0.3) .. (0.8,0)
.. controls +(0,-0.3) and +(0,0.3) .. (0.4,-1);
\draw (-0.4,1) .. controls +(0,-0.3) and +(0,0.3) .. (0,0)
.. controls +(0,-0.3) and +(0,0.3) .. (-0.4,-1);
\draw [dashed, red] (0.4,0) [partial ellipse=0:180:0.4 and 0.2];
\draw [blue, ->-=0.7] (0,1.5)--(0,1) .. controls +(0,-0.3) and +(0,0.3) .. (0.4,0)
.. controls +(0,-0.3) and +(0,0.3) .. (0,-1)--(0,-1.5);
\begin{scope}[shift={(0.2,0.5)}]
\draw [fill=white] (-0.2,-0.2) rectangle (0.2,0.2);
\node at (0,0) {\tiny $f$};
\end{scope}
\draw [line width=2.5, white] (0.4,0) [partial ellipse=200:320:0.4 and 0.2];
\draw [red] (0.4,0) [partial ellipse=180:360:0.4 and 0.2];
\node[black, above] at (0,1.5) {\tiny $X$};
\node[black, above] at (-1,1.5) {\tiny $\overline{J_H}$};
\node[black, above] at (1,1.5) {\tiny $J_H$};
\end{tikzpicture}}}
$}
$$\\

\item[(iii)] The associated topological full center is an object $Z_{\mathrm{full}}(A_H)\in \mathcal Z(\mathcal C_K)
\simeq \mathcal C_K\boxtimes \mathcal C_K^{\mathrm{rev}},$ and in the modular setting it has expansion
$$
Z_{\mathrm{full}}(A_H)\cong
\bigoplus_{j,k=1}^r
z^H_{kj}\,
X_k\boxtimes (X_j^*)^{\mathrm{rev}}.
$$
Writing $X_x$ for the simple object of charge $x\in G_K$, one has
$X_x^*=X_{-x}$. Consequently, the coefficient in the doubled basis
$X_a\boxtimes X_b^{\mathrm{rev}}$ is $z^H_{a,-b}$.
In Alterfold theory, it is represented by the $\mathcal D_H$-colored tube:
$$
\sum_{j,k=1}^r \frac{1}{\mu^3}\,
\vcenter{\hbox{\scalebox{0.7}{
\begin{tikzpicture}[scale=0.7]
\draw [line width=0.6cm, brown!20!white] (0,0) [partial ellipse=-0.1:180.1:2 and 1.5];
\begin{scope}[shift={(2.5,0)}]
\draw [line width=0.6cm] (0,0) [partial ellipse=0:180:2 and 1.5];
\draw [white, line width=0.57cm] (0,0) [partial ellipse=-0.1:180.1:2 and 1.5];
\draw [blue] (0,0) [partial ellipse=0:180:2.15 and 1.65];
\draw [blue] (0,0) [partial ellipse=0:180:1.85 and 1.35];
\end{scope}
\begin{scope}[shift={(2.5,0)}]
\draw [line width=0.6cm] (0,0) [partial ellipse=180:360:2 and 1.5];
\draw [white, line width=0.57cm] (0,0) [partial ellipse=178:362:2 and 1.5];
\draw [blue, -<-=0.5] (0,0) [partial ellipse=178:362:2.15 and 1.65]
node[black, pos=0.7, below] {\tiny $X_j$};
\draw [blue, ->-=0.5] (0,0) [partial ellipse=178:362:1.85 and 1.35]
node[black, pos=0.7, above] {\tiny $X_k$};
\end{scope}
\draw [line width=0.6cm, brown!20!white] (0,0) [partial ellipse=180:360:2 and 1.5];
\begin{scope}[shift={(4.5,0)}]
\draw [red, dashed](0,0) [partial ellipse=0:180:0.4 and 0.25];
\draw [white, line width=4pt] (0,0) [partial ellipse=290:270:0.4 and 0.25];
\draw [red] (0,0) [partial ellipse=260:360:0.4 and 0.25];
\draw [red] (0,0) [partial ellipse=180:230:0.4 and 0.25];
\end{scope}
\end{tikzpicture}}}}
\;\cong\;
\frac{1}{\mu}\,
\vcenter{\hbox{\scalebox{0.7}{
\begin{tikzpicture}[xscale=0.8, yscale=0.6]
\begin{scope}[shift={(0,3)}]
\draw (0,0) [partial ellipse=0:360:0.6 and 0.3];
\end{scope}
\path [fill=white] (-0.6,0) rectangle (0.6,2.7);
\draw [blue, ->-=0.5] (-0.2,2.8)--(-0.2,0) node [left, pos=0.6] {\tiny $X_j$};
\draw [blue, -<-=0.5] (0.2,2.8)--(0.2,0) node [right, pos=0.6] {\tiny $X_k$};
\draw (-0.6,3)--(-0.6,0) (0.6,3)--(0.6,0);
\draw [blue] (-0.2,-3.2)--(-0.2,0) (0.2,-3.2)--(0.2,0);
\draw (-0.6,-3)--(-0.6,0) (0.6,-3)--(0.6,0);
\begin{scope}[shift={(0,-3)}]
\draw [dashed](0,0) [partial ellipse=0:180:0.6 and 0.3];
\draw (0,0) [partial ellipse=180:360:0.6 and 0.3];
\end{scope}
\draw [red, dashed](0,0) [partial ellipse=0:180:0.6 and 0.3];
\draw [white, line width=4pt] (0,0) [partial ellipse=220:270:0.6 and 0.3];
\draw [red] (0,0) [partial ellipse=180:280:0.6 and 0.3];
\draw [red] (0,0) [partial ellipse=300:360:0.6 and 0.3];
\end{tikzpicture}}}}.
$$\\

\item[(iv)] In the modular setting, the Alterfold torus partition function and the
$\alpha$-inductions are related by
$$
z^H_{jk}
=
\dim\Hom_{\mathcal D_H}\!\bigl(\alpha_H^+(X_j),\alpha_H^-(X_k)\bigr).
$$
In charge notation these multiplicities reduce to
\begin{equation}
\label{eq:alterfold-pointed-multiplicities}
z^H_{xy}=
\begin{cases}
1,&x,y\in H^\perp\ \text{and}\ x-y\in H,\\
0,&\text{otherwise}.
\end{cases}
\end{equation}
Thus the same coefficients encode braided attachment, the full-center
expansion in \textup{(iii)}, and the genus-one diagram in \textup{(i)}.
\end{enumerate}

Thus the pair $(\mathcal C_K,A_H)$, equivalently the isotropic subgroup
$H\subset G_K$, determines intrinsically the Alterfold modular invariant,
$\alpha$-inductions, and topological full center attached to the corresponding
condensation.
\end{theorem}

\begin{proof}
The Morita-context and graphical statements are the constructions of
\cite{Alterfold1,Alterfold2} applied to $A_H$. For
\eqref{eq:alterfold-pointed-multiplicities}, a homogeneous module
intertwiner exists precisely when $x-y\in H$, and is then unique up
to scalar, as in Proposition~\ref{prop:defects-etale}. Compatibility
with the two braided actions additionally requires trivial monodromy
with $H$, hence $x\in H^\perp$; the coset condition then also gives
$y\in H^\perp$. This identifies the coefficients with the surviving
charge sectors of \eqref{eq:defects-local-modules}.
\end{proof}

\textbf{Boundary state (torus partition vector).}\\
Write $\mathcal D=\mathcal D_H$ in the following diagrams, and let
$|a,b\rangle$ denote the doubled torus basis vector labelled by
$X_a\boxtimes X_b^{\mathrm{rev}}$. The full-center multiplicities give
the state represented by the plain $\mathcal D$-colored double-square surface:
$$
\vcenter{\hbox{
\begin{tikzpicture}[rotate=-90, scale=1.1]
\path [fill=brown!20!white] (-1,0.5)--(2, 0.5)--(1, -0.5)--(-2, -0.5)--cycle;
\draw (-1,0.5)--(2, 0.5) (-2,-0.5)--(1, -0.5);
\draw (-2,-0.5)--(-1, 0.5) (1,-0.5)--(2, 0.5) node [right] {\tiny $B$} node [left] {\tiny $A$};
\node at (-0.9, 0.3) {\tiny $\mathcal{D}$};
\end{tikzpicture}}}
\hspace{1cm}
\vcenter{\hbox{\scalebox{1}{
\begin{tikzpicture}[xscale=0.8, yscale=0.6]
\begin{scope}[shift={(0,2)}]
\draw (0,0) [partial ellipse=0:360:0.6 and 0.3];
\end{scope}
\draw (-0.6, 2)--(-0.6, 0) (0.6, 2)--(0.6, 0); % upper one
\draw (-0.6, -2)--(-0.6, 0) (0.6, -2)--(0.6, 0);
\draw [blue] (0, -2.3) --(0, 1.7);
\draw [line width=0.2cm,white] (0, -0.18) --(0, -0.38);
\draw [red] (0,0) [partial ellipse=180:360:0.6 and 0.3];
\draw [red, dashed] (0,0) [partial ellipse=0:180:0.6 and 0.3];
%\draw [blue] (-0.1, -2.3)--(-0.1, 1.7);
\begin{scope}[shift={(0,-2)}]
\draw [dashed](0,0) [partial ellipse=0:180:0.6 and 0.3];
\draw (0,0) [partial ellipse=180:360:0.6 and 0.3];
\end{scope}
\end{tikzpicture}}}}
\qquad \longleftrightarrow \qquad
Z^{\mathcal D}
=
\sum_{a,b\in G_K}
\widehat z^{\mathcal D}_{ab}\,
|a,b\rangle,
\qquad
\widehat z^{\mathcal D}_{ab}
:=
z^{\mathcal D}_{a,-b}.$$

Under the state-space identification of
Proposition~\ref{prop:Alterfold-cs-reduction},
$$
\mathbb V_{\mathcal C_K}(T^2)\big|_{\mathrm{Cob}^B}
\cong Z^{RT}_{\mathcal Z(\mathcal C_K)}(T^2)
\cong\mathcal H^{CS}_{\mathrm{dbl}}(T^2),
$$
$Z^{\mathcal D}$ is the corresponding doubled Chern--Simons torus
state. Pairing it with a bordism state produces a closed amplitude
by Proposition~\ref{prop:cylinder-gluing}. The following diagrams
extract its coefficients by gluing a labelled tube. Here
$\alpha_{\mathcal D}$ denotes the coefficient functional, and
$I:\mathcal D\to\mathcal Z(\mathcal C_K)$ is the right adjoint of
the forgetful functor from $\mathcal Z(\mathcal D)$, transported along
the Morita equivalence of centers; thus
$I(\mathbf1_{\mathcal D})\cong Z_{\mathrm{full}}(A_H)$.
\begin{align*}
\widehat z^{\mathcal D}_{ab}
=
\alpha_{\mathcal D}
\!\bigl(X_a\boxtimes X_b^{\mathrm{rev}}\bigr)
=
z^{\mathcal D}_{a,-b}
=&\;
\frac{1}{\mu^2}
\vcenter{\hbox{\scalebox{0.5}{
\begin{tikzpicture}
\begin{scope}[shift={(1,1)}, scale=1.8]
\path[fill=brown!20!white] (-1,0.5)--(2,0.5)--(1,-0.5)--(-2,-0.5)--cycle;
\draw (-1,0.5)--(2,0.5) (-2,-0.5)--(1,-0.5);
\draw (-2,-0.5)--(-1,0.5) (1,-0.5)--(2,0.5)
  node[below] {\tiny $B$} node[above] {\tiny $A$};
\node at (0,0) {$\mathcal D$};
\end{scope}
\begin{scope}[shift={(1,-1)}, scale=1.8]
\draw (-1,0.5)--(2,0.5) (-2,-0.5)--(1,-0.5);
\draw (-2,-0.5)--(-1,0.5) (1,-0.5)--(2,0.5)
  node[above] {\tiny $B$} node[below] {\tiny $A$};
\node at (2,1) {$\mathcal Z(\mathcal C_K)$};
\draw[red] (-0.5,-0.5)--(0.5,0.5);
\path[fill=white] (0.2,0.2) circle (0.2cm);
\draw[blue] (-1.3,0.2)--(1.7,0.2);
\draw[blue] (-1.7,-0.2)--(-0.4,-0.2) (0,-0.2)--(1.3,-0.2);
\end{scope}
\end{tikzpicture}}}}\\[1ex]
=&\;
\frac{1}{\mu^2}\,
\vcenter{\hbox{\scalebox{0.7}{
\begin{tikzpicture}[scale=0.7]
\draw[line width=0.6cm, brown!20!white] (0,0) [partial ellipse=-0.1:180.1:2 and 1.5];
\begin{scope}[shift={(2.5,0)}]
  \draw[line width=0.6cm] (0,0) [partial ellipse=0:180:2 and 1.5];
  \draw[white, line width=0.57cm] (0,0) [partial ellipse=-0.1:180.1:2 and 1.5];
  \draw[blue] (0,0) [partial ellipse=0:180:2.15 and 1.65];
  \draw[blue] (0,0) [partial ellipse=0:180:1.85 and 1.35];
\end{scope}
\begin{scope}[shift={(2.5,0)}]
  \draw[line width=0.6cm] (0,0) [partial ellipse=180:360:2 and 1.5];
  \draw[white, line width=0.57cm] (0,0) [partial ellipse=178:362:2 and 1.5];
  \draw[blue, <-] (0,0) [partial ellipse=178:362:2.15 and 1.65]
    node[black, pos=0.7, below] {\tiny $X_b$};
  \draw[blue, ->] (0,0) [partial ellipse=178:362:1.85 and 1.35]
    node[black, pos=0.7, above] {\tiny $X_a$};
\end{scope}
\draw[line width=0.6cm, brown!20!white] (0,0) [partial ellipse=180:360:2 and 1.5];
\begin{scope}[shift={(4.5,0)}]
  \draw[red, dashed] (0,0) [partial ellipse=0:180:0.4 and 0.25];
  \draw[white, line width=4pt] (0,0) [partial ellipse=290:270:0.4 and 0.25];
  \draw[red] (0,0) [partial ellipse=260:360:0.4 and 0.25];
  \draw[red] (0,0) [partial ellipse=180:230:0.4 and 0.25];
\end{scope}
\end{tikzpicture}}}}\\[1ex]
=&\;
\dim\operatorname{Hom}_{\mathcal Z(\mathcal C_K)}\!\bigl(
I(\mathbf 1_{\mathcal D}),\,
X_a\boxtimes X_b^{\mathrm{rev}}
\bigr).
\end{align*}

The graphical normalization is
$$
\mu=\operatorname{Dim}(\mathcal C_K)=|G_K|=|\det K|,
\qquad
\mu^2=\operatorname{Dim}(\mathcal Z(\mathcal C_K)).
$$
The coefficient extraction therefore links the local attachment
multiplicities to the global torus state in a fixed normalization.

\begin{remark}
For a Lagrangian subgroup $L$, the boundary category
$\mathcal A_L=(\mathcal C_K)_{A_L}$ of
Theorem~\ref{thm:defect-classification} satisfies
$\mathcal Z(\mathcal A_L)\simeq\mathcal C_K$. Choosing
$\mathcal A_L$ as Alterfold input realizes the single bulk by
Theorem~\ref{thm:center-case}; choosing the underlying fusion category
$\mathcal C_K$ gives the doubled full-center realization above.
For a general isotropic $H$, the full center remains Lagrangian in
the double, although $A_H$ need not define a vacuum boundary for
$\mathcal C_K$.
\end{remark}

For domain walls, we use the folding construction already established
in Theorem~\ref{thm:defect-classification}. A wall between
$\mathcal C_K$ and $\mathcal C_{K'}$ is specified by its Lagrangian
algebra $A_M$ in
$$
\mathcal C_K \boxtimes \mathcal C_{K'}^{\mathrm{rev}}
\;\simeq\;
\mathcal C\!\bigl((G_K,q_K)\oplus(G_{K'},q_{K'}^{-1})\bigr).
$$

Its wall-line category provides a center presentation of this folded
bulk, to which Proposition~\ref{prop:Alterfold-cs-reduction} applies.
Composition is represented by gluing the corresponding Morita-context
diagrams. In the invertible case, mutually inverse walls $M,N$ give
the cancellation identities
$$
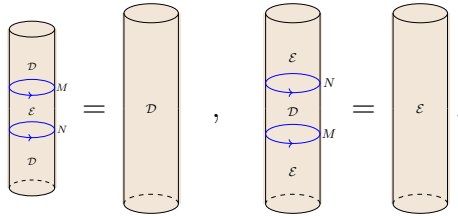
\vcenter{\hbox{\scalebox{0.5}{
\begin{tikzpicture}[yscale=0.7]
\path [fill=brown!20!white](-0.65,-3) rectangle (0.65,3);
\begin{scope}[shift={(0,3)}]
\path [fill=brown!20!white] (0,0) [partial ellipse=0:180:0.6 and 0.3];
\draw (0,0) [partial ellipse=0:360:0.6 and 0.3];
\end{scope}
\draw (-0.6,3)--(-0.6,0) (0.6,3)--(0.6,0);
\draw (-0.6,-3)--(-0.6,0) (0.6,-3)--(0.6,0);
\begin{scope}[shift={(0,-3)}]
\path [fill=brown!20!white] (0,0) [partial ellipse=180:360:0.6 and 0.3];
\draw [dashed](0,0) [partial ellipse=0:180:0.6 and 0.3];
\draw (0,0) [partial ellipse=180:360:0.6 and 0.3];
\end{scope}
\draw[blue, ->-=0.75] (0,0.8) [partial ellipse=0:360:0.6 and 0.3];
\draw[blue, ->-=0.75] (0,-0.8) [partial ellipse=0:360:0.6 and 0.3];
\draw (0.8,0.8) node{\tiny $M$};
\draw (0.8,-0.8) node{\tiny $N$};
\draw (0,1.6) node{\tiny $\mathcal D$};
\draw (0,-0.1) node{\tiny $\mathcal E$};
\draw (0,-2) node{\tiny $\mathcal D$};
\end{tikzpicture}}}}
=
\vcenter{\hbox{\scalebox{0.6}{
\begin{tikzpicture}[yscale=0.7]
\path [fill=brown!20!white](-0.65,-3) rectangle (0.65,3);
\begin{scope}[shift={(0,3)}]
\path [fill=brown!20!white] (0,0) [partial ellipse=0:180:0.6 and 0.3];
\draw (0,0) [partial ellipse=0:360:0.6 and 0.3];
\end{scope}
\draw (-0.6,3)--(-0.6,0) (0.6,3)--(0.6,0);
\draw (-0.6,-3)--(-0.6,0) (0.6,-3)--(0.6,0);
\begin{scope}[shift={(0,-3)}]
\path [fill=brown!20!white] (0,0) [partial ellipse=180:360:0.6 and 0.3];
\draw [dashed](0,0) [partial ellipse=0:180:0.6 and 0.3];
\draw (0,0) [partial ellipse=180:360:0.6 and 0.3];
\end{scope}
\draw (0,0) node{\tiny $\mathcal D$};
\end{tikzpicture}}}}
\quad,\quad
\vcenter{\hbox{\scalebox{0.6}{
\begin{tikzpicture}[yscale=0.7]
\path [fill=brown!20!white](-0.65,-3) rectangle (0.65,3);
\begin{scope}[shift={(0,3)}]
\path [fill=brown!20!white] (0,0) [partial ellipse=0:180:0.6 and 0.3];
\draw (0,0) [partial ellipse=0:360:0.6 and 0.3];
\end{scope}
\draw (-0.6,3)--(-0.6,0) (0.6,3)--(0.6,0);
\draw (-0.6,-3)--(-0.6,0) (0.6,-3)--(0.6,0);
\begin{scope}[shift={(0,-3)}]
\path [fill=brown!20!white] (0,0) [partial ellipse=180:360:0.6 and 0.3];
\draw [dashed](0,0) [partial ellipse=0:180:0.6 and 0.3];
\draw (0,0) [partial ellipse=180:360:0.6 and 0.3];
\end{scope}
\draw[blue, ->-=0.75] (0,0.8) [partial ellipse=0:360:0.6 and 0.3];
\draw[blue, ->-=0.75] (0,-0.8) [partial ellipse=0:360:0.6 and 0.3];
\draw (0.8,0.8) node{\tiny $N$};
\draw (0.8,-0.8) node{\tiny $M$};
\draw (0,1.6) node{\tiny $\mathcal E$};
\draw (0,-0.1) node{\tiny $\mathcal D$};
\draw (0,-2) node{\tiny $\mathcal E$};
\end{tikzpicture}}}}
=
\vcenter{\hbox{\scalebox{0.6}{
\begin{tikzpicture}[yscale=0.7]
\path [fill=brown!20!white](-0.65,-3) rectangle (0.65,3);
\begin{scope}[shift={(0,3)}]
\path [fill=brown!20!white] (0,0) [partial ellipse=0:180:0.6 and 0.3];
\draw (0,0) [partial ellipse=0:360:0.6 and 0.3];
\end{scope}
\draw (-0.6,3)--(-0.6,0) (0.6,3)--(0.6,0);
\draw (-0.6,-3)--(-0.6,0) (0.6,-3)--(0.6,0);
\begin{scope}[shift={(0,-3)}]
\path [fill=brown!20!white] (0,0) [partial ellipse=180:360:0.6 and 0.3];
\draw [dashed](0,0) [partial ellipse=0:180:0.6 and 0.3];
\draw (0,0) [partial ellipse=180:360:0.6 and 0.3];
\end{scope}
\draw (0,0) node{\tiny $\mathcal E$};
\end{tikzpicture}}}}\,.
$$
Here $\mathcal D$ and $\mathcal E$ label the Morita presentations, and the circles $M,N$ represent inverse wall data. These cancellation identities require invertibility. General wall composition retains the relative tensor product and its junction data from Section~\ref{sec:defects}; it need not reduce to the empty cylinder.

The Alterfold realization thus supplies a graphical evaluation of the earlier defect data: attachment diagrams determine the full-center coefficients, colored tubes define torus states, and gluing computes their amplitudes and compatible wall compositions. Non-Lagrangian condensations are included by the same construction, as they already are in Proposition~\ref{prop:defects-etale}. The contribution here is this common evaluation calculus, whose local Wilson-line coefficients are made explicit in Section~\ref{sec:gauge-defect}.

\newcommand{\TableCylinder}[1]{%
\makebox[\linewidth][c]{%
\begin{tikzpicture}[xscale=0.5,transform shape, line cap=round, line join=round, >=to,yscale=0.3,baseline=(current bounding box.center)]
\path[use as bounding box] (-0.85,-2.45) rectangle (0.85,2.45);
#1
\end{tikzpicture}%
}%
}

\newcommand{\DrawTableCylinder}{%
  \begin{scope}[shift={(0,2)}]
    \draw (0,0) [partial ellipse=0:360:0.6 and 0.3];
  \end{scope}
  \draw (-0.6,2)--(-0.6,-2);
  \draw (0.6,2)--(0.6,-2);
  \begin{scope}[shift={(0,-2)}]
    \draw[dashed] (0,0) [partial ellipse=0:180:0.6 and 0.3];
    \draw (0,0) [partial ellipse=180:360:0.6 and 0.3];
  \end{scope}
}

\newpage
\begin{table}[H]
\tiny
\centering
\renewcommand{\arraystretch}{1.30}
\begin{tabular}{p{0.24\textwidth}p{0.24\textwidth}p{0.42\textwidth}}
\toprule
\textbf{Categorical data}
& \textbf{$\qquad$3-Alterfold Theory} 
& \textbf{$\qquad$ Abelian CS interpretation} \\
\midrule

Kirby color &
\TableCylinder{%
  \DrawTableCylinder
  \draw[red, thick]        (0,0) [partial ellipse=180:360:0.6 and 0.3];
  \draw[red, dashed, thick](0,0) [partial ellipse=0:180:0.6 and 0.3];
  \node[red, right] at (0.55,0) {\large$\Omega$};
}
&$\Omega_{\mathcal C_K}=\sum_{x\in G_K}X_x.$
Kirby color of the Alterfold input category; all simple
quantum dimensions are $1$.
\\\\

Simple object $X_a\boxtimes X_b^{\mathrm{rev}}$, $(a,b)\in G_K\times G_K$ &
\TableCylinder{%
  \DrawTableCylinder

  \draw [blue, ->-=0.8] (-0.18,-2.3)--(-0.18,1.7)
    node [left, pos=0.8] {\small $a$};

  \draw[red, thick]         (0,0) [partial ellipse=180:360:0.6 and 0.3];
  \draw[red, dashed, thick] (0,0) [partial ellipse=0:180:0.6 and 0.3];

  \draw [blue, -<-=0.8] (0.18,-2.3)--(0.18,1.7)
    node [right, pos=0.8] {\small $b$};
}
& $X_a\boxtimes X_b^{\mathrm{rev}}\in \mathcal C_K\boxtimes \mathcal C_K^{\mathrm{rev}}.$ Doubled Abelian topological charge sector labeled by $(a,b)$.
\\\\

Twist &
\TableCylinder{%
  \DrawTableCylinder

  % red equatorial circle
  \draw[red, thick]        (0,0) [partial ellipse=180:360:0.6 and 0.3];
  \draw[red, dashed, thick](0,0) [partial ellipse=0:180:0.6 and 0.3];

  % first twisted strand
  \draw [blue]
    (-0.11,-2.3)
    .. controls (0.20,-2.1) and (0.6,-1.7) .. (0.6,-0.7);

  \draw [blue, dashed]
    (0.55,-0.7)
    .. controls (0.55,0.0) and (-0.55,0.0) .. (-0.55,0.7);

  \draw [blue, -<-=0.55]
    (-0.6,0.7)
    .. controls (-0.6,1.45) and (-0.22,1.68) .. (0,1.7)
    node [left, pos=0.82] {};

  % shifted second twisted strand
  \draw [blue]
    (-0.06,-2.28)
    .. controls (0.15,-2.05) and (0.52,-1.65) .. (0.52,-0.72);

  \draw [blue, dashed]
    (0.52,-0.72)
    .. controls (0.52,-0.20) and (-0.50,0.20) .. (-0.50,0.72);

  \draw [blue, ->-=0.4]
    (-0.50,0.72)
    .. controls (-0.50,1.42) and (-0.16,1.66) .. (0.03,1.69)
    node [right, pos=0.82] {};
}
& $\theta_{(a,b)} = q_K(a)\,q_K(b)^{-1} .\qquad\qquad\qquad\qquad\qquad\qquad$ Doubled self-rotation phase; $T$-modular transformation.
\\\\

Braiding &
\TableCylinder{%
  \DrawTableCylinder

  \draw[blue, -<-=0.8]
    (-0.45,-2.21)
    .. controls (-0.35,-0.3) and (0.35,0.8) .. (0.35,1.76);

  \draw[blue, ->-=0.8]
    (0.45,-2.21)
    .. controls (0.35,-0.8) and (-0.35,0.8) .. (-0.35,1.76);

  \draw[red, thick]         (0,0) [partial ellipse=180:360:0.6 and 0.3];
  \draw[red, dashed, thick] (0,0) [partial ellipse=0:180:0.6 and 0.3];

  \draw[blue, ->-=0.9]
    (-0.52,-2.17)
    .. controls (-0.42,-0.3) and (0.27,0.8) .. (0.27,1.76);

  \draw[blue,-<-=0.9]
    (0.37,-2.21)
    .. controls (0.27,-0.8) and (-0.43,0.8) .. (-0.43,1.76);
}
&  $c^{\mathrm{dbl}}_{(a,b),(c,d)}
   =c_K(a,c)c_K(d,b)^{-1}.$
   Doubled exchange phase in the chosen Abelian case.
\\\\

Surface state spaces &
\centering
$\vcenter{\hbox{
\begin{tikzpicture}[scale=0.15,transform shape, line cap=round, line join=round, >=to]

\filldraw[fill=white, draw=black, thick]
  (-6.2,0.0)
  .. controls (-6.2, 1.8) and (-4.8, 2.2) .. (-3.4, 2.2)
  .. controls (-2.0, 2.2) and (-1.2, 1.8) .. (-0.4, 1.2)
  .. controls ( 0.4, 0.7) and ( 1.0, 0.8) .. ( 1.7, 1.3)
  .. controls ( 2.5, 1.9) and ( 3.1, 2.2) .. ( 4.3, 2.2)
  .. controls ( 5.7, 2.2) and ( 6.6, 1.8) .. ( 7.3, 1.2)
  .. controls ( 8.0, 0.7) and ( 8.7, 0.8) .. ( 9.4, 1.2)
  .. controls (10.2, 1.7) and (10.9, 1.5) .. (11.5, 1.2)
  .. controls (12.3, 0.8) and (13.1, 1.7) .. (14.2, 2.2)
  .. controls (15.4, 2.2) and (16.6, 1.8) .. (16.6, 0.0)
  .. controls (16.6,-1.8) and (15.4,-2.2) .. (14.2,-2.2)
  .. controls (13.1,-2.2) and (12.3,-1.7) .. (11.5,-1.2)
  .. controls (10.9,-0.8) and (10.2,-0.7) .. ( 9.4,-1.2)
  .. controls ( 8.7,-1.5) and ( 8.0,-1.7) .. ( 7.3,-1.2)
  .. controls ( 6.6,-0.8) and ( 5.7,-2.2) .. ( 4.3,-2.2)
  .. controls ( 3.1,-2.2) and ( 2.5,-1.9) .. ( 1.7,-1.3)
  .. controls ( 1.0,-0.8) and ( 0.4,-0.7) .. (-0.4,-1.2)
  .. controls (-1.2,-1.8) and (-2.0,-2.2) .. (-3.4,-2.2)
  .. controls (-4.8,-2.2) and (-6.2,-1.8) .. (-6.2,0.0)
  -- cycle;

\newcommand{\HandleCycles}[3]{%
  \begin{scope}[shift={(#1,0)}]
    \draw[black, thick]
      (-0.95,0) .. controls (-0.55,0.42) and (0.55,0.42) .. (0.95,0);
    \draw[black, thick]
      (-0.95,0) .. controls (-0.55,-0.42) and (0.55,-0.42) .. (0.95,0);

    \draw[blue, ->-=0.80]
      (-1.70,0.25)
      .. controls (-1.70,0.95) and (1.70,0.95) .. (1.70,0.25)
      .. controls (1.70,-0.45) and (-1.70,-0.45) .. (-1.70,0.25);

    \draw[blue, -<-=0.70]
      (-1.70,-0.25)
      .. controls (-1.70,0.45) and (1.70,0.45) .. (1.70,-0.25)
      .. controls (1.70,-0.95) and (-1.70,-0.95) .. (-1.70,-0.25);

    \node[blue] at (0,1.55) {\Huge $#2$};
    \node[blue] at (0,-1.55) {\Huge $#3$};
  \end{scope}
}

\HandleCycles{-3.4}{a_1}{b_1}
\HandleCycles{ 4.3}{a_2}{b_2}
\HandleCycles{14.2}{a_g}{b_g}

\node at (9.3,0.0) {\Huge $\cdots$};
\node at (5.5,-3.45) {\Huge $\Sigma_g$};

\end{tikzpicture}%
}}$
& $\mathbb V_{\mathcal C_K}(\Sigma_g)\big|_{\mathrm{Cob}^B}
\cong
 \mathcal H^{CS}_{\mathrm{dbl}}(\Sigma_g) \cong \bigoplus_{\mathbf a,\mathbf b\in G_K^{g}} \mathbb C|\mathbf a,\mathbf b\rangle.$
Hence $\dim\mathcal H^{CS}_{\mathrm{dbl}}(\Sigma_g) = |G_K|^{2g} = |\det K|^{2g}\qquad$. Hilbert space.
\\\\

Genus-one $S$-Matrix &
\centering
$
\vcenter{\hbox{\scalebox{0.35}{%
\begin{tikzpicture}[xscale=0.8, yscale=0.6]
\draw[line width=0.8cm] (0,0) [partial ellipse=0:180:2 and 1.5];
\draw[white,line width=0.77cm] (0,0) [partial ellipse=-0.1:180.1:2 and 1.5];
\draw[line width=0.8cm] (0,0) [partial ellipse=180:360:2 and 1.5];
\draw[white,line width=0.77cm] (0,0) [partial ellipse=178:362:2 and 1.5];

\begin{scope}[shift={(-2,0)}]
  \draw[red] (0,0) [partial ellipse=-60:280:0.5 and 0.3];
  \draw[blue,  ->] (0,-0.3) [partial ellipse=0:170:0.2 and 0.45];
  \draw[blue, thick] (0,-0.3) [partial ellipse=190:360:0.2 and 0.45]
    node[pos=0.5, below right] {\tiny $(a,b)$};
\end{scope}

\begin{scope}[shift={(2,0)}]
  \draw[blue,  ->] (0,0) [partial ellipse=-80:280:0.5 and 0.3]
    node[below] {\tiny $(c,d)$};
\end{scope}
\end{tikzpicture}%
}}}$
& $S_{(a,b),(c,d)} = |G_K|^{-1} \Omega_K(a,c)\Omega_K(b,d)^{-1}.$ Doubled Hopf-link evaluation; $S$-modular transformation.
\\\\

Partition function &
\centering
$\qquad\vcenter{\hbox{%
\begin{tikzpicture}[scale=0.45, baseline=-0.45ex]
\draw[dashed] (0, 0) rectangle (2, 2);
\draw[dashed] (.8, .8) rectangle (2.8, 2.8);
\draw[dashed] (0, 2)--+(0.8, 0.8);
\draw[dashed] (2, 2)--+(0.8, 0.8);
\draw[dashed] (2, 0)--+(0.8, 0.8);
\draw[dashed] (0, 0)--+(0.8, 0.8);
\draw (0.7, 1.3) node[right]{\tiny $M$};
\end{tikzpicture}}}$
& $Z(M,\varnothing,\varnothing)
=
Z^{CS}_{\mathbb T\times\mathbb T,\,K\oplus(-K)}(M)$
for a closed, entirely $B$-colored manifold $M$.
\\\\

Boundary state $Z^{\mathcal D}$ &
\centering
$\vcenter{\hbox{
\begin{tikzpicture}[rotate=-90, scale=0.4, baseline=-0.08cm]
  \path[fill=brown!20!white] (-1,0.5)--(2,0.5)--(1,-0.5)--(-2,-0.5)--cycle;
  \draw (-1,0.5)--(2,0.5);
  \draw (-2,-0.5)--(1,-0.5);
  \draw (-2,-0.5)--(-1,0.5);
  \draw (1,-0.5)--(2,0.5);
  \node at (-0.9,0.2) {\tiny $\mathcal D$};
  \node[left]  at (2,0.30) {\tiny $A$};
  \node[right] at (2,0.55) {\tiny $B$};
\end{tikzpicture}}}
\hspace{-0.2cm}
\vcenter{\hbox{\scalebox{1}{
\begin{tikzpicture}[scale=0.23, baseline=-0.08cm]
\begin{scope}[shift={(0,2)}]
\draw (0,0) [partial ellipse=0:360:0.6 and 0.3];
\end{scope}
\draw (-0.6, 2)--(-0.6, 0) (0.6, 2)--(0.6, 0);
\draw (-0.6, -2)--(-0.6, 0) (0.6, -2)--(0.6, 0);
\draw [blue] (0, -2.3) --(0, 1.7);
\draw [line width=0.2cm,white] (0, -0.18) --(0, -0.38);
\draw [red] (0,0) [partial ellipse=180:360:0.6 and 0.3];
\draw [red, dashed] (0,0) [partial ellipse=0:180:0.6 and 0.3];
\begin{scope}[shift={(0,-2)}]
\draw [dashed](0,0) [partial ellipse=0:180:0.6 and 0.3];
\draw (0,0) [partial ellipse=180:360:0.6 and 0.3];
\end{scope}
\end{tikzpicture}}}}$
& $Z^{\mathcal D}
=
\sum_{a,b\in G_K}
\widehat z^{\mathcal D}_{ab}\,|a,b\rangle,
\quad
\widehat z^{\mathcal D}_{ab}
=
z^{\mathcal D}_{a,-b}.$
Under the doubled-center identification this determines a state in
$\mathcal H^{CS}_{\mathrm{dbl}}(T^2)$. Boundary wavefunction.
\\\\

Isotropic subgroup $H\subset G_K$; connected \'etale algebra $A_H$ &
$\qquad\qquad\vcenter{\hbox{\scalebox{0.35}{%
\begin{tikzpicture}[xscale=0.8, yscale=0.6]
\draw [line width=0.6cm, brown!20!white] (0,0) [partial ellipse=-0.1:180.1:2 and 1.5];
\begin{scope}[shift={(2.5,0)}]
\draw [line width=0.6cm] (0,0) [partial ellipse=0:180:2 and 1.5];
\draw [white, line width=0.57cm] (0,0) [partial ellipse=-0.1:180.1:2 and 1.5];
\draw [blue] (0,0) [partial ellipse=0:180:2.15 and 1.65];
\draw [blue] (0,0) [partial ellipse=0:180:1.85 and 1.35];
\end{scope}
\begin{scope}[shift={(2.5,0)}]
\draw [line width=0.6cm] (0,0) [partial ellipse=180:360:2 and 1.5];
\draw [white, line width=0.57cm] (0,0) [partial ellipse=178:362:2 and 1.5];
\draw [blue, -<-=0.5] (0,0) [partial ellipse=178:362:2.15 and 1.65]
node[black, pos=0.7, below] {\tiny $X_j$};
\draw [blue, ->-=0.5] (0,0) [partial ellipse=178:362:1.85 and 1.35]
node[black, pos=0.7, above] {\tiny $X_k$};
\end{scope}
\draw [line width=0.6cm, brown!20!white] (0,0) [partial ellipse=180:360:2 and 1.5];
\begin{scope}[shift={(4.5,0)}]
\draw [red, dashed](0,0) [partial ellipse=0:180:0.4 and 0.25];
\draw [white, line width=4pt] (0,0) [partial ellipse=290:270:0.4 and 0.25];
\draw [red] (0,0) [partial ellipse=260:360:0.4 and 0.25];
\draw [red] (0,0) [partial ellipse=180:230:0.4 and 0.25];
\end{scope}
\end{tikzpicture}}}}$
& $A_H=\bigoplus_{h\in H}h,\,\,\,q_K|_H=1.$ Condensation algebra in the single pointed theory; via full center it gives doubled Alterfold boundary data.

\\\\

Lagrangian algebra in $\qquad \qquad \mathcal C_K\boxtimes \mathcal C_{K'}^{\mathrm{rev}}$ &
$\qquad\qquad\qquad\vcenter{\hbox{\scalebox{0.3}{%
\begin{tikzpicture}[xscale=0.8, yscale=0.6]
\path [fill=brown!20!white](-0.65,-3) rectangle (0.65,3);
\begin{scope}[shift={(0,3)}]
\path [fill=brown!20!white] (0,0) [partial ellipse=0:180:0.6 and 0.3];
\draw (0,0) [partial ellipse=0:360:0.6 and 0.3];
\end{scope}
\draw (-0.6,3)--(-0.6,0) (0.6,3)--(0.6,0);
\draw (-0.6,-3)--(-0.6,0) (0.6,-3)--(0.6,0);
\begin{scope}[shift={(0,-3)}]
\path [fill=brown!20!white] (0,0) [partial ellipse=180:360:0.6 and 0.3];
\draw [dashed](0,0) [partial ellipse=0:180:0.6 and 0.3];
\draw (0,0) [partial ellipse=180:360:0.6 and 0.3];
\end{scope}
\draw[blue, ->-=0.75] (0,0.8) [partial ellipse=0:360:0.6 and 0.3];
\draw[blue, ->-=0.75] (0,-0.8) [partial ellipse=0:360:0.6 and 0.3];
\draw (0.8,0.8) node{\tiny $M$};
\draw (0.8,-0.8) node{\tiny $N$};
\draw (0,1.6) node{\tiny $\mathcal D$};
\draw (0,-0.1) node{\tiny $\mathcal E$};
\draw (0,-2) node{\tiny $\mathcal D$};
\end{tikzpicture}}}}$
=
$\vcenter{\hbox{\scalebox{0.35}{%
\begin{tikzpicture}[xscale=0.8, yscale=0.6]
\path [fill=brown!20!white](-0.65,-3) rectangle (0.65,3);
\begin{scope}[shift={(0,3)}]
\path [fill=brown!20!white] (0,0) [partial ellipse=0:180:0.6 and 0.3];
\draw (0,0) [partial ellipse=0:360:0.6 and 0.3];
\end{scope}
\draw (-0.6,3)--(-0.6,0) (0.6,3)--(0.6,0);
\draw (-0.6,-3)--(-0.6,0) (0.6,-3)--(0.6,0);
\begin{scope}[shift={(0,-3)}]
\path [fill=brown!20!white] (0,0) [partial ellipse=180:360:0.6 and 0.3];
\draw [dashed](0,0) [partial ellipse=0:180:0.6 and 0.3];
\draw (0,0) [partial ellipse=180:360:0.6 and 0.3];
\end{scope}
\draw (0,0) node{\tiny $\mathcal D$};
\end{tikzpicture}}}}$
&
Gapped domain wall/topological interface. Supported on a Lagrangian subgroup $M\subset G_K\oplus G_{K'}$ for $q_K\oplus q_{K'}^{-1}$.
\\\\
\bottomrule
\end{tabular}
\caption{Dictionary between Alterfold tube-category data, the doubled pointed modular category $\mathcal Z(\mathcal C_K)\simeq \mathcal C_K\boxtimes \mathcal C_K^{\mathrm{rev}}$, and the corresponding Abelian Chern--Simons interpretation. In the center case $\mathcal C_K\simeq \mathcal Z(\mathcal A)$, the single-copy dictionary is recovered by suppressing the reversed factor.}
\label{tab:alterfold-cs-dictionary}
\end{table}
\newpage

\subsection{Toward the non-Abelian case}
The non-Abelian case is subtler. The classification of Abelian Chern--Simons theories by finite quadratic modules does not extend to general non-Abelian theories, and even the modular data $S$ and $T$ need not determine a modular tensor category: inequivalent modular categories can have identical modular data~\cite{MignardSchauenburg}, while higher-genus or punctured-surface mapping-class-group representations can distinguish such examples~\cite{wenwen2019}. The rigorous statement supplied by Alterfold theory is always
$$
\mathbb V_{\mathcal A}\big|_{\mathrm{Cob}^B}\cong Z^{RT}_{\mathcal Z(\mathcal A)}.
$$
Thus Alterfold theory naturally produces Reshetikhin--Turaev theories associated to Drinfeld centers. To recover directly a single-copy non-Abelian Chern--Simons theory modeled by a modular tensor category $\mathcal C$, one would therefore need
$
\mathcal C\simeq \mathcal Z(\mathcal A)
$ for some fusion category $\mathcal A$. Equivalently, $\mathcal C$ must be Witt trivial, which is a strong condition and does not hold for a generic chiral modular tensor category. On the other hand, if one applies the construction to $\mathcal C$ itself, then modularity gives
$
\mathcal Z(\mathcal C)\simeq
\mathcal C\boxtimes\mathcal C^{\mathrm{rev}}.
$ Thus Alterfold theory naturally sees the doubled, orientation-reversal-invariant theory, precisely paralleling the doubled Abelian case above. This suggests two complementary lessons.

First, Alterfold theory should presently be viewed as a construction principle for the defect sector of Chern--Simons theories: centers, modular categories, Morita contexts, modular invariants, full centers, $\alpha$-induction, and condensations. These are exactly the structures that one expects to control defects, boundaries, walls, and condensation phenomena in non-Abelian topological field theories.

Second, the Abelian Chern--Simons theories discussed in this paper provide a clean testing ground. Because the relevant modular categories are pointed and completely classified by finite quadratic modules, one can identify explicitly when Alterfold theory recovers a single Chern--Simons theory and when it naturally recovers only the doubled theory. Any future extension of the Alterfold formalism to  non-Abelian Chern--Simons theory should recover these Abelian examples as the simplest cases. Thus,  results above
give a rigorous \emph{TQFT-level} bridge.
$$
\begin{tikzpicture}[baseline=(current bounding box.center)]
  \node (A) at (0,1.2) {\text{Alterfold}};
  \node (B) at (-0.6,0) {RT};
  \node (C) at (0.6,0) {CS};
  \node (D) at (0,0) {$\cong$};
  \draw[->] (A) -- node[right] {$_{\mathcal Z(\mathcal A)}$} (D);
\end{tikzpicture}
$$

\section{Gauge-theoretic interpretation and explicit defect calculations}
\label{sec:gauge-defect}

Sections~\ref{sec:tv}--\ref{sec:Alterfold} establish the center realizations,
classify condensations, boundaries and walls, and construct their defect
and Morita data. Here we retain those results and compute their Wilson-line
interpretation: endpoint amplitudes, a boundary associator, braided
attachment phases, lattice lifts, and the toric-code case.
The interpretation uses the transported defect theory of
Theorem~\ref{thm:defect-functor}; it does not constitute an independent
construction from a gauge-field action with defects.
Throughout, use the normalized pointed realization
$\mathcal C_K\simeq\operatorname{Vec}_{G_K}^{\omega_K,c_K}$ and the
associator convention of Section~\ref{sec:defects}. Write
$\omega=\omega_K$ and $c=c_K$ when the bulk is fixed.

\subsection{Endpoint amplitudes}

Let $A=A_{H,\mu}$ be the condensation algebra of
Proposition~\ref{prop:defects-etale}, and let
$F_H(X)=X\otimes A$ be its free right module. The space of vacuum
endpoints of a Wilson sector $W_x$ is
\begin{equation}
\label{eq:gauge-endpoint-space}
E_x:=\operatorname{Hom}_{(\mathcal C_K)_A}(F_H(X_x),A)
\cong\operatorname{Hom}_{\mathcal C_K}(X_x,A),
\qquad
\dim E_x=\begin{cases}1,&x\in H,\\0,&x\notin H.\end{cases}
\end{equation}

Thus an absorbed line ends at a defect-local operator, represented by a
map to $A$, rather than by a bulk map to $\mathbf1$. The endpoint fusion described by this equation is depicted in Figure~\ref{fig:condensation}. Choose normalized homogeneous endpoint bases $\varepsilon_h$.
Their fusion and change of normalization are
\begin{equation}
\label{eq:gauge-endpoint-bases}
\varepsilon_h\varepsilon_k=\mu(h,k)\varepsilon_{h+k},
\qquad
\varepsilon'_h=\nu(h)\varepsilon_h
\ \Longrightarrow\ 
\mu'(h,k)=\mu(h,k)\frac{\nu(h)\nu(k)}{\nu(h+k)}.
\end{equation}

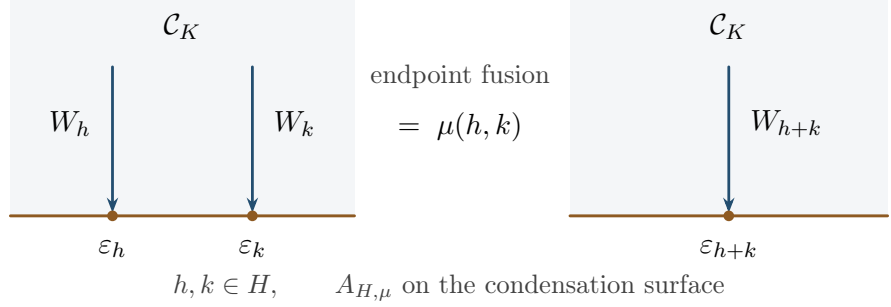
\begin{figure}[htbp]
\centering
\begin{tikzpicture}[defbase]
\fill[defbulk!5] (0,0) rectangle (4.55,2.9);
\fill[defbulk!5] (7.4,0) rectangle (11.6,2.9);
\draw[defedge] (0,0)--(4.55,0);
\draw[defedge] (7.4,0)--(11.6,0);
\node at (2.275,2.5) {$\mathcal C_K$};
\node at (9.5,2.5) {$\mathcal C_K$};
\draw[defline] (1.35,1.97)--(1.35,.04);
\draw[defline] (3.2,1.97)--(3.2,.04);
\node[left=4pt] at (1.35,1.22) {$W_h$};
\node[right=4pt] at (3.2,1.22) {$W_k$};
\node[defend] at (1.35,0) {};
\node[defend] at (3.2,0) {};
\node[below=5pt] at (1.35,0) {$\varepsilon_h$};
\node[below=5pt] at (3.2,0) {$\varepsilon_k$};
\node at (5.95,1.18) {$=\;\mu(h,k)$};
\node[defnote] at (5.95,1.85) {endpoint fusion};
\draw[defline] (9.5,1.97)--(9.5,.04);
\node[right=4pt] at (9.5,1.22) {$W_{h+k}$};
\node[defend] at (9.5,0) {};
\node[below=5pt] at (9.5,0) {$\varepsilon_{h+k}$};
\node[defnote] at (5.8,-.95) {$h,k\in H,\qquad A_{H,\mu}\text{ on the condensation surface}$};
\end{tikzpicture}
\caption{Fusion of condensed endpoints in chosen homogeneous bases:
$\varepsilon_h\varepsilon_k=\mu(h,k)\varepsilon_{h+k}$.
The left panel denotes the ordered fusion of two endpoint insertions,
including fusion of their incident Wilson lines. The coefficient
$\mu(h,k)$ is a basis-dependent fusion amplitude.}
\label{fig:condensation}
\end{figure}

The coherence conditions are precisely
\eqref{eq:defects-associativity}--\eqref{eq:defects-commutativity}.
Consequently $\mu$ records endpoint-fusion amplitudes in a chosen basis;
its coboundary freedom is endpoint normalization, not an additional
condensation invariant.

\subsection{An explicit boundary associator}

We use the free-module construction of
\cite[Theorem~A.4]{DavydovSimmons2018} to express the boundary
associator in the tensor convention adopted here.
Let $L\le G_K$ be Lagrangian, set $A=A_{L,\mu}$, and write
$\mathcal A_L=(\mathcal C_K)_A$.
Its pointed description and the invariance of its associator class
were established in \eqref{eq:defects-boundary-associator}.
The purpose below is to fix an explicit representative for the
subsequent defect calculations.

The notation of \cite[Appendix~A]{DavydovSimmons2018} corresponds to
ours by
$$
(A,B,\alpha,\eta,\gamma)
\longleftrightarrow
(G_K,L,\omega,\mu,\lambda).
$$
Their free-module tensorator moves the algebra factor past the
second object using $c_{A,Y}$. We retain instead the inverse-braiding
convention: a right $A$-module $(M,r_M)$ is equipped with the left
action
$$
A\otimes M
\xrightarrow{\,c_{M,A}^{-1}\,}
M\otimes A
\xrightarrow{\,r_M\,}
M
$$
when forming $\otimes_A$. Accordingly, we implement their
construction with this crossing and the normalizations specified
below.

Choose a normalized section $s:Q\to G_K$ of the quotient map,
where $Q=G_K/L$, and define
\begin{equation}
\label{eq:gauge-screening}
\begin{gathered}
\lambda(a,b)=s(a)+s(b)-s(a+b)\in L,\\
\lambda(a,b)+\lambda(a+b,d)
=\lambda(b,d)+\lambda(a,b+d).
\end{gathered}
\end{equation}
Thus $\lambda(a,b)$ is the charge absorbed when the chosen
representatives of two boundary sectors fuse.

Write $F_x=X_x\otimes A$, with homogeneous components indexed by
$h\in L$ and of degree $x+h$. In our convention, the free-module
tensorator
$J_{x,y}:F_x\otimes_A F_y\to F_{x+y}$ has coefficient
\begin{equation}
\label{eq:gauge-free-tensorator}
T_{x,y}(h,k)=
\frac{\omega(x,h,y+k)\,\omega(y,h,k)\,\mu(h,k)}
{\omega(h,y,k)\,\omega(x,y,h+k)\,c(y,h)}.
\end{equation}
This is obtained by reassociating
$(X_x\otimes X_h)\otimes(X_y\otimes X_k)$, moving $X_h$
past $X_y$ using $c_{X_y,X_h}^{-1}$, and multiplying the
algebra factors. The resulting map descends to the balanced
tensor product.

For $\ell\in L$, normalize the absorption isomorphism
$D_{z,\ell}:F_{z+\ell}\to F_z$ to send the component indexed
by $0$ to that indexed by $\ell$ with coefficient $1$.
Its coefficient from component $k$ to component $\ell+k$ is
then fixed by right $A$-linearity:
\begin{equation}
\label{eq:gauge-absorption}
D_{z,\ell}(k)=\omega(z,\ell,k)\mu(\ell,k).
\end{equation}
Following the quotient construction of
\cite[Appendix~A]{DavydovSimmons2018}, take
$$
B_a=F_{s(a)},\qquad
\rho_{a,b}
=D_{s(a+b),\lambda(a,b)}J_{s(a),s(b)}
:B_a\otimes_A B_b\longrightarrow B_{a+b}.
$$

\begin{prop}[Davydov--Simmons construction in the chosen tensor convention]
\label{prop:boundary-associator}
Let
$$
s_a=s(a),\quad s_{ab}=s(a+b),\quad
s_{bd}=s(b+d),\quad s_{abd}=s(a+b+d),
$$
and put
$$
h=\lambda(a,b),\qquad k=\lambda(b,d),\qquad
p=\lambda(a+b,d),\qquad r=\lambda(a,b+d).
$$
With the preceding tensorators and absorption maps,
$\mathcal A_L\simeq\operatorname{Vec}_{Q}^{\beta_{L,\mu}}$,
where the normalized associator cocycle is
\begin{equation}
\label{eq:gauge-explicit-beta}
\begin{aligned}
\beta_{L,\mu}(a,b,d)
={}&
\frac{\omega(s_a,s_b,s_d)\,\omega(s_{ab},s_d,h)}
{\omega(s_a,s_{bd},k)\,\omega(s_{ab},h,s_d)}
\\[1mm]
&\quad{}\cdot
\frac{\omega(s_{abd},r,k)}{\omega(s_{abd},p,h)}
\frac{c(s_d,h)\,\mu(r,k)}{\mu(p,h)}.
\end{aligned}
\end{equation}
In particular, if $\omega=1$ and $s$ is additive, then
$\beta_{L,\mu}=1$.
\end{prop}

\begin{proof}
We evaluate the comparison of quotient fusion maps used in
\cite[Appendix~A]{DavydovSimmons2018}, with the tensorator
\eqref{eq:gauge-free-tensorator}.
Writing $a_{B_a,B_b,B_d}$ for the module-category associativity
constraint, the scalar is determined by
$$
\rho_{a,b+d}
(\operatorname{id}\otimes_A\rho_{b,d})
\,a_{B_a,B_b,B_d}
=
\beta_{L,\mu}(a,b,d)\,
\rho_{a+b,d}
(\rho_{a,b}\otimes_A\operatorname{id}).
$$
Both composites are isomorphisms between the same simple modules.
Evaluate them on the class of the tensor of the three components
indexed by $0$. By
\eqref{eq:gauge-free-tensorator}--\eqref{eq:gauge-absorption},
their coefficients are, respectively,
$$
\frac{\omega(s_a,s_b,s_d)\,
      \omega(s_{abd},r,k)\mu(r,k)}
     {\omega(s_a,s_{bd},k)}
\quad\text{and}\quad
\frac{\omega(s_{ab},h,s_d)\,
      \omega(s_{abd},p,h)\mu(p,h)}
     {\omega(s_{ab},s_d,h)c(s_d,h)}.
$$
The target components agree because $r+k=p+h$, by
\eqref{eq:gauge-screening}. Their ratio gives
\eqref{eq:gauge-explicit-beta}.
Normalization follows from the normalized input data, and the
cocycle identity follows from the pentagon for $\mathcal A_L$.
If $s$ is additive, then $h=k=p=r=0$, so the formula reduces
to $\omega(s_a,s_b,s_d)$.
\end{proof}

For the fixed tensor convention, changes of the auxiliary choices
have the usual cohomological effect. Rescaling the fusion maps by
a normalized $2$-cochain $\gamma:Q^2\to U(1)$ gives
\begin{equation}
\label{eq:gauge-beta-basis}
\rho'_{a,b}=\gamma(a,b)\rho_{a,b}
\quad\Longrightarrow\quad
\beta'(a,b,d)=\beta(a,b,d)
\frac{\gamma(b,d)\gamma(a,b+d)}
{\gamma(a+b,d)\gamma(a,b)}.
\end{equation}
A change of section $s'=s+t$, with normalized $t:Q\to L$,
satisfies
\begin{equation}
\label{eq:gauge-section-change}
\lambda'=\lambda+\delta t,\qquad
(\delta t)(a,b)=t(a)+t(b)-t(a+b).
\end{equation}
Transporting the fusion maps along the corresponding absorption
isomorphisms changes $\beta$ by a coboundary. The same holds for
a change of endpoint bases in \eqref{eq:gauge-endpoint-bases}.
Thus, with the tensor convention fixed, the representative depends
on these choices, whereas
$[\beta_L]\in H^3(Q;U(1))$ depends only on the embedded Lagrangian
subgroup and the bulk quadratic data.

\subsection{Braided attachment and wall transmission}

The central transport and junction spaces are those of
Section~\ref{sec:defects}; the $\alpha$-inductions and full-center
multiplicities are those of Section~\ref{sec:Alterfold}.
Their additional Wilson-line content can be expressed locally.
For $\alpha_H^\pm(X_x)$, the two right actions on $A\otimes X_x$
use $c_{X_x,A}$ and $c_{A,X_x}^{-1}$, respectively.
On a homogeneous charge $h\in H$ their relative factor is
\begin{equation}
\label{eq:gauge-alpha-phase}
\frac{c_K(x,h)}{c_K(h,x)^{-1}}=b_K(x,h).
\end{equation}
Thus the distinction between the two attachments is the mutual
Aharonov--Bohm phase. Their protected intertwiner multiplicity
$z^H_{jk}=\dim\operatorname{Hom}_{\mathcal D_H}
(\alpha_H^+(X_j),\alpha_H^-(X_k))$
has the transmission interpretation established by folding;
in the standard doubled basis its coefficient is $z^H_{a,-b}$,
as in Section~\ref{sec:Alterfold}.

For a wall with transmission support $M_{12}$, a pair
$(x,y)\in M_{12}$ means that the folded Wilson sector
$X_x\boxtimes(X_y^*)^{\mathrm{rev}}$ can be absorbed.
The charge-matching relation for two successive walls is
\begin{equation}
\label{eq:gauge-wall-relation}
M_{23}\circ M_{12}
=\{(x,z):\exists y,\ (x,y)\in M_{12},\ (y,z)\in M_{23}\}.
\end{equation}
It forgets intermediate-channel multiplicities and coherence data.
The actual surface fusion is the relative tensor product of
bimodule categories from Section~\ref{sec:defects}; for chosen center
presentations $\mathcal C_i\simeq\mathcal Z(\mathcal A_i)$, it is
$\mathcal M_{12}\boxtimes_{\mathcal A_2}\mathcal M_{23}$.

\subsection{Lattice realizations of symmetry walls}

Theorem~\ref{thm:defect-classification}(4) already identifies invertible
walls with isometries of finite quadratic modules. A lattice isometry
$U\in\mathcal O(\Lambda,K)$ realizes such a wall on the chosen gauge
variables. Its action on Wilson charges is contragredient:
\begin{equation}
\label{eq:gauge-lattice-lift}
u_U([x])=[U^{-T}x].
\end{equation}
Indeed,
$$
KU=U^{-T}K,\qquad
U^{-1}K^{-1}U^{-T}=K^{-1},
$$
so $U^{-T}$ preserves $K\Lambda$ and
$$
q_K(u_U([x]))
=\exp\!\bigl(\pi i x^T U^{-1}K^{-1}U^{-T}x\bigr)=q_K([x]).
$$
Moreover $(UV)^{-T}=U^{-T}V^{-T}$; hence
$U\mapsto u_U$ is a homomorphism
$\mathcal O(\Lambda,K)\to\mathcal O(G_K,q_K)$.
Its image consists of the symmetries realized by lattice isometries
of this presentation. A finite-quadratic isometry outside the image
still defines a quantum topological wall by the cited theorem;
\eqref{eq:gauge-lattice-lift} does not provide its gauge-field realization. 	Figure~\ref{fig:symmetry-wall} depicts the transmission of a Wilson charge across an invertible symmetry wall.

\begin{figure}[htbp]
\centering
\begin{tikzpicture}[defbase]
\fill[defbulk!5] (0,0) rectangle (5.25,3);
\fill[defwall!10] (5.25,0) rectangle (6.35,3);
\fill[defbulk!5] (6.35,0) rectangle (11.6,3);
\draw[defedge] (5.25,0)--(5.25,3);
\draw[defedge] (6.35,0)--(6.35,3);
\node at (2.6,2.5) {$\mathcal C_K$};
\node at (9,2.5) {$\mathcal C_K$};
\node at (5.8,2.45) {$\mathfrak D_u$};
\draw[defline] (.6,1.15)--(5.25,1.15);
\draw[draw=defbulk,line width=1.05pt] (5.25,1.15)--(6.35,1.15);
\draw[defline] (6.35,1.15)--(11,1.15);
\node[above=5pt] at (2.6,1.15) {$X_x$};
\node[above=5pt] at (9,1.15) {$X_{u(x)}$};
\node[defnote] at (5.8,-.48) {$\Gamma_u=\{(x,u(x)):x\in G_K\},\qquad u\in\mathcal O(G_K,q_K)$};
\node[defnote] at (5.8,-1.05) {For a lattice lift: $\ u([x])=[U^{-T}x],\quad U^T K U=K.$};
\end{tikzpicture}
\caption{An invertible symmetry wall transports a Wilson sector
$X_x$ to $X_{u(x)}$. Its support is the graph of an isometry $u$.
The lattice formula applies when $u$ admits a lift $U$ in the chosen
$K$-matrix presentation.}
\label{fig:symmetry-wall}
\end{figure}
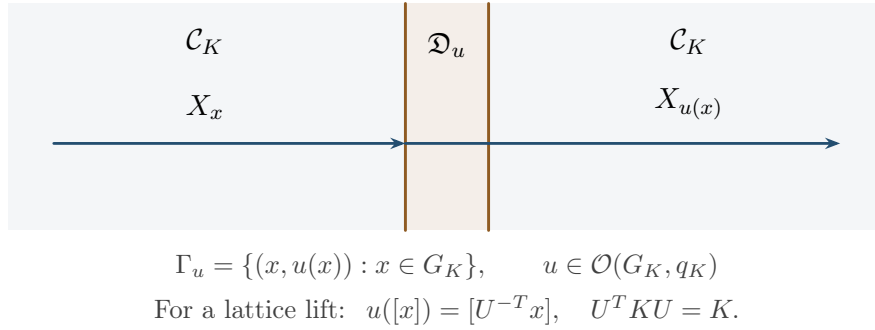

\section{The toric code and semion as examples of the full framework}
\label{subsec:toric-code-defects}

The toric code provides a concrete example in which the principal constructions of this paper can be evaluated explicitly. Its bulk excitations and surface degeneracies are familiar from Kitaev's lattice model \cite{kitaev2003anyons}, while its gapped boundaries and interfaces admit the categorical description of Kitaev and Kong \cite{kitaev2012models}. We recover these results together with partition functions, bordism maps, and full-center states from a single BF lattice. The example thereby exhibits the compatibility of the gauge-theoretic, categorical, and Alterfold descriptions, including their normalizations and gluing laws.

\subsection{The toric code theory}
We use the construction of Sections~\ref{sec:Abelian-CS}--\ref{sec:Alterfold} to the BF lattice
$$
K_{\mathrm{TC}}=
\begin{pmatrix}
0&2\\
2&0
\end{pmatrix},
\qquad
|\det K_{\mathrm{TC}}|=4,
\qquad
\sigma(K_{\mathrm{TC}})=0.
$$
With $\mathbb T^2=\mathbb R^2/\mathbb Z^2$, write
$$
Z_{\mathrm{TC}}
:=Z^{CS}_{\mathbb T^2,K_{\mathrm{TC}}}
=Z^{BF}_2,
\qquad
\mathcal C_{\mathrm{TC}}
:=\mathcal C(G_{K_{\mathrm{TC}}},q_{K_{\mathrm{TC}}}).
$$
Abbreviate the discriminant data by
$G=G_{K_{\mathrm{TC}}}$, $q=q_{K_{\mathrm{TC}}}$, and
$b=b_{K_{\mathrm{TC}}}$. Identifying
$$
G=\mathbb Z_2^2,
\qquad
e=(1,0),\quad m=(0,1),\quad \epsilon=e+m,
$$
with coordinates added modulo two, we obtain
\begin{equation}
\label{eq:gauge-toric-data}
\begin{gathered}
q(a,b)=(-1)^{ab},
\qquad
b\bigl((a,b),(c,d)\bigr)=(-1)^{ad+bc},
\\
\omega=1,
\qquad
c\bigl((a,b),(c,d)\bigr)=(-1)^{ad}.
\end{gathered}
\end{equation}
Here $\omega$ and $c$ specify the chosen pointed braided realization,
and $X_x$ denotes the simple Wilson sector of charge $x\in G$.
All simple objects have dimension one, so the total quantum dimension
is $\sqrt{|G|}=2$. Throughout, we use the normalization of
Theorem~\ref{thm:bf-tqft-data}, for which
$Z_{\mathrm{TC}}(S^3)=1/2$.

\subsection*{State spaces, bordisms, and gluing}

For a connected closed surface $\Sigma_g$ of genus $g$, write
$\mathcal H_{\mathrm{TC}}(\Sigma_g)=Z_{\mathrm{TC}}(\Sigma_g)$,
suppressing the rational polarization as in
Proposition~\ref{prop:toral-hilbert-space}. Choosing an origin for the
Bohr--Sommerfeld torsor gives
\begin{equation}
\label{eq:tc-state-spaces}
\mathcal H_{\mathrm{TC}}(\Sigma_g)
\cong\mathbb C[G^g],
\qquad
\dim\mathcal H_{\mathrm{TC}}(\Sigma_g)=4^g,
\end{equation}
where $\mathbb C[G^g]$ is the complex vector space with basis indexed
by $G^g$. In particular, the dimensions for the sphere, torus, and
genus-two surface are $1$, $4$, and $16$. The formula $4^g$ agrees
with the ground-state degeneracy of Kitaev's toric-code Hamiltonian
\cite{kitaev2003anyons}; here it follows from the general quantization
formula for the Abelian Chern--Simons state space.

Disjoint unions give tensor products, and orientation reversal gives
dual spaces. Moreover, $\sigma(K_{\mathrm{TC}})=0$, so the Maslov
phase in the BKS comparison law of Section~\ref{sec:Abelian-CS} is trivial.

For explicit bordism calculations, we also use the finite
$\mathbb Z_2$-field description of the Turaev--Viro realization in
Section~\ref{sec:tv}. A flat field on a manifold $Y$ is specified,
up to equivalence, by a class in $H^1(Y;\mathbb Z_2)$.
For a connected surface this gives
\begin{equation}
\label{eq:tc-flat-state-space}
\begin{gathered}
\mathcal H_{\mathrm{TC}}(\Sigma_g)
\cong
\operatorname{Fun}\bigl(H^1(\Sigma_g;\mathbb Z_2),\mathbb C\bigr),
\\
\langle f_1,f_2\rangle
=
\frac12
\sum_{a\in H^1(\Sigma_g;\mathbb Z_2)}
\overline{f_1(a)}f_2(a).
\end{gathered}
\end{equation}
Here $\operatorname{Fun}(E,\mathbb C)$ denotes the space of
complex-valued functions on a finite set $E$. The factor $1/2$
accounts for the two constant gauge transformations on a connected
manifold, in agreement with the finite-gauge normalization of
Freed and Quinn \cite[Section~5]{Freed:1991}.
Since $|H^1(\Sigma_g;\mathbb Z_2)|=2^{2g}$, this description has
dimension $4^g$. After choosing cycles on $\Sigma_g$, its basis
is related to the polarized basis by a finite Fourier transform.

Let $X:\Sigma_-\to\Sigma_+$ be a connected bordism with nonempty
connected outgoing boundary, and let
$$
r_\pm:
H^1(X;\mathbb Z_2)
\longrightarrow
H^1(\Sigma_\pm;\mathbb Z_2)
$$
be the restriction maps. In the function description,
\begin{equation}
\label{eq:tc-bordism-map}
\bigl(Z_{\mathrm{TC}}(X)f\bigr)(a_+)
=
\sum_{\substack{
u\in H^1(X;\mathbb Z_2)\\
r_+(u)=a_+
}}
f\bigl(r_-(u)\bigr).
\end{equation}
Thus the bordism map sums over flat fields extending the prescribed outgoing field. The factor associated with the constant gauge transformations on $X$ is canceled by the normalization at the connected outgoing boundary. For empty incoming boundary, take $f=1$. Consequently, a connected manifold $X$ with connected boundary $\Sigma$ determines the vector
\begin{equation}
\label{eq:tc-bordism-state}
Z_{\mathrm{TC}}(X)(a)
=
\#\{u\in H^1(X;\mathbb Z_2):r_X(u)=a\},
\end{equation}
where $r_X$ is restriction and $\#$ denotes cardinality.
This is the extension-counting formula of
\cite[Section~5]{Freed:1991}, expressed here in the
normalization inherited from Theorem~\ref{thm:bf-tqft-data}.
Together with \eqref{eq:tc-flat-state-space}, it realizes the gluing
law of Proposition~\ref{prop:cylinder-gluing}.

For a genus-$g$ handlebody $H_g$, define
$$
\mathcal L_{H_g}
:=
\operatorname{im}\bigl(
H^1(H_g;\mathbb Z_2)
\longrightarrow H^1(\Sigma_g;\mathbb Z_2)
\bigr).
$$
Restriction is injective and $|\mathcal L_{H_g}|=2^g$.
Writing $\mathbf1_E$ for the indicator function of a subset $E$,
we therefore have
$Z_{\mathrm{TC}}(H_g)=\mathbf1_{\mathcal L_{H_g}}$ and
\begin{equation}
\label{eq:tc-handlebody-gluing}
\|Z_{\mathrm{TC}}(H_g)\|^2
=
2^{g-1}
=
Z_{\mathrm{TC}}\!\left(
\mathop{\#}^{g}(S^1\times S^2)
\right).
\end{equation}
The connected sum on the right is the double of $H_g$; for $g=0$
it is understood as $S^3$. The cylinder acts as the identity, so
taking its trace also gives
$Z_{\mathrm{TC}}(\Sigma_g\times S^1)=4^g$.

\subsection*{Closed-manifold partition functions}

Let $M$ be closed, connected, and oriented. Put $\mathsf T_M=\operatorname{Tors}H_1(M;\mathbb Z),$ and let
$\ell_M:\mathsf T_M\times\mathsf T_M\to\mathbb Q/\mathbb Z$
be the nonsingular torsion linking pairing.
In a nondegenerate surgery presentation, this pairing is represented
by the inverse linking matrix modulo integers.
The BF Gauss sum of Theorem~\ref{thm:bf-tqft-data} reduces to
\begin{equation}
\label{eq:tc-partition-gauss}
Z_{\mathrm{TC}}(M)
=
\frac{2^{b_1(M)-1}}{|\mathsf T_M|}
\sum_{u,v\in\mathsf T_M}
\exp\!\bigl(4\pi i\ell_M(u,v)\bigr)
=
2^{b_1(M)-1}|\mathsf T_M[2]|,
\end{equation}
where
$b_1(M)=\dim H^1(M;\mathbb R)$ and
$\mathsf T_M[2]=\{u\in\mathsf T_M:2u=0\}$.
Indeed, character orthogonality makes the sum over $v$ vanish
unless $2u=0$, in which case it equals $|\mathsf T_M|$.
The universal coefficient theorem then yields
\begin{equation}
\label{eq:tc-partition-count}
Z_{\mathrm{TC}}(M)
=
\frac{|H^1(M;\mathbb Z_2)|}{2}.
\end{equation}
This is precisely the untwisted finite-gauge partition function
of Freed and Quinn for gauge group $\mathbb Z_2$
\cite[Eq.~(5.14)]{Freed:1991}.
Thus the surgery calculation and the finite-field description give
the same closed-manifold invariant, including its normalization.
For disconnected closed $M$, the denominator becomes
$|H^0(M;\mathbb Z_2)|=2^{b_0(M)}$, as required by monoidality.

Equivalently, if $M_B$ is obtained by surgery on an $r$-component
link with integral linking matrix $B$, then
\begin{equation}
\label{eq:tc-surgery-mod-two}
Z_{\mathrm{TC}}(M_B)
=
2^{-r-1}
\sum_{a,b\in(\mathbb Z_2)^r}
(-1)^{a^TBb}
=
\frac{|\ker(B\bmod2)|}{2}.
\end{equation}
The first equality is the finite-charge surgery evaluation; the
second follows from character orthogonality.
The surgery presentation of first homology identifies
$|\ker(B\bmod2)|$ with $|H^1(M_B;\mathbb Z_2)|$.

For a lens space $L(p,q)$, with $p>0$ and $\gcd(p,q)=1$,
one has $H_1(L(p,q);\mathbb Z)\cong\mathbb Z/p\mathbb Z$.
Hence
\begin{equation}
\label{eq:tc-partition-examples}
\begin{gathered}
Z_{\mathrm{TC}}(S^3)=\tfrac12,
\qquad
Z_{\mathrm{TC}}(S^1\times S^2)=1,
\qquad
Z_{\mathrm{TC}}(T^3)=4,
\\
Z_{\mathrm{TC}}(\Sigma_g\times S^1)=4^g,
\qquad
Z_{\mathrm{TC}}(L(p,q))
=\frac{\gcd(p,2)}2,
\end{gathered}
\end{equation}
where $T^3=S^1\times S^1\times S^1$.
In particular, $Z_{\mathrm{TC}}(\mathbb RP^3)=1$.
For connected $M,N$, the connected-sum relation is
$$
Z_{\mathrm{TC}}(M\#N)
=
2Z_{\mathrm{TC}}(M)Z_{\mathrm{TC}}(N),
$$
consistently with both sphere gluing and
\eqref{eq:tc-partition-count}.

\subsection*{Torus states, modular transformations, and Wilson observables}

Choose meridian and longitude holonomies
$(\alpha,\beta)\in\mathbb Z_2^2$ on the torus.
Let $|x\rangle$ be the state of a solid torus with core Wilson line
$X_x$; the empty solid torus gives $|0\rangle$.
In the function realization these states are
$$
|(a,b)\rangle(\alpha,\beta)
=
\delta_{\alpha,b}(-1)^{a\beta},
\qquad (a,b)\in G,
$$
where $\delta_{u,v}$ is the Kronecker delta.
They are orthonormal for \eqref{eq:tc-flat-state-space}. Let $S$ exchange the meridian and longitude with the
orientation-preserving sign convention, and let $T$ be the
Dehn-twist operator. In the basis
$(|0\rangle,|e\rangle,|m\rangle,|\epsilon\rangle)$,
\begin{equation}
\label{eq:tc-modular-data}
S_{xy}=\tfrac12 b(x,y),
\qquad
S=\frac12
\begin{pmatrix}
1&1&1&1\\
1&1&-1&-1\\
1&-1&1&-1\\
1&-1&-1&1
\end{pmatrix},
\qquad
T=\operatorname{diag}(1,1,1,-1).
\end{equation}
These coincide with the standard toric-code modular matrices
\cite[Supplemental Material, Section~I]{LanWangWen2015TC}.
The normalized Gauss sum
$\frac12\sum_{x\in G}q(x)=1$ gives a trivial framing-anomaly
factor, and direct multiplication yields
$S^2=(ST)^3=I_4$.  These torus data also illustrate the recent reconstruction theorem of \cite{Galviz7}, which shows that, for Abelian Chern--Simons topological orders, the normalized torus Berry data determine the finite quadratic module $(G,q)$, and hence the full all-genus extended TQFT. Thus, for the toric code, the genus-one data displayed in \eqref{eq:tc-modular-data} already determine the complete Abelian topological order.
For an orientation-preserving diffeomorphism
$f:\Sigma_g\to\Sigma_g$, define
$$
M_f
=
(\Sigma_g\times[0,1])/((x,1)\sim(f(x),0)),
$$
and let $\rho(f)$ be the operator assigned to its mapping cylinder.
The gluing law gives
\begin{equation}
\label{eq:tc-mapping-torus}
Z_{\mathrm{TC}}(M_f)
=
\operatorname{Tr}\rho(f)
=
\bigl|
\ker\bigl(
f^*-\operatorname{id}:
H^1(\Sigma_g;\mathbb Z_2)
\to H^1(\Sigma_g;\mathbb Z_2)
\bigr)
\bigr|.
\end{equation}
Indeed, in the function basis $\rho(f)$ permutes the cohomology
classes, so its trace counts fixed classes.
For the torus generators,
$\operatorname{Tr}S=\operatorname{Tr}T=2$.

Let $\mathscr L\subset S^3$ be a framed oriented link whose
components have charges $x_i=(a_i,b_i)$, framings $f_i$, and
pairwise linking numbers $\ell_{ij}$.
Writing $Z_{\mathrm{TC}}(S^3;\mathscr L)$ for the partition function
with these Wilson insertions, the ribbon evaluation gives
\begin{equation}
\label{eq:tc-wilson-link}
\frac{Z_{\mathrm{TC}}(S^3;\mathscr L)}
     {Z_{\mathrm{TC}}(S^3)}
=
(-1)^{
\sum_i f_i a_i b_i+
\sum_{i<j}\ell_{ij}(a_i b_j+a_j b_i)
}.
\end{equation}
Thus an $e$--$m$ Hopf link has normalized expectation $-1$, and
a unit framing twist of $\epsilon$ contributes $-1$.
These are the mutual electric--magnetic statistics and fermionic
composite of Kitaev's model \cite{kitaev2003anyons}.

For a genus-$g$ surface with incoming marked Wilson insertions
$x_1,\ldots,x_n$, denote its state space by
$\mathcal H_{\mathrm{TC}}(\Sigma_g;x_1,\ldots,x_n)$.
Character orthogonality gives
\begin{equation}
\label{eq:tc-punctured-space}
\dim\mathcal H_{\mathrm{TC}}(\Sigma_g;x_1,\ldots,x_n)
=
4^g\,\delta_{x_1+\cdots+x_n,0}.
\end{equation}
Indeed, the dimension is
$2^{2g-2}\sum_{y\in G}b(x_1+\cdots+x_n,y)$.
In particular, a sphere junction space is one-dimensional exactly
when the incident charges fuse to the vacuum.

\subsection*{Condensations, boundary states, and junctions}

The only nonzero isotropic subgroups are
$$
L_e=\{0,e\},
\qquad
L_m=\{0,m\},
\qquad
L_a^\perp=L_a.
$$
Proposition~\ref{prop:defects-etale} therefore gives exactly three
connected condensation algebras:
$$
\mathbf1,
\qquad
A_e=\mathbf1\oplus X_e,
\qquad
A_m=\mathbf1\oplus X_m.
$$
The latter two are Lagrangian and define the elementary gapped
boundaries. Their local-module categories are
$\mathrm{Vec}_{\mathbb C}$, so neither leaves a nontrivial
deconfined bulk. Both multiplication cochains can be chosen equal to $1$. The quotient sections generated by $m$ and $e$, respectively,
are additive. Proposition~\ref{prop:boundary-associator} then gives
\begin{equation}
\label{eq:gauge-toric-boundaries}
\beta_e=\beta_m=1,
\qquad
\mathcal A_e\simeq\mathcal A_m
\simeq\operatorname{Vec}_{\mathbb Z_2},
\end{equation}
where
$\mathcal A_a=(\mathcal C_{\mathrm{TC}})_{A_a}$ is the
boundary-line category.
Their bulk-to-boundary functors nevertheless differ:
$$
\begin{array}{c|cc}
&\text{vacuum boundary sector}&\text{nontrivial boundary sector}\\
\hline
L_e&0,e&m,\epsilon\\
L_m&0,m&e,\epsilon
\end{array}
$$
In the terminology of Kitaev and Kong
\cite[Section~2]{kitaev2012models}, these are the rough and smooth
boundaries, respectively. The displayed identifications reproduce
their bulk-to-boundary maps and the equivalence of the two abstract
boundary-line categories.

For $a\in\{e,m\}$, define the torus condensation vector by
$|\mathfrak b_a\rangle=\sum_{x\in L_a}|x\rangle$.
Then
\begin{equation}
\label{eq:tc-condensation-states}
\begin{gathered}
|\mathfrak b_e\rangle=|0\rangle+|e\rangle,
\qquad
|\mathfrak b_m\rangle=|0\rangle+|m\rangle,
\\
S|\mathfrak b_a\rangle
=
T|\mathfrak b_a\rangle
=
|\mathfrak b_a\rangle.
\end{gathered}
\end{equation}
These vectors carry the algebra multiplicities, without unit-norm
rescaling. In particular, they differ from the empty-solid-torus
state $|0\rangle$. Let $\mathcal H_{\mathrm{TC}}(\mathrm{Ann};a,b)$ be the state
space of a spatial annulus with gapped boundary conditions $a,b$.
Its dimension and the condensation-vector overlap are
\begin{equation}
\label{eq:tc-annulus}
\dim\mathcal H_{\mathrm{TC}}(\mathrm{Ann};a,b)
=
\langle\mathfrak b_a|\mathfrak b_b\rangle
=
|L_a\cap L_b|
=
\begin{cases}
2,&a=b,\\
1,&a\ne b.
\end{cases}
\end{equation}
This agrees with the standard cylinder degeneracies
\cite[Supplemental Material, Section~I]{LanWangWen2015TC}.

For $u,v\in G/L_a$, let $B_u,B_v$ be the corresponding simple
boundary lines and write $F_a(X_x)=B_{[x]}$, with $[x]=x+L_a$.
The mixed junction space satisfies
\begin{equation}
\label{eq:tc-boundary-junction}
\dim\operatorname{Hom}_{\mathcal A_a}
\bigl(F_a(X_x)\otimes B_u,B_v\bigr)
=
\delta_{[x]+u,v}.
\end{equation}
The boundary-changing categories of
\eqref{eq:defects-boundary-lines} are, as linear categories,
$$
\mathcal W_{e,m}\simeq\mathcal W_{m,e}
\simeq\mathrm{Vec}_{\mathbb C}.
$$
In the reference category $\operatorname{Vec}_{\mathbb Z_2}$,
the two boundaries are represented by its action on itself and
its action on $\mathrm{Vec}_{\mathbb C}$, where both simple objects
act as the identity functor. The module-functor category in either
direction has one simple object.

Let $\sigma_{ab}$ be the unique simple line changing boundary
$a$ to boundary $b$, and let $\mathbf1_a,t_a$ be the two simple
lines on boundary $a$. Then
\begin{equation}
\label{eq:tc-changing-lines}
\sigma_{ba}\circ\sigma_{ab}
\simeq
\mathbf1_a\oplus t_a,
\qquad a\ne b.
\end{equation}
In the regular module model, the composite forgets the grading
and then applies the adjoint, placing the resulting vector space
in both degrees. It therefore decomposes into the original graded
space and its degree shift. The opposite composite decomposes
into the two simple module endofunctors on
$\mathrm{Vec}_{\mathbb C}$. These descriptions recover the
equal-boundary and mixed-boundary channel counts in
\eqref{eq:tc-annulus}.

The local formulas of Section~\ref{sec:gauge-defect} also become
explicit. For the boundary $A_a$, write
$$
E_x^{(a)}
=\operatorname{Hom}_{\mathcal A_a}(F_a(X_x),A_a).
$$
Equations~\eqref{eq:gauge-endpoint-space} and
\eqref{eq:gauge-endpoint-bases} give
$$
\dim E_x^{(a)}=\mathbf1_{L_a}(x),\qquad
\varepsilon_h\varepsilon_k=\varepsilon_{h+k}
\quad(h,k\in L_a).
$$
In particular, $\varepsilon_e^2=\varepsilon_0$ on the electric
boundary and $\varepsilon_m^2=\varepsilon_0$ on the magnetic boundary.

Choose quotient representatives $s_e(u)=um$ and $s_m(u)=ue$,
where $u\in\mathbb Z_2$. Both sections are additive, so
$\lambda_e=\lambda_m=0$. Nevertheless, the free-module tensorators
need not have identical coefficients. For $h,k\in\mathbb Z_2$,
\eqref{eq:gauge-free-tensorator} gives
$$
T_{um,vm}(he,ke)=1,\qquad
T_{ue,ve}(hm,km)=(-1)^{vh},
$$
while \eqref{eq:gauge-absorption} gives $D_{z,\ell}(k)=1$.
Substitution into \eqref{eq:gauge-explicit-beta} yields
$\beta_e=\beta_m=1$.

For a bulk charge $x=(u,v)$, the relative attachment phases in
\eqref{eq:gauge-alpha-phase} are
$$
b(x,e)=(-1)^v,\qquad b(x,m)=(-1)^u.
$$
Thus the two boundaries have equivalent abstract fusion categories,
but distinct braided attachments. These phases recover
$L_e^\perp=L_e$ and $L_m^\perp=L_m$ directly.

\subsection*{The six walls and their transmission data}
\begin{prop}
\label{prop:toric-six-walls}
Let $V=G\oplus G$ carry the folded quadratic form
$(x,y)\mapsto q(x)q(y)^{-1}$.
For an isometry $u:G\to G$, put
$\Gamma_u=\{(x,u(x)):x\in G\}$.
Then the six Lagrangian subgroups of $V$ are
\begin{equation}
\label{eq:gauge-toric-walls}
\Gamma_{\mathrm{id}},
\qquad
\Gamma_s,
\qquad
L_a\times L_b
\quad(a,b\in\{e,m\}),
\end{equation}
where $s(e)=m$, $s(m)=e$, and $s(\epsilon)=\epsilon$.
The two graph walls are invertible; the four product walls
factor through the vacuum and are noninvertible.
\end{prop}

\begin{proof}
A Lagrangian subgroup $M$ has order four.
Let $\pi_1:M\to G$ be projection to the first factor.
If $\pi_1$ is injective, $M$ is the graph of a homomorphism
$u:G\to G$. Isotropy gives $q\circ u=q$, so $u$ preserves the
nondegenerate pairing and is injective. Since $\epsilon$ is the
unique element with quadratic value $-1$, the only possibilities
are $u=\mathrm{id}$ and $u=s$.

Otherwise $\ker\pi_1$ has order two: an order-four kernel would
give the nonisotropic subgroup $\{0\}\times G$.
Write $\ker\pi_1=\{0\}\times L_b$.
Orthogonality forces the second coordinate of every element of
$M$ to lie in $L_b^\perp=L_b$, and isotropy forces its first
projection to be $L_a$ for $a=e$ or $m$.
Hence $M\subset L_a\times L_b$, with equality by the orders.

Conversely, each displayed subgroup is isotropic of order four.
The assertions about invertibility and vacuum factorization
follow from Theorem~\ref{thm:defect-classification}(4) and
composition of the corresponding vacuum boundaries.
\end{proof}

For $a,b\in\{e,m\}$, let $w_a$ be the coefficient row of
$|\mathfrak b_a\rangle$ and let $W_{ab}$ be the transmission
matrix of the wall supported on $L_a\times L_b$.
Rows label incoming charges and columns label outgoing charges,
both ordered as $(0,e,m,\epsilon)$. Then
\begin{equation}
\label{eq:gauge-toric-matrices}
\begin{gathered}
w_e=(1,1,0,0),
\qquad
w_m=(1,0,1,0),
\\
(W_{ab})_{xy}
=
\mathbf1_{L_a}(x)\mathbf1_{L_b}(y),
\qquad
W_{ab}=w_a^T w_b.
\end{gathered}
\end{equation}
For a graph wall,
$(W_u)_{xy}=\delta_{y,u(x)}$, so
$$
W_{\mathrm{id}}=I_4,
\qquad
W_s=
\begin{pmatrix}
1&0&0&0\\
0&0&1&0\\
0&1&0&0\\
0&0&0&1
\end{pmatrix}.
$$
These are the six toric-code wall matrices of Lan, Wang and Wen
\cite[Supplemental Material, Section~II]{LanWangWen2015TC}:
two invertible walls and four products of boundary vectors.
Comparison with an outgoing-row convention requires transposition.
The subgroup calculation above identifies their supports with
the corresponding folded Lagrangian algebras.

All six matrices commute with $S$ and $T$. Their products satisfy
\begin{equation}
\label{eq:tc-transmission-composition}
\begin{gathered}
W_{ab}W_{cd}
=
|L_b\cap L_c|\,W_{ad},
\qquad
W_s^2=I_4,
\\
W_sW_{ab}=W_{s(a),b},
\qquad
W_{ab}W_s=W_{a,s(b)}.
\end{gathered}
\end{equation}
These matrix identities retain intermediate-channel
multiplicities. The full surface composition is the
bimodule-category composition of Section~\ref{sec:defects}.

The electric--magnetic wall admits the lattice lift
$$
U_s=
\begin{pmatrix}
0&1\\
1&0
\end{pmatrix},
\qquad
U_s^T K_{\mathrm{TC}}U_s=K_{\mathrm{TC}},
\qquad
u_{U_s}=s.
$$
Thus both elements of
$\mathcal O(G,q)=\{\mathrm{id},s\}$ lift in this presentation.
Moreover $q=q^{-1}$, so the reverse bulk is equivalent to the
original bulk, as in Section~\ref{sec:symmetries}.

Each Lagrangian subgroup $M\subset V$ admits a linear complement
and hence an additive quotient section. The folded ambient
associator is trivial, so \eqref{eq:gauge-explicit-beta} gives
$\beta_M=1$, including when the folded algebra requires a
nontrivial multiplication cochain. Therefore its wall-line category
$$
\mathcal W_M
=
(\mathcal C_{\mathrm{TC}}\boxtimes
 \mathcal C_{\mathrm{TC}}^{\mathrm{rev}})_{A_M}
$$
satisfies
\begin{equation}
\label{eq:gauge-toric-wall-lines}
\mathcal W_M
\simeq\operatorname{Vec}_{V/M}
\simeq\operatorname{Vec}_{\mathbb Z_2^2}
\qquad\text{for all six walls}.
\end{equation}
The folded algebras and their central bulk functors distinguish
the walls despite these equivalent abstract wall-line categories.

\subsection*{Turaev--Viro, full centers, and Alterfold amplitudes}

The center presentation
$\mathcal C_{\mathrm{TC}}
\simeq\mathcal Z(\operatorname{Vec}_{\mathbb Z_2})$
and Theorem~\ref{thm:center-case} give
\begin{equation}
\label{eq:tc-tv-single}
Z_{\mathrm{TC}}
\simeq Z^{TV}_{\operatorname{Vec}_{\mathbb Z_2}}.
\end{equation}
For a triangulation of a connected closed manifold $M$ with
$n_0$ vertices, the untwisted state sum is
$2^{-n_0}|Z^1(M;\mathbb Z_2)|$, where $Z^1$ is the group of
simplicial cocycles. Since the coboundary group
$B^1(M;\mathbb Z_2)$ has order $2^{n_0-1}$, this state sum equals
\eqref{eq:tc-partition-count}. The Turaev--Viro and finite-gauge
calculations therefore agree with the BF normalization.

Using the underlying fusion category
$\mathcal C_{\mathrm{TC}}$ as the Alterfold input instead gives
the doubled theory of Theorem~\ref{thm:double-case}:
\begin{equation}
\label{eq:tc-alterfold-double}
\begin{gathered}
Z^{TV}_{\mathcal C_{\mathrm{TC}}}
\simeq
Z_{\mathrm{TC}}\otimes\overline{Z_{\mathrm{TC}}},
\\
Z^{TV}_{\mathcal C_{\mathrm{TC}}}(M)
=
\frac{|H^1(M;\mathbb Z_2)|^2}{4},
\qquad
\dim\mathcal H_{\mathrm{dbl}}(\Sigma_g)=16^g.
\end{gathered}
\end{equation}
The partition-function formula assumes that $M$ is closed and
connected. The bar denotes the orientation-reversed,
complex-conjugate theory, and
$$
\mathcal H_{\mathrm{dbl}}(\Sigma_g)
:=
\mathcal H_{\mathrm{TC}}(\Sigma_g)
\otimes
\overline{\mathcal H_{\mathrm{TC}}(\Sigma_g)}.
$$
In particular, the doubled partition functions on $S^3$ and $T^3$
are $1/4$ and $16$. The choice of fusion-category input thus
distinguishes the single toric-code theory from its doubled
Alterfold realization.

For a condensation algebra $A$, let $z^A=(z^A_{jk})$ be the
full-center multiplicity matrix of Section~\ref{sec:Alterfold},
whose entries are the dimensions of the intertwiner spaces
between $\alpha_A^+(X_j)$ and $\alpha_A^-(X_k)$.
For the three condensation algebras,
\begin{equation}
\label{eq:tc-full-centers}
z^{\mathbf1}=I_4,
\qquad
z^{A_e}=w_e^T w_e,
\qquad
z^{A_m}=w_m^T w_m.
\end{equation}
For $A=\mathbf1$, both inductions are the identity.
For $A_a$, compatibility of the two braided attachments requires
charges in $L_a^\perp=L_a$; the intertwiner space is
one-dimensional precisely when $j,k\in L_a$.
Thus
$z^{A_a}_{jk}=\mathbf1_{L_a}(j)\mathbf1_{L_a}(k)$.
The resulting full centers describe the transparent wall and
the product walls $L_e\times L_e$ and $L_m\times L_m$.

Under the torus-state identification, define
$$
|Z_{\mathrm{full}}(A)\rangle
:=
\sum_{j,k\in G}
z^A_{jk}|j\rangle\otimes\overline{|k\rangle}.
$$
The barred vector belongs to the conjugate torus space.
Explicitly,
\begin{equation}
\label{eq:tc-doubled-torus-states}
\begin{aligned}
|Z_{\mathrm{full}}(\mathbf1)\rangle
&=
\sum_{x\in G}|x\rangle\otimes\overline{|x\rangle},
\\
|Z_{\mathrm{full}}(A_a)\rangle
&=
|\mathfrak b_a\rangle
\otimes\overline{|\mathfrak b_a\rangle}.
\end{aligned}
\end{equation}
Every charge is self-dual, so the second-index sign in the general
full-center expansion is invisible in this example.
These expressions are vectors in the $16$-dimensional doubled
torus space, whereas \eqref{eq:tc-partition-count} and
\eqref{eq:tc-alterfold-double} are closed-manifold scalars.

The toric-code example completes the comparison developed in this
paper. Its established bulk and defect data are recovered together
with their state spaces, bordism maps, and gluing laws in a common
normalization. The contribution of the framework is their compatible
realization through the Chern--Simons/Reshetikhin--Turaev equivalence
and the associated defect and Morita constructions. The finite
quadratic module thereby connects explicit gauge-theoretic
calculations with the categorical organization of boundaries,
interfaces, and junctions, while keeping the original bulk distinct
from its doubled realization.

\subsection{The semion theory}
\label{sec:semion-example}

The toric code does not detect the signature-dependent phases of the
construction. We therefore consider the positive even lattice $K_s=(2)$,
whose Chern--Simons theory is the semion theory $U(1)_2$. This example
tests the modular anomaly, the phases of closed-manifold amplitudes,
and the obstruction to a gapped vacuum boundary. We retain the
normalization and surgery convention of Sections~\ref{sec:Abelian-CS}
and~\ref{sec:reciprocity}.

\subsection*{Bulk data and state spaces}

Write $Z_s=Z^{CS}_{\mathbb R/\mathbb Z,(2)}$ and
$\mathcal C_s=\mathcal C(\mathbb Z_2,q_s)$. Then
\begin{equation}
\label{eq:semion-data}
\sigma(K_s)=1,\qquad G_s=\{0,s\},\qquad
q_s(s)=i,\qquad b_s(s,s)=-1,\qquad D_s=\sqrt2.
\end{equation}
For $a,b,c\in\{0,1\}$, a compatible choice of associator and braiding is
$$
\omega_s(a,b,c)=(-1)^{abc},\qquad c_s(a,b)=i^{ab}.
$$
Thus $X_s\otimes X_s\cong\mathbf1$, the nontrivial associator scalar
is $-1$, and the self-exchange scalar is $i$, agreeing with the standard
semion data in \cite[Section~5.3.1]{RowellStongWang2009}.
In particular, the trivial associator used for the toric code cannot
be retained in this example.

Proposition~\ref{prop:toral-hilbert-space} and the pointed fusion rules give
\begin{equation}
\label{eq:semion-spaces}
\mathcal H_s(\Sigma_g)\cong\mathbb C[(\mathbb Z_2)^g],\qquad
\dim\mathcal H_s(\Sigma_g;x_1,\ldots,x_n)
=2^g\delta_{x_1+\cdots+x_n,0},
\end{equation}
where the first identification uses the polarization choices of
Section~\ref{sec:Abelian-CS}, and the marked charges have incoming orientation.
In the solid-torus basis $(|0\rangle,|s\rangle)$,
\begin{equation}
\label{eq:semion-modular}
S=\frac1{\sqrt2}\begin{pmatrix}1&1\\1&-1\end{pmatrix},\qquad
T_0=\begin{pmatrix}1&0\\0&i\end{pmatrix},\qquad
S^2=I_2,\qquad (ST_0)^3=e^{\pi i/4}I_2.
\end{equation}
Here $T_0$ records the ribbon twists. The phase is precisely
$D_s^{-1}\sum_xq_s(x)=e^{\pi i\sigma(K_s)/4}$.
With the framing normalization $T_{\mathrm{fr}}=e^{-\pi i/12}T_0$,
one instead has $(ST_{\mathrm{fr}})^3=I_2$.
The corresponding BKS composition factor is
$e^{\pi i\mu/4}$, where $\mu$ is the Maslov index from
Section~\ref{sec:Abelian-CS}. These formulas retain the signature contribution
that was trivial for the toric code.

\subsection*{Partition functions and Wilson lines}

For surgery on an $r$-component framed link with integral linking
matrix $B$, Theorem~\ref{thm:surgery-expression} and reciprocity give
\begin{equation}
\label{eq:semion-surgery}
Z_s(M_B)=2^{-(r+1)/2}e^{\pi i\sigma(B)/4}
\sum_{a\in\{0,1\}^r}\exp\!\left(-\frac{\pi i}{2}a^TBa\right).
\end{equation}
The inverse quadratic phase and positive signature correction follow
the surgery convention of Section~\ref{sec:reciprocity}; reversing
orientation complex-conjugates the answer. In particular,
\begin{equation}
\label{eq:semion-partitions}
Z_s(S^3)=2^{-1/2},\qquad Z_s(S^1\times S^2)=1,\qquad
Z_s(\Sigma_g\times S^1)=2^g.
\end{equation}
For the one-component matrix $(p)$, with $p>0$, the same convention yields
\begin{equation}
\label{eq:semion-lens}
Z_s(M_{(p)})
=\frac{e^{\pi i/4}}2\bigl(1+(-i)^p\bigr)
=\frac1{\sqrt{2p}}\sum_{a=0}^{p-1}e^{2\pi i a^2/p}.
\end{equation}
Thus $Z_s(\mathbb RP^3)=0$ and $Z_s(M_{(3)})=i/\sqrt2$.
The vanishing at $p=2$ is cancellation between the two charge sectors;
an unsigned count of flat fields would not reproduce it.

For a framed link $\mathscr L\subset S^3$ with labels $a_i s$,
$a_i\in\{0,1\}$, framings $f_i$, and linking numbers $\ell_{ij}$,
the positive ribbon-twist convention gives
\begin{equation}
\label{eq:semion-wilson}
\frac{Z_s(S^3;\mathscr L)}{Z_s(S^3)}
=i^{\sum_i f_i a_i}
(-1)^{\sum_{i<j}\ell_{ij}a_i a_j}.
\end{equation}
A unit twist of $s$ contributes $i$, whereas an $s$--$s$ Hopf link
contributes $-1$. In \eqref{eq:semion-surgery}, the surgery evaluation
uses the inverse of these ribbon phases, as specified above.

For the transparent semion wall, the folded group is
$V=\mathbb Z_2^2$, with
$$
\omega_V((a,b),(c,d),(e,f))=(-1)^{ace+bdf},
\qquad
c_V((a,b),(c,d))=i^{ac-bd}.
$$
On the diagonal subgroup
$\Delta=\{(0,0),(1,1)\}$, both restrictions are trivial, so the
multiplication cochain may be chosen as $\mu_\Delta=1$.
The quotient section
$$
s:\mathbb Z_2\longrightarrow V,\qquad s(a)=(a,0),
$$
is additive. Hence $\lambda=0$, and
\eqref{eq:gauge-explicit-beta} reduces to
$$
\beta_\Delta(a,b,c)
=\omega_V(s(a),s(b),s(c))
=(-1)^{abc}.
$$
Consequently
$$
\mathcal W_\Delta
\simeq\operatorname{Vec}_{\mathbb Z_2}^{\,\beta_\Delta}.
$$
This explicitly recovers the nontrivial associator class of the
underlying semion fusion category from the boundary formula.
\subsection*{Boundaries, walls, and the doubled theory}

Since $q_s(s)\ne1$, Proposition~\ref{prop:defects-etale} leaves only
the unit connected condensation algebra. Moreover, $G_s$ has no
Lagrangian subgroup. Hence $\mathcal C_s$ admits no gapped vacuum
boundary and no single-copy Turaev--Viro center presentation,
in agreement with the Abelian boundary analysis of
\cite{Kapustin2010}.

In the folded group $G_s\oplus G_s$, the quadratic form is
$(a,b)\mapsto i^{a-b}$. Its unique Lagrangian subgroup is
$\Delta=\{(0,0),(1,1)\}$. Thus the transparent wall is the only
elementary fully gapped self-wall in the classification of
Theorem~\ref{thm:defect-classification}. Its line category is
$$
\mathcal W_\Delta\simeq
\operatorname{Vec}_{\mathbb Z_2}^{\omega_s}.
$$
Here $\operatorname{Vec}_{\mathbb Z_2}^{\omega_s}$ denotes graded vector
spaces with associator $\omega_s$; unlike the toric-code wall categories,
its associator class is nontrivial.

The doubled bulk is the double-semion theory, with twists
$1,i,-i,1$. Theorem~\ref{thm:double-case} gives
\begin{equation}
\label{eq:semion-double}
\begin{gathered}
\mathcal Z(\operatorname{Vec}_{\mathbb Z_2}^{\omega_s})
\simeq\mathcal C_s\boxtimes\mathcal C_s^{\mathrm{rev}},\qquad
Z^{TV}_{\operatorname{Vec}_{\mathbb Z_2}^{\omega_s}}
\simeq Z_s\otimes\overline{Z_s},\\
Z^{TV}_{\operatorname{Vec}_{\mathbb Z_2}^{\omega_s}}(M)=|Z_s(M)|^2,
\qquad \dim\mathcal H_{\mathrm{dbl}}(\Sigma_g)=4^g.
\end{gathered}
\end{equation}
The scalar formula assumes $M$ closed and connected. Its values on
$S^3,T^3,\mathbb RP^3$ are respectively $1/2,4,0$.
The full center of the unit defines the doubled torus vector
$|0\rangle\otimes\overline{|0\rangle}
+|s\rangle\otimes\overline{|s\rangle}$.
The semion example therefore checks both the anomalous bulk formulas
and their anomaly-free doubled realization.

\subsection{A non-split quotient and its boundary associator}

To test the screening terms in
\eqref{eq:gauge-explicit-beta}, consider
$$
K=\begin{pmatrix}0&3\\3&2\end{pmatrix},
\qquad \sigma(K)=0,\qquad
G\cong\mathbb Z_9.
$$
Taking the class of $(1,0)$ as generator gives
$$
q(x)=e^{-2\pi i x^2/9},\qquad
\omega=1,\qquad c(x,y)=e^{-2\pi i xy/9}.
$$
The subgroup $L=3\mathbb Z_9$ is Lagrangian, and
$c|_{L\times L}=1$, so we may choose $\mu=1$.
Identify $Q=G/L$ with $\mathbb Z_3$ and choose
$s(a)=a$ for $a\in\{0,1,2\}$. Then
$$
\lambda(a,b)=3\left\lfloor\frac{a+b}{3}\right\rfloor.
$$
Equation~\eqref{eq:gauge-explicit-beta} therefore gives
$$
\beta(a,b,d)
=c(s(d),\lambda(a,b))
=\exp\!\left(
-\frac{2\pi i}{3}\,
d\left\lfloor\frac{a+b}{3}\right\rfloor
\right).
$$
This cocycle has nontrivial cohomology class: the product
$$
\prod_{j=0}^{2}\beta(1,j,1)=e^{-2\pi i/3}
$$
is unchanged by normalized coboundaries and would equal $1$
for a trivial class. Hence the boundary category is
$\operatorname{Vec}_{\mathbb Z_3}^{\,\beta}$ with nontrivial
associator, although the chosen bulk associator is trivial.

%\section*{Declarations}

%\textbf{Funding.} No funding was received to assist with the preparation of this manuscript.

%\textbf{Competing interests.} The author declares that there are no competing interests relevant to the content of this article.

%\textbf{Data availability.} Data sharing is not applicable to this article as no datasets were generated or analysed during the current study.

\bibliographystyle{alpha}
\renewcommand{\refname}{References}
\bibliography{refs}
%\vspace{-0.4cm}
\end{document}